\documentclass[noinfoline]{imsart}

\usepackage[utf8]{inputenc}
\usepackage{amsmath,amssymb,color,amsthm,enumerate,natbib}
\usepackage[usenames,dvipsnames]{xcolor}
\usepackage{float}
\usepackage{verbatim}
\usepackage{graphicx}
\usepackage{dsfont}

\usepackage{xr}
\RequirePackage[colorlinks,citecolor=blue,urlcolor=blue]{hyperref}

\newcommand{\indep}{\rotatebox[origin=c]{90}{$\models$}}

\usepackage{algorithm}
\usepackage{algorithmic}
\usepackage{mathrsfs}

\setattribute{tablecaption}{shape}{}
\setattribute{tablecaption}{size} {\footnotesize\itshape}
\setattribute{tablename}   {size} {\scshape}
\setattribute{tablename}   {skip} {\@: }

\DeclareMathOperator*{\argmin}{arg\,min}

\usepackage{tikz,enumerate}
\usepackage[framemethod=TikZ]{mdframed}
\usepackage{tikz}
\usepackage[framemethod=TikZ]{mdframed}

\usetikzlibrary{arrows,automata,decorations.markings}
\usetikzlibrary{patterns,snakes}
\usetikzlibrary{positioning}
\usetikzlibrary{calc}

\tikzstyle{vertex}=[circle, draw, fill, inner sep=0pt, minimum size=0.15cm]

\usetikzlibrary{arrows}
\newdimen\arrowsize
\pgfarrowsdeclare{arcsq}{arcsq}
{
  \arrowsize=0.2pt
  \advance\arrowsize by .5\pgflinewidth
  \pgfarrowsleftextend{-4\arrowsize-.5\pgflinewidth}
  \pgfarrowsrightextend{.5\pgflinewidth}
}
{
  \arrowsize=1.5pt
  \advance\arrowsize by .5\pgflinewidth
  \pgfsetdash{}{0pt}   \pgfsetroundjoin     \pgfsetroundcap      \pgfpathmoveto{\pgfpoint{0\arrowsize}{0\arrowsize}}
  \pgfpatharc{-90}{-140}{4\arrowsize}
  \pgfusepathqstroke
  \pgfpathmoveto{\pgfpointorigin}
  \pgfpatharc{90}{140}{4\arrowsize}
  \pgfusepathqstroke
}
\newcommand{\SI}{\boldsymbol{\Sigma}}
\newcommand{\pa}{\mathit{pa}}
\newcommand{\XX}{\mathbf{X}}

\newcommand{\XE}{\mathbf{X}^{\e}}

\newcommand{\Xone}{\mathbf{X}^{1}}
\newcommand{\Xtwo}{\mathbf{X}^{2}}
\newcommand{\Yone}{\mathbf{Y}^{1}}
\newcommand{\Ytwo}{\mathbf{Y}^{2}}

\newcommand{\YY}{\mathbf{Y}}
\newcommand{\Z}{\mathbf{Z}}

\newcommand{\Ye}{\mathbf{Y}^{e}}

\newcommand{\G}{\mathbf{G}}
\newcommand{\E}{\mathcal{E}}
\newcommand{\e}{e}

\newcommand{\Var}{\mathrm{Var}}

\newtheorem{lemma}{Lemma}
\newtheorem{theorem}{Theorem}
\newtheorem{proposition}{Proposition}
\newtheorem{assumption}{Assumption}

\newtheorem{definition}{Definition}
\newtheorem{remark}{Remark}

\begin{document}

\begin{frontmatter}
\title{Causal Dantzig: fast inference in  linear structural equation models with hidden variables under additive interventions}

\runtitle{Causal Dantzig}


\begin{aug}
\author{\fnms{Dominik} \snm{Rothenh\"ausler}\ead[label=e2]{rothenhaeusler@stat.math.ethz.ch}}, \author{\fnms{Peter} \snm{B\"uhlmann}} \\ \and
\author{\fnms{Nicolai} \snm{Meinshausen}\ead[label=e1]{meinshausen@stat.math.ethz.ch}} \ead[label=e3]{buehlmann@stat.math.ethz.ch}
\ead[label=u1,url]{http://stat.ethz.ch}
\runauthor{D. Rothenh\"ausler, P. B\"uhlmann and N. Meinshausen}
\affiliation{ETH Z\"urich}
\address{Seminar f\"ur Statistik\\
ETH Z\"urich\\
8092 Z\"urich\\
Switzerland\\
\printead{e2} \\
\printead*{e3} \\
\printead*{e1}}
\end{aug}

\begin{abstract}
Causal inference is known to be
very challenging when only observational data are
available. Randomized experiments are
often costly and impractical and in instrumental variable
regression the number of instruments has to exceed the number of causal predictors.
It was recently shown in
\citet{peters2016causal} that causal inference for the full model is possible when
data from distinct observational environments are available, exploiting that the
conditional distribution of  a response variable is invariant  under the correct causal model.
Two shortcomings of such an approach are the high computational effort for
 large-scale data and the assumed absence of hidden confounders. Here
 we show that these two shortcomings can be addressed if one is
 willing to make a more restrictive assumption on the type of
 interventions that generate different environments. Thereby, we look at  a
 different notion of invariance, namely inner-product invariance.
By avoiding a computationally cumbersome reverse-engineering approach such as in \citet{peters2016causal}, it allows for large-scale causal inference in linear structural equation models.
We discuss identifiability conditions for the causal
parameter and derive asymptotic confidence intervals in the low-dimensional setting. In the case of non-identifiability we show that
the solution set of causal Dantzig has predictive guarantees under certain interventions. We derive finite-sample bounds in the high-dimensional setting and investigate its performance on simulated datasets.
\end{abstract}

\end{frontmatter}

\section{Introduction}
Using only observational data to infer causal relations is a
challenging task and only possible under certain circumstances and
assumptions. In the context of structural equation models
\citep{Bollen1989,robins2000marginal,Pearl2009}, one possibility is to characterize the Markov equivalence
class of graphs under the assumption of acyclicity and usually
faithfulness  \citep{Verma1991,andersson1997,Tian2001, Hauser2012,Chickering2002}. Based on
the Markov equivalence class, some causal effects and often only bounds for them can be
inferred, see for example \citet{Maathuis2009} and \citet{vanderweele2010signed}. Other approaches
exploit non-Gaussianity or nonlinearities, while making suitable
assumptions about the causal model \citep{Shimizu2006,Hoyer2008}.

If both observational and data under interventions are available and
the target and effect of the interventions is perfectly known, the
task of inferring causal relationships becomes easier.
 \citet{hapb14}, for example, modify the greedy
equivalence search of \citet{Chickering2002} to such a
scenario. If an instrumental variable is available, then different
forms of instrumental variable regression
\citep{Wright1928,bowden1990instrumental,angrist1996identification,didelez2010assumptions}
can be used to infer the causal effect of a single variable on a
target of interest.

Consider a setting where data are recorded in different environments. The environments can have an
arbitrary and unknown intervention effect on all predictor variables
and the method exploits that the conditional distribution of the
target $Y$ of interest, given its causal parents, is invariant across
environments under arbitrary interventions on all variables
(excluding, just as in instrumental variable regression, direct
interventions on the response or target $Y$). While it was demonstrated in \citet{peters2016causal} that the method can infer a
full causal model, there are two major shortcomings:
\begin{enumerate}[(i)]
\item It is assumed for invariant causal prediction (ICP) \citep{peters2016causal} that there are no hidden variables that influence
  $Y$ and its parents simultaneously.
\item ICP scans all potential subsets of variables and tests
  whether the conditional distribution of $Y$ given a subset of
  variables is invariant across all environments. This makes the
  method computationally prohibitively expensive as soon as the number
  of predictor variables starts to exceed one or two dozens.
\end{enumerate}
We will show that both shortcomings can be addressed if we are willing
to make a more specific assumption about the type of interventions
that generate the different environments.

\subsection{Setting and notation}\label{settingandnotation}
Assume we have a $p+1$ variables $X_1,\ldots,X_{p+1}$
from a linear Structural Equation Model (SEM) \citep{Bollen1989,robins2000marginal,Pearl2009},
\begin{align}\label{SEM}
  X_{k} &\leftarrow \sum_{k' \neq k} A_{k,k'} X_{k'} + \eta_{k}, \qquad k=1,\ldots,p+1,
\end{align}
where $\pa(k) := \{ k' : A_{k,k'} \neq 0  \} \subseteq \{1,\ldots,p+1\} \setminus k$ is the set of parents
of variable $k$. For notational simplicity we set $A_{k,k} := 0$ for all $k$. Deviating from convention, we allow
dependence between the components of the noise contribution
$\eta=(\eta_1,\ldots,\eta_{p+1})$ which is equivalent to allowing
for hidden variables as parents of the observed variables $X_1,\ldots,X_{p+1}$, see Figure~\ref{fig:H} for an example. The variables form a directed graph $G=(V,E)$, where the nodes
$V=\{1,\ldots,p+1\}$ are given by the variables themselves and there
is an edge from variable $k$ to $k'$ if and only if $k\in \pa(k')$.
Furthermore, we allow the underlying graph to be cyclic. 
 The values $(A_{k,k'})$ for $k,k' \in \{1, \ldots,p+1 \}$ form a $(p+1) \times (p+1)$-dimensional matrix that we denote by $A$. We write $\text{Id}_{p+1}$ for the $(p+1) \times (p+1)$-dimensional identity matrix. To make the distribution of $X_{1},...,X_{p+1}$ well defined in the presence of cycles, we assume that $\text{Id}_{p+1} - A $ is invertible. Note that this is always the case if $G$ is acyclic.

We consider inferring the structural equation for just one of the variables and we take variable $X_{p+1}$ without
loss of generality and denote it by $Y$. Note that $Y$ can be in the parental set of some (or all) of the variables $X_{1},\ldots,X_{p}$, i.e. the matrix $A$ is not necessarily lower triangular. With slight abuse of notation we define $X := (X_{1},\ldots,X_{p})$, $\beta^{0} := A_{p+1,1:p}$ and $\varepsilon := \eta_{p+1}$  such that
\begin{equation} \label{eq:Y}
Y:= X_{p+1} = \sum_{k=1}^p  \beta^{0}_{k} X_{k}  + \varepsilon.
\end{equation}
Note that the vector $\beta^{0}$ has a causal interpretation as it is the coefficient vector $A_{p+1,1:p}$ in the structural equation model~\eqref{eq:Y}. The goal is to infer $\beta^{0}$.

\subsection{Relation to other work}

We have mentioned already major differences to invariant causal prediction \citep{peters2016causal} and the loose relation to the vast literature on instrumental variable regression \citep{didelez2010assumptions} which will be detailed in Section~\ref{sec:comparison-with-instrumental-variables}.
Another method that relies on shift interventions has been published recently \citep{rothenhausler2015backshift}. However, the authors exploit a different type of invariance  as \emph{inner-product invariance} does not hold in this setting. \citet{lewbel2012using} uses heteroscedasticity to infer structural equations. While \citet{lewbel2012using} uses cross-products between exogeneous variables and error terms to identify structural equations, we directly exploit the covariance structure of endogeneous variables and the error terms, resulting in a different method. The comparison in Figure~\ref{fig:kemmeren} about an application has been published in  \citet{meinshausen2016methods}. The concept of \emph{inner-product invariance}, the \emph{causal Dantzig} method and all its corresponding theory  are entirely novel.

\subsection{Overview}

In Section~\ref{sec:cond-inner-prod} we introduce the notion of \emph{inner-product invariance} and discuss under which assumptions this property is satisfied. In Section~\ref{sec:causal-dantzig} we leverage this property to define the unregularized \emph{causal Dantzig} and discuss identifiability, low-dimensional estimation and inference. Furthermore, in the case of non-identifiability we show that the solution set of \emph{causal Dantzig} has predictive guarantees under certain interventions. We conclude with a comparison to instrumental variable regression and a discussion of \emph{inner-product invariance} from the perspective of potential outcomes. In Section~\ref{sec:high-dim} we introduce the regularized \emph{causal Dantzig}, examine its performance in high-dimensional estimation and show how it can achieve consistency under relaxed identifiability assumptions. Practical considerations  for both the regularized and unregularized \emph{causal Dantzig} can be found in Section~\ref{sec:pract-cons}. Numerical examples can be found in Section~\ref{sec:numerical-examples}.

\section{Conditional and inner-product invariance}\label{sec:cond-inner-prod}

In analogy to the setting of \citet{peters2016causal} we assume that
the data are recorded under different discrete environments or experimental conditions $e \in \E$.
The random variable $X$ in environment $\e\in\E$ is denoted by $X^\e$ and the distribution of $\eta$ by $\eta^\e$.
We observe i.i.d. samples of $(X^{e}, Y^{e})$ from each environment $e \in \E$ and for each sample $i$ we observe from which environment $e_{i} \in \E$ it was drawn. This variable $e_{i}$ can be deterministic or random.\\
The distribution
of a variable can be different across environments due
to specific or non-specific interventions.
   A change in the
distribution of $X^\e,\eta^\e$ can be caused by different intervention
mechanisms such as do-interventions or noise-interventions, which can
be randomized or not and known or partially known or unknown.

The type of intervention that generates the environments is arbitrary
in \citet{peters2016causal} with the exception that interventions on
the target $Y$ itself are not allowed. The same requirement is also
necessary for the instrumental
variable approach and we will keep this requirement in the following.
For possible relaxations see~\citet{rothenhausler2015backshift}. Throughout the paper we assume
that the distributions $(X^{e},Y^{e})$ are non-degenerate and that the Gram
matrix of $(X^{e},Y^{e})$ is well-defined and positive definite for all $e \in \E$.

\subsection{Conditional invariance}
The conditional distribution of the target variable $Y$, given  its
parents $\pa(Y)=\pa(X_{p+1})$ is denoted by
\[
Y^\e | X^\e_{pa(Y)}=x.
\]
It was assumed in \citet{peters2016causal} that the conditional
distribution is  invariant for all $x\in \mathbb{R}^{|\pa(Y)|}$  where it is defined in the
absence of hidden confounding (where absence of hidden confounding is
fulfilled in \eqref{SEM} if  all components of
$\eta$ are independent). It then holds for all environments $e,f\in
\E$
and all $x\in \mathbb{R}^{|\pa(Y)|}$ for which the conditional
distributions are well defined that
\begin{equation}\label{eq:condinvariance}
Y^\e | X^\e_{pa(Y)}=x \qquad \stackrel{d}{=} \qquad Y^f | X^f_{pa(Y)}=x.
\end{equation}
 This conditional invariance under the true parental set $\pa(Y)$ is then exploited for inference by testing  for all subsets of
$\{1,\ldots,p\}$ whether the invariance of \eqref{eq:condinvariance} can be rejected. The
intersection of all subsets for which invariance cannot be rejected is
then automatically a subset of the true parental set with controllable probability.

There are two shortcomings of this invariance approach \citep{peters2016causal}  in certain contexts:
\begin{enumerate}[(i)]
\item
The invariance \eqref{eq:condinvariance} becomes invalid under hidden confounding
between $Y$ and the parents of $Y$
as the conditional invariance of \eqref{eq:condinvariance} can be
violated even for the true parental set \citep{peters2016causal}.
\item Testing each subset of $\{1,\ldots,p\}$ restricts the
number of variables to somewhere between $p\le 20$ in practice.
\end{enumerate}
Both of these shortcomings can be addressed when using a different
type of invariance.

\subsection{Inner-product invariance}

We show in the following that the invariance of the \emph{conditional
  distribution} \eqref{eq:condinvariance} can be replaced with an
\emph{inner-product invariance} under a more specific assumption on
the mechanism that generates the different environments.

\begin{definition}\label{def:inner-prod-invar}
\emph{Inner-product invariance} under $\beta^{0} \in \mathbb{R}^p$ is
fulfilled iff
\begin{equation*}
 \mathbb{E}\big[ X^\e_k (Y^\e- X^\e \beta^{0})\big] \quad = \quad \mathbb{E}\big[ X^f_k
 (Y^f- X^f \beta^{0})\big]
\end{equation*}
for all $ \e,f\in\E$  and $ k\in \{1,\ldots,p\}$.
\end{definition}

We will show that \emph{inner-product invariance} is true for the
causal vector $\beta^{0}$ under the assumption of additive interventions
made precise in the following. A derivation of this result from potential outcome assumptions is discussed in Section~\ref{sec:pot-out}. The concept of \emph{inner-product invariance} will then later be exploited for computationally fast causal inference for both  low- and high-dimensional data.

\subsection{Additive interventions}
We assume here that the structural equations \eqref{SEM} are
constant across all environments and that the change in the distribution of $X^\e$ between environments is caused by a shift in the
distribution of $\eta^\e$ between different environments.
\begin{assumption}\label{assum:additive}
Assume that the distributions of $(X_{1}^{e},...,X_{p+1}^{e})$, $e \in \E$, are generated by the linear SEM
\begin{align*}
  X_{k}^{e} &\leftarrow \sum_{k' \neq k} A_{k,k'} X_{k'}^{e} + \eta_{k}^{e}, \qquad \text{ for } k=1,\ldots,p+1 \text{ and } e \in \E.
\end{align*}
Assume that there exist random variables $\eta^0,
\delta^e\in \mathbb{R}^p$ with $\mathrm{Cov}(\eta^{0}, \delta^{e}) = 0$ for all $\e\in \E$ such that $\eta^\e$ can be written as
\begin{equation*}
\eta^\e  \stackrel{d}{=} \eta^0  + \delta^\e \qquad\mbox{for all  } \e\in \E.
\end{equation*}
We assume that $\delta^\e_{p+1}\equiv 0$  for all $\e\in\E$ and $\mathbb{E}[\eta^{0}] = 0$.
\end{assumption}
Note that the components of  $\eta^0$ and of each vector $\delta^e$,
$\e\in\E$ are allowed to  be dependent to allow for hidden confounding.
We call the random variables $\delta^{e}$, $e \in \E$, \emph{additive interventions} as they are additive and specific to the environment $e \in \mathcal{E}$.
$\delta^{e}$ can for example be an additive contribution if $\mathbb{E}(\delta^\e_k)\neq 0$ for some variable $k\in\{1,\ldots,p\}$ or a noise contribution if $\Var(\delta^\e_k)\neq 0$ or both. If $\delta_{k}^{e} \equiv 0$ for some $e \in \E$ and $k \in \{1,\ldots,p\}$ we say that there is \emph{no intervention on variable $k$ in environment $e \in \E$}. The last part of the assumption ensures that the noise part $\delta^e$ that is specific to environment $\e\in\E$ does not include an intervention on the target variable $Y$ itself and is a type of exclusion restriction \citep{Pearl2009}. Mathematically, the crucial property of Assumption~\ref{assum:additive} is that the covariance between the error of covariates and target variable is constant, i.e. that $\text{Cov}(\eta_{1:p}^{e}, \eta_{p+1}^{e})$ is constant across environments $e \in \E$. This allows us to obtain the following result.

\begin{proposition}\label{proposition:const}
Under Assumption~\ref{assum:additive}, we have inner-product invariance
under the true causal coefficients $\beta^{0} = (A_{p+1,k})_{k=1,\ldots,p}$:
\[ \mathbb{E}\big[ X^\e_k (Y^\e- X^\e\beta^{0})\big] \quad = \quad \mathbb{E}\big[ X^f_k
 (Y^f- X^f\beta^{0})\big]  \]
for all $ \e,f\in\E$  and $ k\in \{1,\ldots,p\}$.
\end{proposition}
 The proof of this result can be found in the Appendix. A derivation of this result from potential outcome assumptions is discussed in Section~\ref{sec:pot-out}. We will exploit inner-product invariance to infer the causal effects in linear SEMs in the following.
\begin{figure}
\begin{center}
\begin{tikzpicture}[scale=1.2, line width=0.5pt, minimum size=0.58cm, inner sep=0.3mm, shorten >=1pt, shorten <=1pt]
    \normalsize
    \draw (-1,1.6) node(h) [rectangle, rounded corners=2mm, inner
     sep=1.7mm, draw, dashed] {$H$};
    \draw (-0.8,0) node(y) [rectangle, rounded corners=2mm, inner
     sep=1.7mm, draw] {$\phantom{a} Y\phantom{a}$};
    \draw (1.0,-1.2) node(1) [rectangle, rounded corners=2mm, inner
     sep=1.7mm, draw] {$X_1 $};
    \draw (-2.8,-0.2) node(2) [rectangle, rounded corners=2mm, inner
     sep=1.7mm, draw] {$X_2$};
\draw (1,0.3) node(3) [rectangle, rounded corners=2mm, inner
     sep=1.7mm, draw] {$X_3$};
    \draw[-arcsq] (h) -- (3);
    \draw[-arcsq] (h) -- (y);
    \draw[-arcsq] (h) -- (2);
    \draw[-arcsq] (2) -- (y);
    \draw[-arcsq] (y) -- (1);
    \draw[-arcsq] (2) -- (1);
    \draw[-arcsq] (1) -- (3);
   \end{tikzpicture}
\caption{ \label{fig:H} {\it An example for $p=3$. The hidden variable $H$ is
    supressed notationally and instead the noise contributions at each
  variable $\{X_1,X_2,X_3,X_4=Y\}$ are not assumed to be independent.}}
\end{center}
\end{figure}
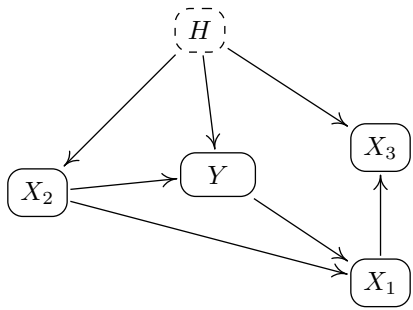

\subsection{Errors-in-variables}

In many real-world applications, we cannot directly observe $X_{1},\ldots,X_{p},Y$, but make a measurement error $\zeta$ when observing it. In other words, we measure
\begin{equation}\label{eq:21}
   \tilde Y^{e} = Y^{e} + \zeta_{y}^{e} \text{ and }  \tilde X_{k}^{e} = X_{k}^{e} + \zeta_{k}^{e}, \qquad e \in \E, k=1,\ldots,p,
\end{equation}
where $\zeta_{y}^{e}, \zeta_{k}^{e}, e \in \E, k=1,\ldots,p$ are centered, jointly independent and independent of $X^{e},Y^{e}, e \in \E$ with finite variance. Furthermore, we make the assumption that the distributions of $ \zeta_{k}^{e}, k = 1,\ldots,p$ are invariant for different settings $e \in \E$. Note that we do not assume that the distribution of $\zeta_{y}^{e}$ is invariant for different settings $e \in \E$. Errors-in-variables exhibit an effect called ``regression dilution'' or ``attenuation''. As an example consider a Structural Equation Model of the following form:
\begin{align*}
\begin{split}
 \text{ latent variables $X_{1}$ and  $Y$ with } Y &= 2 X_{1} + \varepsilon, \\  \text{ observed variables $\tilde X_{1}$ and $\tilde Y$ with }  \tilde X_{1} &= X_{1}+\zeta_{1}, \\ \text{ and } \tilde Y &= Y + \zeta_{y}.
\end{split}
\end{align*}
For now, let us assume that there is no confounding between $X_{1}$ and $Y$ and $X_{1}$. When regressing $\tilde Y$ on $\tilde X_{1}$ we obtain a smaller regression coefficient than when regressing $Y$ on $X_{1}$ due to higher variance of $\tilde X_{1}$. The smaller regression coefficient is by definition the best linear prediction of $\tilde Y$ given $\tilde X_{1}$. In this sense attenuation can be ignored if one wants to make predictions based on $\tilde X_{1}$. However, in causal inference we are interested in knowing what happens when \emph{intervening on $X_{1}$}, and this effect would be underestimated by the regressing $\tilde Y$ on $\tilde X_{1}$. The following proposition shows that if inner-product invariance holds for $X_{1},\ldots,X_{p},Y$ then it also holds for proxy variables $\tilde X_{1},\ldots,\tilde X_{p},\tilde Y$.

\begin{proposition}\label{proposition:errors-variables}
Assume inner-product invariance holds for $X_{1}^{e},\ldots,X_{p}^{e},Y^{e}$, $e \in \E$, under $\beta^{0}$. Assume we have an errors-in-variables model as defined in equation \eqref{eq:21}. Then  inner-product invariance holds for $\tilde X_{1}^{e},\ldots, \tilde X_{p}^{e},\tilde Y^{e}$, $ e \in \E$ under $\beta^{0}$:
\[ \mathbb{E}\big[ \tilde X^\e_k ( \tilde Y^\e- \tilde X^\e\beta^{0})\big] \quad = \quad \mathbb{E}\big[ \tilde X^f_k
 ( \tilde Y^f- \tilde X^f\beta^{0})\big]  \]
for all $ \e,f\in\E$  and  $k\in \{1,\ldots,p\}$.
\end{proposition}
The proof of this result can be found in the Appendix. As a result, methods based on inner-product invariance will be robust with respect to errors-in-variables. Note that the analogous statement is true for instrumental variable regression. Now let us turn to the definition of the unregularized \emph{causal Dantzig}.

\section{Causal Dantzig without regularization}\label{sec:causal-dantzig}

In this section we introduce the unregularized \emph{causal Dantzig}, discuss its basic properties and an example.  We introduce the  unregularized \emph{causal Dantzig} in Section~\ref{sec:estimator-1}. Asymptotic confidence intervals for low-dimensional estimation are discussed in Section~\ref{sec:confidence-intervals}. Section~\ref{sec:impl-example} provides an example and explains basic usage of the method \texttt{causalDantzig} in the R-package \texttt{InvariantCausalPrediction} \citep{R}. Identifiability and consistency issues are discussed  in  Section~\ref{sec:ident-beta0-cons}. We conclude with a comparison to instrumental variable regression in Section~\ref{sec:comparison-with-instrumental-variables}.
\subsection{The estimator}\label{sec:estimator-1}
Assume that we observe i.i.d. samples of $(X^{e},Y^{e})$ in two environments $e \in  \E = \{1,2\}$ 
 with $n_1,n_2$ samples in each environment. Let $\Xone$ and $\Xtwo$ be the $n_1 \times p$ and $n_2 \times p$-dimensional matrices that contain the realized values of the random variables $X^{e}$ in environment $e=1$ and $e=2$ respectively and let $\Yone\in \mathbb{R}^{n_1}$ and $\Ytwo\in \mathbb{R}^{n_2}$ be the respective measurements of the response variables. Define the differences between the two environments in inner-product and Gram matrices, the so-called \emph{Gram-shift matrices}
\begin{align}\label{eq:3}
\begin{split}
\hat \Z &:= \frac{1}{n_{1}}(\Xone)^t \Yone  - \frac{1}{n_{2}}(\Xtwo)^t \Ytwo  \in \mathbb{R}^p\\
\hat \G &:= \frac{1}{n_{1}}(\Xone)^t \Xone  - \frac{1}{n_{2}}(\Xtwo)^t \Xtwo  \in \mathbb{R}^{p\times p}.
\end{split}
\end{align}
Assuming inner-product invariance holds under $\beta^{0}$,
\[ \mathbb{E}[ \hat \Z - \hat \G \beta^{0}] =0 .\]
A simple estimator of $\beta^{0}$ is the empirical minimizer of the
$\ell_{\infty}$-norm of the differences between  $\hat \Z$ and $\hat \G\beta$.
\begin{definition}[Unregularized causal Dantzig]\label{def:dotproduct} The causal Dantzig estimator $\hat \beta$ is defined as a solution to the optimization problem
\begin{equation}\label{eq:hatbeta} \min_{\beta\in \mathbb{R}^p} \| \hat \Z - \hat \G
\beta\|_{\infty} .\end{equation}
\end{definition}
The choice of how to center and scale variables deserves some attention. We will discuss this in Section~\ref{sec:scaling-variables}.
Causal Dantzig is uniquely defined if and only if $\hat \G$ is invertible and can in this
case be written as
\begin{equation}\label{eq:hatbetaInv} \hat{\beta} =  \hat \G^{-1}  \hat \Z .\end{equation}
Note that by equation \eqref{eq:3} this estimator is closely related to least squares in linear regression. Recall that for observations $\YY \in \mathbb{R}^{n}$ and design matrix $\XX \in \mathbb{R}^{n \times p}$, the least squares estimator is defined as
\begin{equation*}
  \hat \beta_{LS} = \left( \XX^{t} \XX \right)^{-1} \XX^{t} \YY.
\end{equation*}
Causal Dantzig is strikingly similar, with the Gram matrices replaced by
differences of Gram matrices in different settings. As such, it is
straightforward to derive asymptotic confidence intervals for this estimator. Many properties from linear regression do not carry over. For example, the causal Dantzig is only asymptotically unbiased.

\subsection{More than two environments}\label{sec:more-than-two}

There are two straightforward extensions to more than two environments $| \E | > 2$. Pooling data from different environments preserves inner-product invariance. If some of the environments are ``observational'' and  the others  are ``interventional'', one option for splitting the data into two environments $\E' = \{ 1,2\}$ is pooling all observational data ($e'=1$) and pooling all interventional data ($e'=2$). Instead of splitting the data into two environments one can change the definition of the estimator to accommodate for more than two settings, for example by defining $\hat \beta $ as a solution to the optimization problem
\begin{equation}\label{eq:23}
 \min_{\beta\in \mathbb{R}^p} \max_{e \in \E} \| \hat \Z^{e} - \hat \G^{e}
\beta\|_{\infty},\end{equation}
where
\begin{align}\label{eq:24}
\begin{split}
\hat \Z^{e} &:= \frac{1}{n_{e}}(\XE)^t \Ye  -  \frac{1}{|\E|-1} \sum_{\tilde e \neq e} \frac{1}{n_{\tilde e}} (\XX^{\tilde e})^t \YY^{\tilde e}  \in \mathbb{R}^p,\\
\hat \G^{e} &:= \frac{1}{n_{e}}(\XE)^t \XE  - \frac{1}{|\E|-1} \sum_{\tilde e \neq e} \frac{1}{n_{\tilde e}} (\XX^{\tilde e})^t \XX^{\tilde e}  \in \mathbb{R}^{p\times p}.
\end{split}
\end{align}
Note that for two environments, solutions of equation~\eqref{eq:23} coincide with equation~\eqref{eq:hatbeta}. It depends on the type of interventions and the signal strength which of the two options mentioned above is better. If the data can be split into two environments $\E'= \{1,2\}$ that are homogeneous, doing so is preferable as the estimators of $\G^{e'}$ and $\Z^{e'}, e' \in \E'$ have low variance. If the environments $\E$ have different (strong) interventions, solving equation~\eqref{eq:23} can be preferable as the effect of several strong interventions might get ``washed out'' when averaging over many environments. We will return later to the case of more than two
environments. For the following discussion we assume that there are two environments $\E = \{1,2\}$.

\subsection{Confidence intervals}\label{sec:confidence-intervals}

In the settings described above $\hat \beta$ is in general only asymptotically unbiased. This bias is unknown as it depends on the unknown amount of confounding between $X^{e}$ and $Y^{e}$. Hence we will only pursue asymptotic confidence intervals.  We will show that the estimator~\eqref{eq:hatbetaInv} is under certain
conditions asymptotically
normally distributed, that is for $n_{1},n_{2} \rightarrow \infty$,
\begin{equation}\label{eq:asymNorm}  \left( \frac{V^{1}}{n_{1}} + \frac{V^{2}}{n_{2}} \right)^{-\frac{1}{2}} \left(\hat{\beta}-\beta^{0} \right) \rightharpoonup \mathcal{N}_{p}(0,\mathrm{Id}_{p}). \end{equation}
The matrices $V^{1}$ and $ V^{2}$ are positive definite under suitable assumptions and can be consistently estimated from the data as $\hat V^{1}$ and $ \hat V^{2}$ as we will discuss later.
We can then define asymptotically valid
confidence intervals for $\beta^{0}_k$ as
\begin{equation}\label{eq:confIV} I_k = \left[\hat{\beta}_k - q \sqrt{\hat{V}_{kk}}, \hat{\beta}_k + q \sqrt{ \hat{V}_{kk}} \right] ,\end{equation}
where $\hat{V}_{kk}$ is the $k$-th diagonal element of $\hat{V} = \hat V^{1}/n_{1} + \hat V^{2} / n_{2}$ and
$q=\Phi^{-1}(1-\alpha/2)$. Here, $\Phi$ denotes the distribution function of a standard Gaussian random variable. The interval $I_k$ has asymptotic
coverage
\[ \mathbb{P}[\beta^{0}_k \in I_k] \rightarrow 1-\alpha \qquad \mbox{ for  }
n_{1},n_{2}\rightarrow \infty.\]
The conditions for asymptotic normality~\eqref{eq:asymNorm} are
fourth-moment conditions on the observed random variables as well as
conditions that guarantee that $V^{1}$ and $V^{2}$ are invertible and that \emph{causal Dantzig} is unique.

\begin{theorem}[Asymptotic normality]\label{theorem:confidence-intervals}  Let $(X^{1},Y^{1})$ and $(X^{2},Y^{2})$ have finite fourth moments and assume that inner product invariance holds under $\beta^{0}$.
Assume that  $(\Xone,\Yone) $ and $(\Xtwo,\Ytwo) $ are independent. Define $\G := \mathbb{E} [ \hat \G ]$ and $\Z := \mathbb{E} [ \hat \Z]$ and let  $\G$ and the covariance matrix of $X^{e} \eta_{p+1}^{e}$, $e \in \E$ be invertible. For $n_{1},n_{2}  \rightarrow \infty$,
\begin{equation*}
  \left( \frac{V^{1}}{n_{1}} + \frac{V^{2}}{n_{2}} \right)^{-\frac{1}{2}} \left( \hat \beta - \beta^{0} \right) \rightharpoonup \mathcal{N} \left(0,\mathrm{Id}_{p} \right),
\end{equation*}
where $V^{e} :=  \mathrm{Cov}(\G^{-1}   (X^{e})^{t} \eta_{p+1}^{e})$, $e \in \{1,2\}$ are invertible.
Note that we allow $n_{1}$ and $n_{2}$ to have different asymptotic growth rates.
\end{theorem}

\begin{remark}[Estimation of $V^{1}$ and $V^{2}$]
The empirical covariance matrix of
\begin{equation*}
  - \hat \G^{-1} \left( \XX_{i\cdot}^{1} \right)^t \XX_{i\cdot}^{1}  \hat \G^{-1} \hat \Z+
  \hat \G^{-1} \left( \XX_{i\cdot}^{1} \right)^t \YY_i^{1} , \mbox{ } i=1,...,n_{1},
\end{equation*}
is a consistent estimator of $V^{1}$. $V^{2}$ can be estimated analogously.
\end{remark}

The proof of this result can be found in the Appendix.  The assumption that $\G$ is invertible will be discussed further in Section~\ref{sec:ident-beta0-cons}. In Section~\ref{sec:high-dim} we will discuss how the regularized \emph{causal Dantzig} can be consistent in some situations where population $\G$ is not invertible. Asymptotic efficiency  is discussed in Section~\ref{sec:asympt-effic} in the Appendix. 

\subsection{Implementation and example}\label{sec:impl-example}

\begin{figure}
\begin{center}
\includegraphics[width=0.4\textwidth]{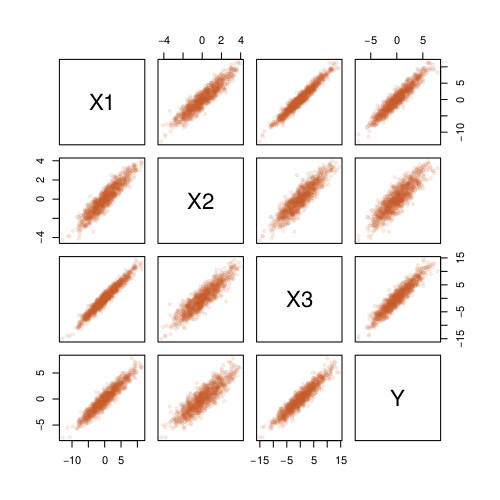}
\includegraphics[width=0.4\textwidth]{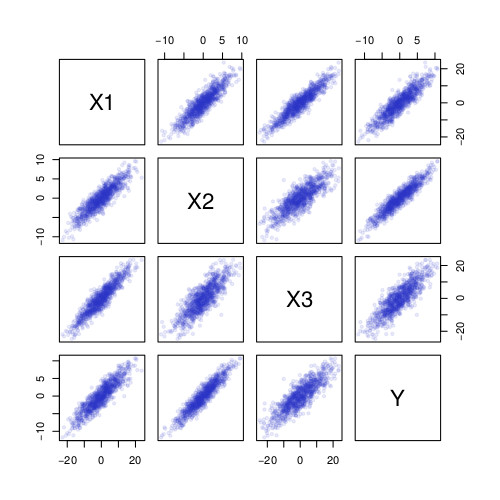}
\caption{ \label{fig:scatter} {\it The scatterplot of the variables in
  the graph of Figure~\ref{fig:H} and SEM~\eqref{eq:SEMex} for
  environment~$1$ (red, left panel) and environment~$2$ (blue,
  right panel). The estimate is based on the difference in the two
  Gram matrices.}}
\end{center}
\end{figure}

We use data generated according to a SEM with the structure given by Figure~\ref{fig:H} as an example.
Suppose the data are generated in two environments $\{1,2\}= \E$
according to
\begin{equation}\label{eq:SEMex}
 \left\{\begin{array}{rcrrrr} X_2^{e} & \leftarrow & &&\eta^{0} + &\sigma^{e} \eta_2 \\
Y^{e} &\leftarrow&  &X_2^{e} + &\eta^{0} + &\eta_y \\
 X_1^{e} & \leftarrow& Y^{e} +& X_2^{e} + &&\sigma^{e}\eta_1   \\
X_3^{e} & \leftarrow & &X_1^{e} + &\eta^{0} + &\sigma^{e}\eta_3
\end{array} \right. ,
\end{equation}
where $(\eta^{0},\eta_y,\eta_1,\eta_2,\eta_3)$ is assumed to be drawn
 from $\mathcal{N}_{5}(0,\mathrm{Id}_5)$ and the noise variances are
$\sigma^{e}=1$ for environment $e=1$ and $\sigma^{e}=4$ for environment
$e=2$. We draw $1000$ i.i.d. samples from each environment and the
corresponding pairwise scatterplots are shown in
Figure~\ref{fig:scatter}.
For one realization we obtain the estimate $\hat{\beta}$ via the
difference in Gram matrices $\G$ and inner products with the target
$\Z$ as
\begin{equation*}
\hat \G = \left( \begin{array}{ccc}
15.9 & 6.5 & 16.1 \\
6.5 & 3.2  &  6.5 \\
 16.1 & 6.5 & 19.1 \\
\end{array} \right), \hat \Z = \left( \begin{array}{c} 6.4 \\   3.2  \\ 6.5 \end{array} \right) \quad \Rightarrow \hat{\beta} = \hat \G^{-1} \hat \Z = \left( \begin{array}{c} -0.04 \\   1.00 \\ 0.03 \end{array} \right),
\end{equation*}
where the correct vector of causal coefficients in this problem is
\begin{equation*}
\beta^{0}= \left( \begin{array}{c} 0 \\   1  \\
 0 \end{array} \right).
\end{equation*}
Asymptotic confidence intervals can be computed via~\eqref{eq:confIV}.

The
procedure is implemented as method \texttt{causalDantzig} in the R-package
\texttt{InvariantCausalPrediction} \citep{R}. The output for the
example above is shown below, where $X$ is the matrix with predictor
variables, $Y$ the outcome of interest and $E$ is an $n$-dimensional
vector with entries $1$ for samples from environment $e=1$ and entries
$2$ for samples from environment $e=2$.

\begin{center}
\begin{small}
\begin{verbatim}

> fit <- causalDantzig(X,Y,E,regularization=FALSE)
> print(fit)
Unregularized causal Dantzig
Call:
causalDantzig(X = X, Y = Y, E = E, regularization = FALSE)

   Estimate StdErr p.value
X1   -0.042  0.059   0.481
X2    0.999  0.106  <2e-16 ***
X3    0.035  0.042   0.403
---
Signif. codes:  0 ‘***’ 0.001 ‘**’ 0.01 ‘*’ 0.05 ‘.’ 0.1 ‘ ’ 1

\end{verbatim}
\end{small}
\end{center}

Only the direct causal effect of the second variable turns out to be statistically significant.
Note that in this setting, instrumental variables regression would fail. One problem is that the number of covariates exceeds the number of ``instruments''. Additionally, the expectation of $X^{1}$ and $X^{2}$ are equal, implying that there is no mean shift due to the two environments. We will discuss these issues in more detail in Section~\ref{sec:comparison-with-instrumental-variables}.

\subsection{Identifiability of $\beta^{0}$ and practical implications}\label{sec:ident-beta0-cons}

In the simplest setting, the number of samples greatly exceeds the
number of parameters, and the interventions $\delta^{e}$, $e \in  \E$  are sufficiently different to make the parameter $\beta^{0}$ identifiable. Theorem~\ref{theorem:popul-g-invert} gives conditions under which this is the case.
\begin{theorem}\label{theorem:popul-g-invert}
Consider a SEM that satisfies Assumption~\ref{assum:additive}. Assume that there exists an ``observational'' environment, i.e. an environment $e \in \E$ with $\delta^{e} \equiv 0$. Furthermore assume that all interventions $\delta^{e}$ are full-rank on its support, i.e.\ that the Gram matrix of $\delta_{S^{e}}^{e}$ is positive definite for $S^{e} = \{ k :  \delta_{k}^{e} \not \equiv 0\}$.
\begin{enumerate}
  \item The causal coefficient is identifiable in the population case if and only if for each $k=1,\ldots,p$ there exists $e \in \E$ such that $\delta_{k}^{e} \not \equiv 0$.
  \item If the  condition in 1. holds then the solution of \emph{causal Dantzig} as defined in equation~\eqref{eq:23} is unique in the population case and equal to $\beta^{0}$.
\end{enumerate}
\end{theorem}
The proof of this result can be found in the Appendix. Usually, there are many different SEMs satisfying Assumption~\ref{assum:additive} that can generate a given observed distribution of $(X^{e},Y^{e}),e \in \E$. Theorem~\ref{theorem:popul-g-invert} gives a condition under which these SEMs all share the same direct causal effect $\beta^{0}$ from $X_{1},\ldots,X_{p}$ to $Y$. If said condition is satisfied, the \emph{causal Dantzig} has a unique solution in the population case that is equal to $\beta^{0}$. Furthermore, it tells us that if this condition is not satisfied, there exist at least two SEMs satisfying Assumption~\ref{assum:additive} with different direct causal effects from $X_{1},\ldots,X_{p}$ to $Y$ that generate the given distribution. Without further assumptions it is then not possible to consistently estimate the direct causal effects, but only a set of potential causal effects. We will characterize this set later.

Note that  Theorem~\ref{theorem:popul-g-invert} describes a rather strong condition for identifiability. Especially if $p$ is large it might be unrealistic to have nonzero interventions $\delta_{k}^{e}$ on each of the variables $X_{k}, k=1,\ldots,p$. However, making additional assumptions can help resolve these identifiability issues.
If the interventions $\delta_{k}^{e}$ only act on a subset of the variables $X_{1},\ldots,X_{p}$ or when the number of covariates exceeds the sample size $p > n$, the  regularized \emph{causal Dantzig} can be consistent under the additional assumption of sparsity. We discuss consistency of the regularized \emph{causal Dantzig} in such scenarios in Section~\ref{sec:finite-sample-bound} and Section~\ref{sec:behaviourccif}.
Alternatively, it can be advisable to first run LASSO on the pooled dataset to select a subset of the variables. Under the assumption of faithfulness, it is sufficient to have nonzero interventions on the selected subset. Some justification for this approach can be found in Section~\ref{sec:mark-blank-estim-1}.

If the assumptions for identifiability of $\beta^{0}$ are not fulfilled it should still be possible to guarantee predictive performance under certain new  environments. The following theorem makes this intuition more precise.  The proof can be found in the Appendix.
\begin{theorem}\label{thm:partiali}
Consider a SEM that satisfies Assumption~\ref{assum:additive}. Assume that there exists an ``observational'' environment, i.e. an environment $e \in \E$ with $\delta^{e} \equiv 0$. Furthermore assume that all interventions $\delta^{e}$ are full-rank on its support, i.e.\ that the Gram matrix of $\delta_{S^{e}}^{e}$ is positive definite for $S^{e} = \{ k :  \delta_{k}^{e} \not \equiv 0\}$. Let $\beta$ be a solution of causal Dantzig as defined in equation~\eqref{eq:23} in the population case.
\begin{enumerate}
  \item Then the distribution of the residuals is invariant, i.e.
\begin{equation*}
  Y^{e} - X^{e} \beta  \stackrel{d}{=} Y^{f} - X^{f} \beta \text{ for all } e,f \in \E.
\end{equation*}
\item For a new  environment $\tilde e \not \in \E$ that satisfies Assumption~\ref{assum:additive} for $(X^{e},Y^{e})$, $e \in \E \cup \{ \tilde e \}$ with $\{k : \delta_{k}^{\tilde e} \not \equiv 0 \} \subset \cup_{e \in \E} S^{e}$, we have
\begin{equation*}
  Y^{e} - X^{e} \beta  \stackrel{d}{=} Y^{\tilde e} - X^{\tilde e} \beta \text{ for all } e \in \E.
\end{equation*}

\end{enumerate}

\end{theorem}
In words, solutions of \emph{causal Dantzig} guarantee that the residuals have the same distribution across all environments $e \in \E$. Perhaps more importantly, solutions of  \emph{causal Dantzig} are guaranteed to have the same predictive performance on new environments $\tilde e \not \in \E$ with arbitrary large additive perturbations $\delta_{k}^{\tilde e}$ as long as these perturbations act on a subset of the variables $\cup_{e \in \E} S^{e}$.

\subsection{Comparison with instrumental variables}\label{sec:comparison-with-instrumental-variables}

Consider a setting where the underlying DAG takes the following form:
\begin{center}
\begin{tikzpicture}[->,>=latex,shorten >=1pt,auto,node distance=1.2cm,
                    thick]
  \tikzstyle{every state}=[draw=black,text=black, inner sep=0.4pt, minimum size=17pt]

  \node[state] (Y) {$Y$};
  \node[state] (H) [above left of=Y] {$H$};
  \node[state] (X) [below left of=H] {$X$};
  \node[state] (E) [left of=X] {$e$};

\draw  (X)  edge  (Y);
\draw  (H)   edge  (Y);
\draw  (H)  edge (X);
\draw  (E)  edge (X);

\end{tikzpicture}
\end{center}
We assume that $H$ is not observed and that $e$ takes values  in $ \{1,2\}$. To be able to use the \emph{causal Dantzig}, we have to define settings $\E$. It is rather straightforwards to write $(X^{1},Y^{1})$ for the variables $(X,Y)$ conditioned on $e=1$ and $(X^{2},Y^{2})$ for the variables $(X,Y)$ conditioned on $e=2$. As $e$ is binary, the method of instrumental variables (IV) coincides with the Wald estimator  \citep{wald1940fitting}. In the population case it can be written as
\begin{equation}\label{eq:4}
 \lim_{n \rightarrow \infty}  \hat \beta_{\text{IV}} =  \frac{\mathbb{E}[Y | e = 1] - \mathbb{E}[Y| e = 2]}{\mathbb{E}[X | e = 1] - \mathbb{E}[X | e = 2]}  = \frac{\mathbb{E}[Y^{1}] - \mathbb{E}[Y^{2}]}{\mathbb{E}[X^{1}] - \mathbb{E}[X^{2}]}.
\end{equation}
Causal Dantzig leads to
\begin{equation}\label{eq:5}
    \lim_{n \rightarrow \infty }\hat \beta =  \frac{\mathbb{E}[X^{1} \cdot Y^{1}] - \mathbb{E}[X^{2} \cdot Y^{2}]}{\mathbb{E}[(X^{1})^{2}] - \mathbb{E}[(X^{2})^{2}]}.
\end{equation}
\begin{center}
\begin{table}[!b]
\begin{tabular}{ l || c | c }
  Consistency & $X = \alpha e + H + \eta_{x}$  &  $X = H + (1+\alpha e ) \eta_{x}$ \\
 & (mean-shift) & (change in error distribution) \\
\hline
  Instrumental variable regression& yes & no \\
  Unregularized causal Dantzig & yes & yes \\
\end{tabular}
\caption{Consistency of the \emph{causal Dantzig} and the instrumental variables approach. Consider a model $Y=  \beta X + H + \eta_{y}$ and a structural equation model for $X$ as depicted in the table. The case on the left is a mean-shift, whereas on the right hand side the error variance changes between setting $e=1$ and $e=2$. We assume $\alpha \neq 0$, and that the random variables $e,\eta_{y},\eta_{x},H$ are independent and non-degenerate with $\mathbb{E}[H]=0$.} \label{table:ivvscd}
\end{table}
\end{center}
Both the IV approach and the \emph{causal Dantzig} have different strengths
and weaknesses  in this setting. For example, equation~\eqref{eq:4} is
based on  means, whereas equation~\eqref{eq:5} is based on
 covariances. If, say, $X= e \cdot \eta_{x} + H$, $Y = \beta X + H
+ \eta_{y}$, with centered noise $\eta_{x},\eta_{y} $ independent of the centered confounder $ H$, then $\mathbb{E}[X | e = 1] = \mathbb{E}[X | e = 2]$.   Hence the IV estimator is not well-defined in the population case and one should use the causal Dantzig. If the instrument is weak, causal Dantzig can exhibit efficiency gains. An example of this can be found in Section~\ref{sec:hidd-instr-vari}. A more general comparison  can be found in Table~\ref{table:ivvscd}. It is also possible to construct examples where equation~\eqref{eq:5} is not
well-defined. For this to happen, the second moments of $X^{1}$ and
$X^{2}$ have to be equal.  \\
A drawback of the IV approach is that the number
of instruments has to equal or exceed the number of endogenous
variables. However, this is not necessary for the \emph{causal Dantzig}. Two
settings  $|\E|  = 2 $ in our framework  correspond to a single
binary exogenous variable. In that case the number of endogenous
variables $p$ can be arbitrarily large as long as $ \G$, the difference of Gram matrices, is invertible. On the other hand, for $p > 2$ the number of endogenous variables exceeds the number of exogenous variables and the IV approach is bound to fail. We compare the performance of the IV approach and \emph{causal Dantzig} on simulated datasets in Section~\ref{sec:hidd-instr-vari}.

\subsection{Inner-product invariance in the potential outcome framework}\label{sec:pot-out}
In this section we will investigate the notion of \emph{inner-product invariance} under potential outcome assumptions \citep{neyman1923application,rubin1974estimating}. Note that here, as in the rest of the paper, we consider a continuous exposure $X \in \mathbb{R}^{p}$. In the following, we use a slightly different notation compared to the rest of the paper. We write $X(e) \in \mathbb{R}^{p}$ for the potential outcome of a continuous exposure if the environment $E$ takes value $e \in \E$. Equivalently we write $Y(x,e) \in \mathbb{R}$ for the potential outcome of the response of a unit if the exposure takes level $X=x$ and environment $E$ takes value $e \in \E$. We assume that these quantities are well-defined. We make the following additional assumptions:
\begin{enumerate}
\item[A1.] Exclusion restriction: 
\begin{align*}
Y(x,e) = Y(x) \text{ holds for all } x \in \text{range}(X) \text{  and } e \in \E
\end{align*}
\item[A2.] Independence:
\begin{align*}
   (X(e),Y(0)) \indep E \text{ for all } e \in \E
\end{align*}
\item[A3.] Constant confounding across environments $\E$:
\begin{equation*}
\text{Cov}(X(e),Y(0)) = \text{Cov}(X(f),Y(0))  \text{ for all } e,f \in \E
\end{equation*}
\item[A4.] Treatment effect homogeneity and linearity:
\begin{align*}
  \mathbb{E}[ Y(x) - Y(0) | X=x ,E=e] &=  \mathbb{E}[ Y(x) - Y(0) ] \\
&=  x \beta^{0} \text{ for all } x \in \text{range}(X) \text{ and } e \in \E
\end{align*}
\item[A5.] The variables are normalized:
\begin{equation*}
  \mathbb{E}[X] = 0 \qquad \text{and} \qquad \mathbb{E}[Y] = 0
\end{equation*}
\end{enumerate}
Note that we did not make any cross-world assumptions \citep{richardson2013single}, i.e. we made no assumptions on the joint distribution of $Y(x)$, $x \in \text{range}(X)$ or on the joint distribution of $X(e)$, $e \in \E$. Condition (A2) can be relaxed to an assumption on the cross-product between $X(e)$ and $Y(0)$. Details can be found in the Appendix in the proof of Proposition~\ref{prop:inner-prod-invarpot}. Condition (A3) is crucial: we allow for confounding (nonzero covariance of $X(e)$ and $Y(0)$), but we assume that the covariance is constant across environments. Loosely speaking, this can be seen as a non-interaction-assumption of environment and confounding. Condition (A4) ensures that the average treatment effect is the same within strata defined by $X$ and $E$ and allows the usage of a linear model. For a discussion of similar assumptions in the context of the IV framework, see \cite{wang2016bounded}.

If these assumptions are fulfilled, then we have \emph{inner-product invariance} under the average treatment effect $\beta^{0}$.
\begin{proposition}\label{prop:inner-prod-invarpot}
  Under assumptions (A1) - (A5) we have \emph{inner-product invariance} under the vector $\beta^{0} \in \mathbb{R}^{p}$ which satisfies $\mathbb{E}[Y(x) - Y(0)] = x \beta^{0}$, i.e.
\begin{equation*}
  \mathbb{E}[ X^{t} (Y - X \beta^{0}) | E=e] =  \mathbb{E}[ X^{t} (Y - X \beta^{0}) | E=f] \text{ for all } e,f \in \E.
\end{equation*}
\end{proposition}
The proof of this result can be found in the Appendix.
Using inner-product invariance for estimating the average treatment effect $\beta^{0}$, it is possible to consistently estimate the average treatment effect in cases in which two-stage least squares (or the Wald estimand) is degenerate. For example, in settings where the dimension of exposure variables $X$ exceeds the number of environments $|E|$ or when $\mathbb{E}[Y-Xb | E = 1] = \mathbb{E}[Y-Xb | E=0]$ for $\E = \{0,1 \}$. In the presence of weak instruments, \emph{causal Dantzig} can exhibit efficiency gains compared to estimators based on conditional means of $X$ and $Y$. This is investigated further in Section~\ref{sec:comparison-with-instrumental-variables} and Section~\ref{sec:numerical-examples}.

\section{Causal Dantzig with regularization}\label{sec:high-dim}

In this section we introduce the  regularized \emph{causal Dantzig}, and discuss its theoretical properties. The estimator is motivated and introduced in Section~\ref{sec:estimator}. Section~\ref{sec:finite-sample-bound} contains  finite sample bounds. The bounds presented in this section involve a quantity that we call the ``causal cone invertibility factor''. The behavior of this quantity is discussed in Section~\ref{sec:behaviourccif}.

\subsection{The estimator}\label{sec:estimator}

Weak interventions on some of the variables (i.e. $\mathbb{E}[( \delta_{k}^{e})^{2}] $ small) may lead to coefficient estimates with high variance in equation~\eqref{eq:hatbetaInv}. Furthermore, if the number of predictors $p$ exceeds the total sample size $n$, the matrix $\hat \G$ is not invertible and the solution to equation~\eqref{eq:hatbeta} is not unique.  In such settings, regularization and shrinkage is desirable and can outperform unpenalized estimation procedures, see e.g. \cite{buhlmann2011statistics}. In particular, $\ell_{1}$-penalized estimation procedures have attracted much interest in high-dimensional models. For linear models, \cite{candes2007dantzig} proposed an $\ell_{1}$-minimization method called the Dantzig selector.
Consider  $\YY = \XX \beta^{*} + \epsilon$ with $\XX \in \mathbb{R}^{n \times p}$, $\YY \in \mathbb{R}^{n}$, $\beta^{*}  \in \mathbb{R}^{p} $. For a tuning parameter $\lambda \ge 0$, the Dantzig selector is defined as a solution to the regularization problem
\begin{align}\label{eq:25}
\begin{split}
  &\min   \| \beta \|_{1} \quad \text{subject to} \quad  \| \tilde \Z - \tilde \G
\beta\|_{\infty} \le \lambda, \\
&\text{where } \tilde \Z =  \XX^{t} \YY /n \text{ and } \tilde \G = \XX^{t} \XX /n.
\end{split}
\end{align}
\begin{figure}
\begin{center}
\includegraphics[width=0.4\textwidth]{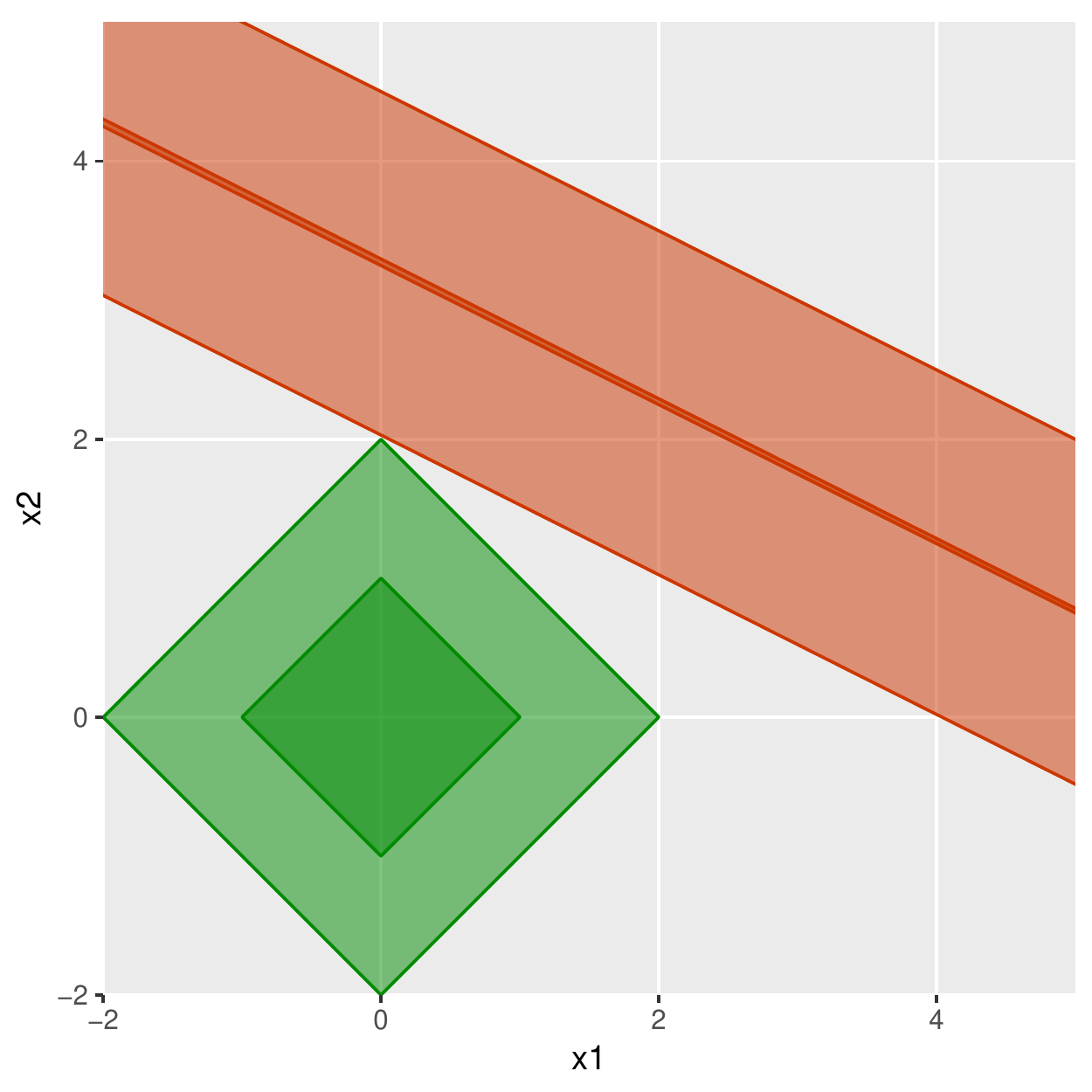}
\caption{ \label{fig:dantzig} {\it A visualization of the Dantzig selector. The red region is the feasible set $ \{ \beta: \| \tilde \Z - \tilde \G \beta ) \|_{\infty} \le \lambda \}$ for $\lambda = \sqrt{1.25}$. The two green regions are the level sets $\{ \beta : \| \beta \|_{1} \le 1 \}$ and $\{ \beta : \| \beta \|_{1} \le 2 \}$. The Dantzig selector for $\lambda = \sqrt{1.25}$ is at the intersection of the light green and the red region.}}
\end{center}
\end{figure}
The geometry of the Dantzig selector is depicted in Figure~\ref{fig:dantzig}. The $\ell_{1}$-minimization favors sparse solutions, i.e. vectors in which many coefficients are exactly zero.  This facilitates interpretation. Furthermore, if $\lambda$ gets larger, the Dantzig selector shrinks towards the zero vector. Choosing $\lambda$ is a trade off: small values will generally result in larger variance of the estimator, but smaller bias.
We propose the regularized \emph{causal Dantzig} $\hat \beta^{\lambda}$, which in analogy to equation~\eqref{eq:hatbeta}  is defined as a solution to
\begin{align}\label{eq:7}
\begin{split}
  \min\| \beta \|_{1}& \quad \text{subject to} \quad  \| \hat \Z - \hat \G
\beta\|_{\infty} \le \lambda, \\
\text{where } \hat \Z &=  \frac{(\Xone)^{t} \Yone}{n_{1}}   - \frac{(\Xtwo)^{t} \Ytwo}{n_{2}}, \\   \text{ and } \hat \G &= \frac{(\Xone)^{t} \Xone}{n_{1}}   -  \frac{(\Xtwo)^{t} \Xtwo}{n_{2}}.
\end{split}
\end{align}
On a superficial level, the difference to the Dantzig selector is merely that $\XX^{t} \YY/n$ is replaced by $ (\Xone)^{t} \Yone/n_{1} - (\Xtwo)^{t} \Ytwo/n_{2}$  and $\XX^{t} \XX /n$ is replaced by $ (\Xone)^{t} \Xone/n_{1} - (\Xtwo)^{t} \Xtwo/n_{2}$. Hence the geometry of the optimization problem is akin to the Dantzig selector and the \emph{causal Dantzig} inherits its variable selection, shrinkage and regularization properties. Furthermore,  the \emph{causal Dantzig} can be cast as a linear program for fixed $\lambda$. Details can be found in the Appendix, Section~\ref{sec:causal-dantzig-as}.

\subsection{Finite-sample bound}\label{sec:finite-sample-bound}

The  regularized \emph{causal Dantzig}  is related to the Dantzig selector and enjoys similar
 properties. Notably, it attains  the same rates of convergence under comparable regularity conditions. To this end, we introduce the quantity  ``causal cone  invertibility factor'', similar to the ``cone invertibility factor''
 for the Dantzig selector as defined in \cite{ye2010rate}. For ease of exposition we will first treat the case $\E = \{1,2\}$. The treatment of the general case is sketched in Remark~\ref{remark:extension}.

\subsubsection{Causal Cone Invertibility Factor}\label{sec:caus-cone-invert}
Let $\hat \SI$ denote the empirical covariance matrix of $\XX$ and consider a set
$S \subset \{1,\ldots,p\}$. Later we will mainly be interested in the case where $S$ is the active set of $\beta^{0}$.
\cite{ye2010rate} proved bounds for the Dantzig selector
that involve the so-called cone
invertibility factor (CIF).
For the upper bound, the relevant quantity in \cite{ye2010rate}  is  $\mathrm{CIF}_{q}(S) $.
Roughly speaking, the cone invertibility factor is a
lower bound on the $\ell_{\infty}$-norm of $\hat \SI u$, given that $u$ lies in
the cone $\{ u : \| u_{S^{c}} \|_{1} \le \| u_{S} \|_{1} \}$ and has
unit norm $\| u \|_{q} =1$. To make the quantity comparable across
different norms, it is scaled by a factor $| S |^{1/q}$. To be more precise,
\begin{equation*}
  \mathrm{CIF}_{q}(S) = \inf_{u} \left\{ \frac{|S|^{1/q}\| \hat \SI u \|_{\infty}}{\| u \|_{q}} : \| u_{S^{c}} \|_{1} \le \| u_{S} \|_{1} \right\}.
\end{equation*}
 Now we are ready to define the \emph{causal cone invertibility factor} $\mathrm{CCIF}_{q}(S,\hat \G)$:
\begin{equation}\label{eq:14}
  \mathrm{CCIF}_{q}(S,\hat \G) := \inf_{u} \left\{ \frac{|S|^{1/q}\| \hat \G u \|_{\infty}}{\| u \|_{q}} : \| u_{S^{c}} \|_{1} \le \| u_{S} \|_{1} \right\}.
\end{equation}
Analogously define $\mathrm{CCIF}_{q}(S,\G)$ for $\G := \mathbb{E}[\hat \G]$. Here and in the following, notationally we do not treat the case $q=\infty$ separately. Instead, with small abuse of notation we set $|S|^{1/q}:=1$ for $q=\infty$.  In the new definition, the positive semi-definite matrix $\hat \SI$ is replaced by the symmetric matrix $\hat \G$. As $\hat \SI$, the matrix $\hat \G$ is not positive
definite in high-dimensional settings and even indefinite in general. However, it can be shown that the CCIF behaves similarly to
the CIF in several ways.  This is further discussed in Section~\ref{sec:behaviourccif}. For now, let us turn to the
finite-sample bound of the \emph{causal Dantzig}.

\subsubsection{Finite sample bound}\label{sec:finite-sample-bound-1}
The finite-sample results of the \emph{causal Dantzig} are  analogous to
the Dantzig selector while the issue of identifiability is now addressed by the causal cone invertibility factor $\mathrm{CCIF}_{q}(S, \hat \G)$. Similarly as in \cite{ye2010rate}, define $  z^{*} := \| \hat \Z - \hat \G \beta^{0}
\|_{\infty}$ and let $S$ denote the active set of $\beta^{0}$. The first result is purely algebraic and follows from the definitions of
$\mathrm{CCIF}_{q}(S,\hat \G)$ and  the \emph{causal Dantzig}.
\begin{lemma}\label{lemma:finite-sample-bound}
On the event $z^{*} \le \lambda$ we have
\begin{equation}\label{eq:35}
  \| \hat \beta^{\lambda} - \beta^{0} \|_{q} \le \frac{|S|^{1/q} (\lambda + z^{*})}{\mathrm{CCIF}_{q}(S, \hat \G)}
 \le \frac{2|S|^{1/q} \lambda}{\mathrm{CCIF}_{q}(S,\hat \G)} \qquad \text{ for all } q \ge 1.
\end{equation}
\end{lemma}
The proof can be found in the Appendix. There are two terms on the right-hand side in equation~\eqref{eq:35} that deserve further attention. First, $\mathrm{CCIF}_{q}(S,\hat \G)$ is bounded away from zero under
certain assumptions, as discussed in Section~\ref{sec:behaviourccif}. Secondly, it is crucial to
understand the behavior of $  z^{*} := \| \hat \Z -  \hat \G \beta^{0} \|_{\infty}$. Using a union bound over the $p$ entries, it can be shown that with high probability, $z^{*}$ is of the order $\max_{e \in \E} \max(\log(p)/n_{e}, \sqrt{\log(p)/n_{e}})$:
\begin{lemma}\label{lemma:finite-sample-bound-1}
Assume that inner-product invariance holds for $(X^{e},Y^{e}), e \in \{1,2\}$ under $\beta^{0}$.  Assume $X^{1},X^{2},  \eta_{p+1}^{1}, \eta_{p+1}^{2}$ are centered and multivariate Gaussian.  Let $t \ge 0$. Then, with probability  exceeding $1- 4 \exp(-t)$,
\begin{align*}
  z^{*} &\le    \sigma_{\varepsilon} \sum_{e \in \{1,2\}}  \sigma_{\text{max}}^{e}\left( \sqrt{\frac{4t +  4\log(p)}{n_{e}}} +   \frac{4t + 4 \log(p)}{n_{e}} \right), \\
  \text{where }  \sigma_{\varepsilon}  &:=\sqrt{ \Var (\eta_{p+1}^{e})} \quad \text{ and } \quad  \sigma_{\text{max}}^{e} :=   \max_{k} \sqrt{ \Var(X_{k}^{e}) }.
\end{align*}
\end{lemma}
The proof can be found in the Appendix. This result can be extended to situations where $(X^{1}, \eta_{p+1}^{1})$ and $(X^{2}, \eta_{p+1}^{2})$ have subgaussian tails, see e.g. exercise~14.3 in \cite{buhlmann2011statistics}.
By combining Lemma~\ref{lemma:finite-sample-bound} and Lemma~\ref{lemma:finite-sample-bound-1} we obtain the following theorem. The proof can be found in the Appendix.

\begin{theorem}\label{theorem:finite-sample-bound-2}
Let $\lambda \asymp 5 C \sqrt{\log(p)/ \min_{e \in  \{1,2\}}n_{e}} \rightarrow 0$ for a constant $C >0$ that satisfies $ \sigma_{\varepsilon} \cdot \sigma_{\text{max}}^{e} \le C < \infty$ for $ e \in \{1,2\}$. Under the assumptions mentioned in Lemma~\ref{lemma:finite-sample-bound-1},
\begin{equation*}
  \| \hat \beta^{\lambda} - \beta^{0} \|_{q}   \le \frac{10 C}{\mathrm{CCIF}_{q}(S,\hat \G)}  |S|^{1/q} \sqrt{\frac{  \log(p)}{ \min_{e \in  \{1,2\}}n_{e} }}
\end{equation*}
 with $\mathbb{P} \rightarrow 1$ for $n_{1},n_{2},p \rightarrow \infty$.
\end{theorem}
Another consequence of these two Lemmata is the screening property of the \emph{causal Dantzig} under a so-called betamin-condition. The short proof can be found in the Appendix.
\begin{proposition}\label{prop:screening}
Let $\hat S$ denote the active set of $\hat \beta^{\lambda}$. Using the notation of Theorem~~\ref{theorem:finite-sample-bound-2}, assume that
\begin{equation*}
  \min_{k \in S} | \beta_{k}^{0} | > \frac{10 C}{\mathrm{CCIF}_{\infty}(S,\hat \G)} \sqrt{\frac{\log(p)}{\min_{e \in \{1,2\}} n_{e}}}.
\end{equation*}
Then under the assumptions mentioned in Theorem~\ref{theorem:finite-sample-bound-2} for $q= \infty$, we have
\begin{equation*}
  \mathbb{P}[\hat S \supseteq S] \rightarrow 1 \qquad \text{ for $n_{1},n_{2},p \rightarrow \infty$}.
\end{equation*}
\end{proposition}
Note that the convergence rate in Theorem~\ref{theorem:finite-sample-bound-2} coincides with the usual  rate of  convergence in high-dimensional linear regression (\cite{ye2010rate}) under comparable assumptions. For consistency in the $\ell_{2}$ norm in the regression setting it is required that  $ |S|  \log(p)/n \rightarrow 0$, that $\lambda \asymp C \sqrt{\log(p)/n}$ for constant $C>0$ large enough and that the population quantity $\mathrm{CIF}_{2}(S)$ is bounded away from zero. In our framework, if $n_{1}\asymp n_{2}$, the assumptions on the asymptotic behavior of $n=n_{1}+n_{2},p,|S|$ and $\lambda$ stay essentially the same, but $\mathrm{CCIF}_{2}(S, \hat \G)$ plays the role of $\mathrm{CIF}_{2}(S)$.
The next section aims to shed some light on the behavior of this quantity.
\begin{remark}\label{remark:extension}
The results of this section can be extended to more than two settings $|\E| > 2$. To be more precise, in the general case one can define the regularized \emph{causal Dantzig} as a solution to
\begin{equation*}
 \min_{\beta\in \mathbb{R}^p}  \| \beta \|_{1} \text{ subject to }\max_{e \in \E} \| \hat \Z^{e} - \hat \G^{e}
\beta\|_{\infty} \le \lambda,
\end{equation*}
where $\hat \Z^{e}, \hat \G^{e}, e \in \E$ are defined as in equation~\eqref{eq:24}. The causal cone invertibility factor is then defined as
  \begin{equation*}
\mathrm{CCIF}_{q}(S, \{ \hat \G^{e}, e \in \E \}) := \inf_{u} \max_{e \in \E} \left\{ \frac{|S|^{1/q}\left\| \hat \G^{e} u \right\|_{\infty}}{\| u \|_{q}} : \| u_{S^{c}} \|_{1} \le \| u_{S} \|_{1} \right\}.
\end{equation*}
With this notation, it is straightforward to obtain analogous results to Lemma~\ref{lemma:finite-sample-bound}-\ref{lemma:CCIF-bound}, Theorem~\ref{theorem:finite-sample-bound-2} and Proposition~\ref{prop:screening}.
\end{remark}

\subsection{Behavior of the causal cone invertibility factor}\label{sec:behaviourccif}

In the preceding section we showed that the causal cone invertibility factor $\mathrm{CCIF}_{q}(S, \hat \G)$ is a crucial quantity to understand the behavior of the regularized \emph{causal Dantzig}. How do we guarantee that this quantity is bounded away from zero? There are two issues that we will treat separately. First, for $p>n= n_{1}+n_{2}$,
 $ \hat \G$ is not invertible. Secondly, the environments might not be
sufficiently different to make population version $\G$ invertible. In Section~\ref{sec:general-properties} we will discuss how to relate the empirical causal cone invertibility factor to the population causal cone invertibility factor. In Section~\ref{sec:p--n} we consider the case where the environments are sufficiently different to make the population version $\G$ invertible. In Section~\ref{sec:population-g-not} we examine a setting where the environments are not sufficiently different, i.e. where $\G$ is not invertible.

\subsubsection{General properties}\label{sec:general-properties}

In this section we discuss how to relate the empirical causal cone invertibility factor $\mathrm{CCIF}_{q}(S, \hat \G)$ to the population quantity  $\mathrm{CCIF}_{q}(S,  \G)$. The following Lemma gives a deterministic bound for these quantities. The proof can be found in the Appendix.
\begin{lemma}\label{lemma:CCIF-bound} Let $ q \ge 1$. Then,
\begin{equation*}
|\mathrm{CCIF}_{q}(S,\hat \G) - \mathrm{CCIF}_{q}(S,\G)  | \le 2 |S | \| \hat \G - \G \|_{\infty},
\end{equation*}
where $\| A \|_{\infty} := \max_{i,j} |A_{i,j}|$ denotes the matrix max norm.
\end{lemma}
Hence the problem is reduced to understanding the behavior of $\| \hat \G - \G \|_{\infty}$. Let the rows of $\XX^{e}$ consist of i.i.d. centered multivariate Gaussian random variables for $e \in \{1,2\}$.
It can be shown that with probability at least $ 1-4 \exp(-t)$,
\begin{equation}\label{eq:20}
  \| \hat \G - \G \|_{\infty} \le    \sum_{e \in \{1,2\}} ( \sigma_{\text{max}}^{e})^{2}  \left( \sqrt{\frac{4t +  8\log(p)}{n_{e}}} +   \frac{4t + 8 \log(p)}{n_{e}} \right).
\end{equation}
This result can be extended to situations where $X^{1}$ and $X^{2}$ have subgaussian tails, see e.g. exercise~14.3 in \cite{buhlmann2011statistics}. Hence by Lemma~\ref{lemma:CCIF-bound}, even if $\hat \G$ is not
invertible, the quantity in equation~\eqref{eq:14} is well behaved for
$ \sqrt{\min(n_{1},n_{2})} \gg |S| \sqrt{\log(p)}$, in the sense that it is strictly bounded away from
zero if the same is true for the population quantity. The latter assumption
is nontrivial and depends on the distribution of the interventions
$\delta^{e},e \in \{1,2\}$.

\subsubsection{Population $\G$ invertible}\label{sec:p--n}

Under the assumptions discussed in Section~\ref{sec:general-properties}, $\mathrm{CCIF}_{q}(S,\hat \G)$ is bounded away from zero if $\mathrm{CCIF}_{q}(S,\G)$ is bounded away from zero. Hence, the problem is reduced to understanding the population quantity $\mathrm{CCIF}_{q}(S,\G)$. If $\G$ is invertible, then
\begin{align}\label{eq:1}
\begin{split}
  \mathrm{CCIF}_{q}(S,\G) &\ge \min_{u} \frac{|S|^{1/q} \| \G u \|_{\infty}}{\|u \|_{q}}  \\
&= \min_{ \| u \|_{q} = |S|^{1/q}} \| \G u \|_{\infty} >0.
\end{split}
\end{align}
As
\begin{align*}
\G &= \mathbb{E} \left[ \left(X_{1:p}^{1} \right)^{t}  X_{1:p}^{1} - \left(X_{1:p}^{2} \right)^{t} X_{1:p}^{2} \right] \\
&= ((\mathrm{Id}-A)^{-1})_{1:p,1:p} \mathbb{E} [ (\delta_{1:p}^{1})^{t} \delta_{1:p}^{1} - (\delta_{1:p}^{2})^{t} \delta_{1:p}^{2} ] ((\mathrm{Id} - A)^{-t})_{1:p,1:p},
\end{align*}
this is a measure of the difference in the intervention strength $\delta^{e}$ between the two settings $e=1$ and $e=2$. In this sense, this bound is similar to the discussion in Section~\ref{sec:population-g-not}. However, the bound fails to capture appropriately what happens if the interventions only act on a subset of the variables $X_{i}, i=1,\ldots,p$. In that case the bound in equation~\eqref{eq:1} is not useful as $\G$ is not invertible. The next section shows that in some of these settings it is still true that $\mathrm{CCIF}_{q}(S,\G)>0$.

\subsubsection{Population $\G$ not invertible}\label{sec:population-g-not}

The setting of Section~\ref{sec:p--n} and the bound in equation~\eqref{eq:1} are rather restrictive. Consider a situation with a block structure in the Gram matrix, i.e. where $\mathbb{E}[X_{k}^{e}X_{k'}^{e}] = 0$ for all $k \le k_{0} < k'$ and $e \in \E$.  In this case, there might be no interventions on the variables $\{X_{k'}, k' > k_{0}\}$, i.e. $\delta_{k'}^{e} \equiv 0$ for all $k' > k_{0}$. As a result, $\G$ might not be invertible. However, if $\G_{1:k_{0},1:k_{0}}$ is invertible and $S \subset \{1,\ldots k_{0}\}$, then
\begin{align*}
\begin{split}
  \mathrm{CCIF}_{q}(S,\G) &\ge  \inf \left\{ \frac{|S|^{1/q}\| \G_{1:k_{0},1:k_{0}} u_{1:k_{0}} \|_{\infty}}{\| u \|_{q}} : \| u_{S^{c}} \|_{1} \le \| u_{S} \|_{1}  \right\} \\
&\ge \inf \left\{ \frac{|S|^{1/q}\| \G_{1:k_{0},1:k_{0}} u_{1:k_{0}} \|_{\infty}}{2 \| u_{1:k_{0}} \|_{q}}\right\} > 0.
\end{split}
\end{align*}
Hence, under the assumptions discussed in Section~\ref{sec:finite-sample-bound-1}, the \emph{causal Dantzig} is a consistent estimator for $\beta^{0}$. Generally speaking, the \emph{causal Dantzig} tends to screen out variables that have not been affected by the intervention. In this light it is crucial that the interventions act on the variables in the active set of $\beta^{0}$ directly or indirectly.

\section{Practical considerations}\label{sec:pract-cons}

In this section we discuss practical considerations for the \emph{causal Dantzig}. Recommendations are given for centering and scaling of the variables, choice of the regularization parameter $\lambda$ and a procedure for preselection.

\subsection{Centering and scaling}\label{sec:scaling-variables}

Centering and scaling in the \emph{causal Dantzig} setting is a bit more intricate than in a regression setting. Let $ \hat \mu^{e} \in \mathbb{R}^{p+1}$ denote the empirical mean of $(\XX^{e},\YY^{e})$. For centering, we recommend substracting $\frac{1}{|\E|} \sum_{e \in \E} \hat \mu^{e}$ from each sample. By mean-centering globally (and not with an environment-specific intercept), the estimator is able to leverage changes in mean between environments. For scaling, define
\begin{equation}\label{eq:22}
  c_{k,e} = \frac{\mathbb{E} \left[(X_{k}^{e})^{2} \right]}{n_{e}} + \frac{1}{(|\E|-1)^{2}} \sum_{e' \neq e} \frac{\mathbb{E} \left[(X_{k}^{e'})^{2}\right]}{n_{e'}} \text{ for $e \in \E$ and } k=1,\ldots,p.
\end{equation}
We recommend to scale the $k$-th row of $\hat \Z^{e}$ and $ \hat \G^{e}$ by approximately $1/\sqrt{c_{k,e}}$ for all $k=1,\ldots,p$ and $e \in \E$.
What is the motivation behind this scaling? In the following we will discuss the special case $\E = \{1,2\}$. In absence of noise in equation~\eqref{eq:7}, $ \| \Z - \G \beta^{0} \|_{\infty} = 0$. By allowing for $ \| \hat \Z -  \hat \G \beta^{0} \|_{\infty} \le \lambda$, we account for the variance of $\hat \Z - \hat \G \beta^{0}$. Since we work with a supremum bound and the same $\lambda$ for all components, we want all scaled components to have roughly the same variance.
To be more precise, we want
\begin{equation}\label{eq:31}
\text{Var} \left( \frac{( \hat \Z - \hat \G \beta^{0})_{k} }{\sqrt{c_{k,1}}}\right)  = \text{Var} \left( \frac{( \hat \Z -\hat \G \beta^{0})_{l} }{\sqrt{c_{l,1}}}\right) \text{ for all } k,l=1,\ldots,p.
\end{equation}
It can be challenging to scale according to equation~\eqref{eq:31} as the correlation between $X_{k}^{e}$ and $\eta_{p+1}^{e} = Y^{e} - X^{e} \beta^{0}$ is unknown and changes for different $k$.
In the absence of confounding however and if $X_{k}$ and $X_{l}$ are not descendants of $Y$ in the graph $G$, $\varepsilon=\eta_{p+1}^{e}$ is independent of $X_{k}^{e}$ and $X_{l}^{e}$ and the scaling of equation~\eqref{eq:22} implies
\begin{align*}
\text{Var} \left( \frac{(\hat \Z - \hat \G \beta^{0})_{k} }{\sqrt{c_{k,1}}}\right) =                                                                                      \sigma_{\varepsilon}^{2}  = \text{Var} \left( \frac{(\hat \Z - \hat \G \beta^{0})_{l} }{\sqrt{c_{l,1}}}\right),
\end{align*}
where $\sigma_{\varepsilon}$ denotes the standard deviation of $\varepsilon=\eta_{p+1}^{e}$. The scaling of equation~\eqref{eq:22} still has some theoretical justification in more general cases. In the presence of confounding and for general $k,l$ it depends on the joint distribution of $X_{k}^{e}, X_{l}^{e}$ and $\varepsilon = \eta_{p+1}^{e}$ whether $\text{Var} \left( \frac{( \hat \Z - \hat \G \beta^{0})_{k} }{\sqrt{c_{k,1}}}\right) $ and  $\text{Var} \left( \frac{( \hat \Z - \hat \G \beta^{0})_{l} }{\sqrt{c_{l,1}}}\right) $ are of the same order. Notably, if equation~\eqref{eq:22} holds with equality and  if the variables $X_{k}^{e}, k=1,\ldots,p$ and $\varepsilon = \eta_{p+1}^{e}, e \in \{1,2\}$ are centered multivariate Gaussian, using  moment inequalities,
\begin{equation*}
 \mathbb{E} \left[ \left(X_{k}^{e} \right)^{2} \right] \sigma_{\varepsilon}^{2} \le   \Var(X_{k}^{e} \eta_{p+1}^{e}) \le 2  \mathbb{E} \left[\left(X_{k}^{e} \right)^{2} \right]  \sigma_{\varepsilon}^{2}
\end{equation*}
 for $e \in \{1,2\}, k \in \{1,\ldots,p\}$. Using independence of samples from different environments $e \in \{1,2\}$,
\begin{equation*}
 \sum_{e \in \{1,2\}}  \frac{ \mathbb{E} \left[\left(X_{k}^{e} \right)^{2} \right] }{n_{e}}  \sigma_{\varepsilon}^{2}   \le   \text{Var} \left(( \hat \Z -\hat \G \beta^{0})_{k}\right)  \le 2 \sum_{e \in \{1,2\}} \frac{ \mathbb{E} \left[\left(X_{k}^{e} \right)^{2} \right] }{n_{e}} \sigma_{\varepsilon}^{2}
\end{equation*}
for all $k=1,\ldots,p$. Using equation~\eqref{eq:22},
\begin{equation*}
 \sigma_{\varepsilon}^{2} \le \text{Var} \left( \frac{( \hat \Z -\hat \G \beta^{0})_{l} }{\sqrt{c_{l,1}}}\right) \le 2 \sigma_{\varepsilon}^{2} \text{ for all } k=1,\ldots,p.
\end{equation*}
Hence $\text{Var} \left( \frac{( \hat \Z - \hat \G \beta^{0})_{l} }{\sqrt{c_{l,1}}}\right) $ and $\text{Var} \left( \frac{( \hat \Z - \hat \G \beta^{0})_{k} }{\sqrt{c_{k,1}}}\right) $ are of the same order for all $k,l =1,\ldots,p$.

\subsection{Choosing $\lambda$}\label{sec:choosing-lambda}
Large segments of the regularization path of the \emph{causal Dantzig} are usually poor estimates of $\beta^{0}$. Hence it is crucial to use an appropriate value of the  regularization parameter $\lambda$. From a theoretical perspective one would choose $\lambda$ as in Theorem~\ref{theorem:finite-sample-bound-2}. However, $\sigma_{\varepsilon}$ and $\sigma_{max}^{e}$ are usually unknown in real-world datasets. Hence, in practice we propose to choose $\lambda$ by $k$-fold cross-validation. Concretely, in each environment $e \in \E$ the samples are split into $k$ groups of approximately equal size. Denote $\hat \beta^{\lambda,-i}$ the \emph{causal Dantzig} estimator that is calculated on all samples except the samples from group $i$. Let $\hat \Z^{i}$ and $ \hat \G^{i}$ be defined as in equation~\eqref{eq:3}, using the samples from group $i$. Then we can choose $\hat \lambda^{\text{cv}}$ as a solution to
\begin{equation*}
  \hat \lambda^{\text{cv}} = \argmin_{\lambda} \frac{1}{k} \sum_{i=1}^{k} \| \hat \Z^{i}  - \hat \G^{i} \hat \beta^{\lambda,-i} \|_{\infty}.
\end{equation*}
We define the cross-validated \emph{causal Dantzig} as $\hat \beta^{\text{cv}} := \hat \beta^{\hat \lambda^{\text{cv}}}$. Two exemplary regularization paths and the solution chosen by cross-validation are depicted in Figure~\ref{fig:cool}. \begin{figure}
\begin{center}
\includegraphics[width=0.4\textwidth]{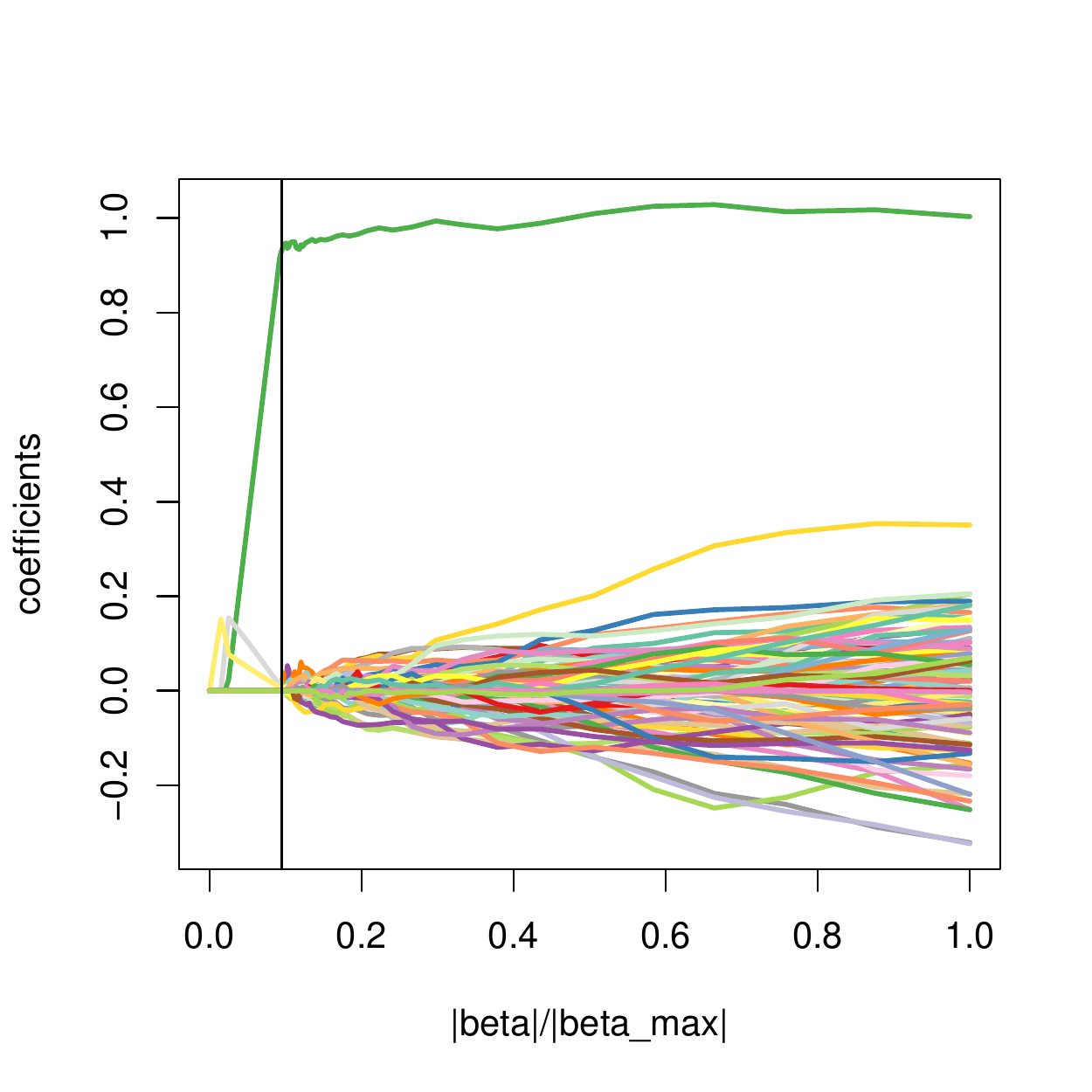}
\includegraphics[width=0.4\textwidth]{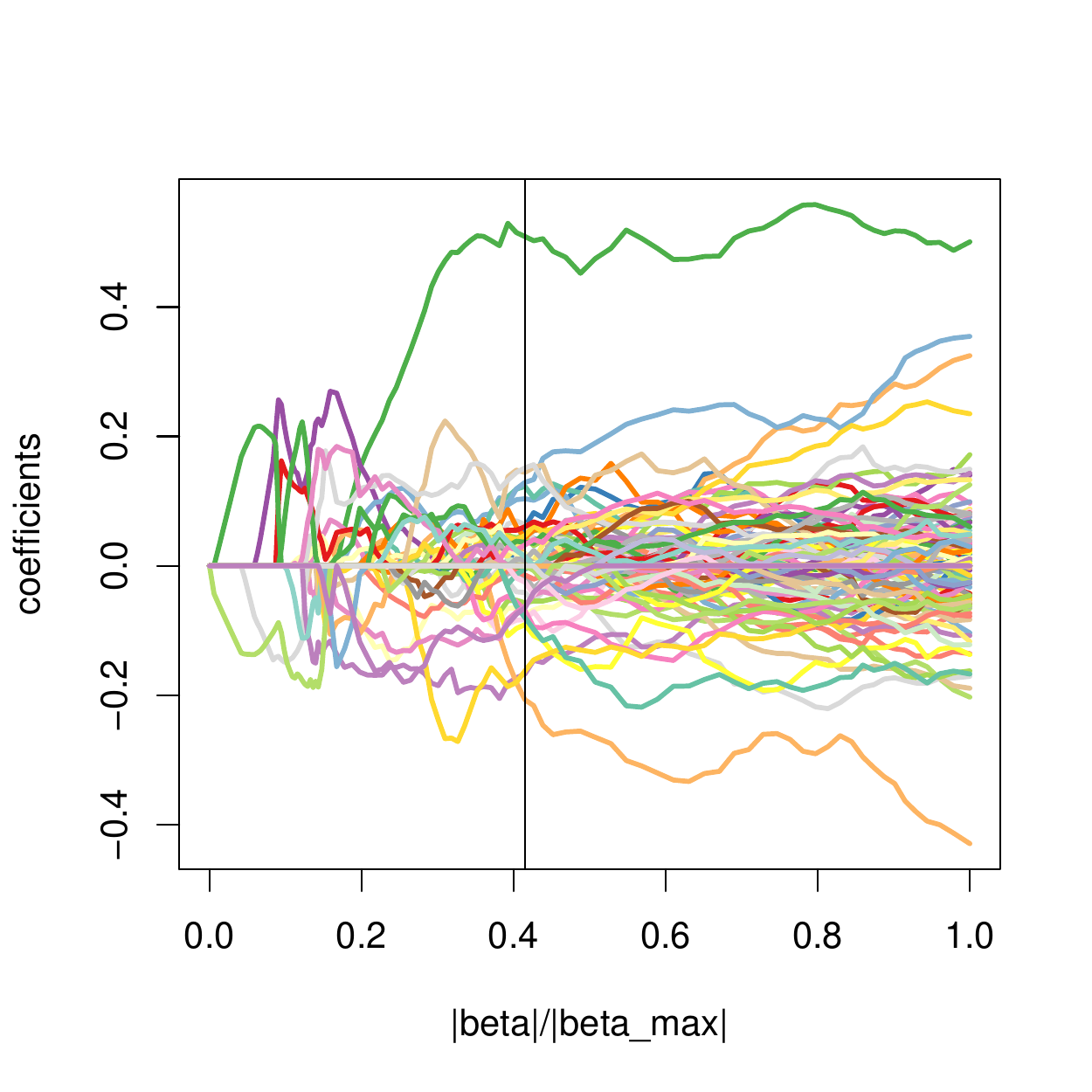}
\caption{ \label{fig:cool} {\it  Two typical regularization paths for the
     causal Dantzig.  The black vertical line specifies the
    solution chosen by $10$-fold cross-validation. On the left-hand
    side $p=100$, $n=200$. On the right-hand side $p=200$, $n=60$. In
    both cases the standard deviation of the interventions is $2.5$ and the
    variance of the errors is $1$. One component of $\beta^{0}$ is equal to one (upper
    green line), all others are zero.}}
\end{center}
\end{figure}

\subsection{Preselection with hidden variables}\label{sec:mark-blank-estim-1}

An alternative of running the \emph{causal Dantzig} directly on a high-dimensional dataset is doing preselection. In the first stage we recommend to run Lasso on observational data, if available. If observational data is not available, one could run Lasso on the pooled dataset. In the second stage, one would run the \emph{causal Dantzig} with or without regularization on the active set of the first stage. Ideally, the first stage would screen out as many variables as possible, except for the parental set of the target variable $Y$. Quite often this will result in a set that contains a superset of the parental set implying a very useful dimensionality reduction. The following Lemma provides some justification for this approach.
\begin{lemma}\label{lemma:hidden-markov-blanket}
Assume that the distribution $X_{1},...,X_{p},Y$ is generated by a linear acyclic Gaussian structural equation model with directed acyclic graph $D_{total}$ that consists of both the observed variables $X_{1}, \ldots, X_{p}, Y$ and (potentially) hidden confounders $H_{1},....,H_{q}$. Assume that the joint distribution of the variables $X_{1},...,X_{p},Y,H_{1},...,H_{q}$ is faithful \citep{Pearl2009} to $D_{total}$. Let $S$ denote the active set of regressing $Y$ on $X_{1},\ldots,X_{p}$ in the population case. Then,
\begin{align*}
    \{ k : X_{k}  \text{ is a parent or a child of $Y$ in $D_{total}$ }  \} \subset S.
\end{align*}
\end{lemma}
The proof can be found in the Appendix. We test this two-step procedure on real world data in Section \ref{sec:gene-knock-exper}.
However, note that for valid $p$-values (with the unregularized \emph{causal Dantzig}) we would have a post-selection problem due to the screening step.

\section{Numerical examples}\label{sec:numerical-examples}

Section~\ref{sec:hidd-low-dimens} explores actual coverage and length of the asymptotic confidence intervals as defined in Section~\ref{sec:confidence-intervals}. In Section~\ref{sec:hidd-instr-vari} we compare the \emph{causal Dantzig} to instrumental variable regression for $p=1$ under different types of interventions. In Section~\ref{sec:hidd-high-dimens} we evaluate the performance of parameter selection by cross-validation as defined in Section~\ref{sec:choosing-lambda}. Finally, in Section~\ref{sec:gene-knock-exper} we discuss an application to real-world data that has been published in \citet{meinshausen2016methods}.

\subsection{Causal Dantzig in low dimensions: confidence intervals}\label{sec:hidd-low-dimens}

In this section we explore the actual coverage and average length
of the asymptotic confidence intervals constructed according to
Theorem~\ref{theorem:confidence-intervals}.
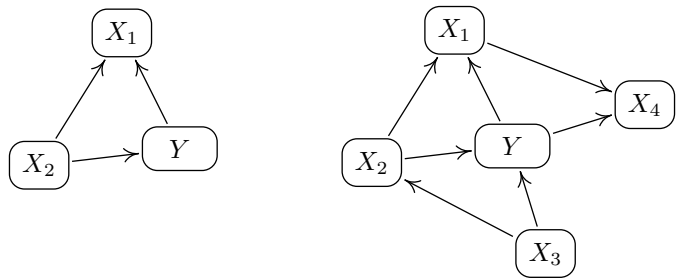
\begin{figure}
\begin{center}
\begin{tikzpicture}[scale=1.1, line width=0.5pt, minimum size=0.58cm, inner sep=0.3mm, shorten >=1pt, shorten <=1pt]
    \normalsize
    \draw (-1,1.4) node(1) [rectangle, rounded corners=2mm, inner
     sep=1.7mm, draw] {$X_{1}$};
    \draw (-0.3,0) node(y) [rectangle, rounded corners=2mm, inner
     sep=1.7mm, draw] {$\phantom{a} Y\phantom{a}$};
    \draw  (-2,-0.2) node(2) [rectangle, rounded corners=2mm, inner
     sep=1.7mm, draw] {$X_2$};
        \draw (0.1,-1.3) node(3) [rectangle, rounded corners=2mm, inner
     sep=1.7mm] {$ $};
    \draw (1.3,0.5) node(4) [circle] {$ $};
    \draw[-arcsq] (2) -- (y);
    \draw[-arcsq] (2) -- (1);
    \draw[-arcsq] (y) -- (1);
  \end{tikzpicture}
\begin{tikzpicture}[scale=1.1, line width=0.5pt, minimum size=0.58cm, inner sep=0.3mm, shorten >=1pt, shorten <=1pt]
    \normalsize
    \draw (-1,1.4) node(1) [rectangle, rounded corners=2mm, inner
     sep=1.7mm, draw] {$X_{1}$};
    \draw (-0.3,0) node(y) [rectangle, rounded corners=2mm, inner
     sep=1.7mm, draw] {$\phantom{a} Y\phantom{a}$};
    \draw  (-2,-0.2) node(2) [rectangle, rounded corners=2mm, inner
     sep=1.7mm, draw] {$X_2$};
    \draw (1.3,0.5) node(4) [rectangle, rounded corners=2mm, inner
     sep=1.7mm, draw] {$X_4$};
    \draw (0.1,-1.3) node(3) [rectangle, rounded corners=2mm, inner  sep=1.7mm, draw] {$X_3$};
    \draw[-arcsq] (3) -- (y);
    \draw[-arcsq] (2) -- (y);
    \draw[-arcsq] (3) -- (2);
    \draw[-arcsq] (2) -- (1);
    \draw[-arcsq] (y) -- (1);
    \draw[-arcsq] (y) -- (4);
    \draw[-arcsq] (1) -- (4);
   \end{tikzpicture}
\caption{ \label{fig:graphs} {\it The  graphs (A) and (B) used in the
    simulations. The noise distributions at all variables follow a
    factor model which allows for hidden confounding. }}
\end{center}
\end{figure}
We simulate data from two linear SEMs shown in
Figure~\ref{fig:graphs}. Specifically, the data are generated according
to the equations
\begin{align}\label{eq:simulatins}
  (A): \left\{\begin{array}{rcrrr}
X_2 & \leftarrow &  && \eta_2 \\
Y & \leftarrow&  &X_2 +& \eta_y \\
 X_1 & \leftarrow& Y - &X_2 + &\eta_1
\end{array} \right. , \quad
(B): \left\{\begin{array}{rcrrr}
X_3 & \leftarrow &  && \eta_3 \\
X_2 & \leftarrow &  &X_3+& \eta_2 \\
Y & \leftarrow &  -X_3+&X_2+& \eta_y \\
 X_1 & \leftarrow & -X_2+ & Y + &\eta_1 \\
 X_4 & \leftarrow & - Y +& X_1+ &\eta_4
\end{array} \right. ,
\end{align}
where the noise distributions of $(\eta_1,\eta_2,\eta_y)$ and
$(\eta_1,\eta_2,\eta_3,\eta_4,\eta_y)$  respectively depend on the
environment. Specifically, for SEM (A), we assume a factor model for the noise
\[ (\eta_1,\eta_2,\eta_y)^t = A H +
\sigma_j (\varepsilon_1,\varepsilon_2,\varepsilon_y)^t ,\]
where $(\varepsilon_1,\varepsilon_2,\varepsilon_y)^t \sim
\mathcal{N}(0,1_3)$, and the entries in both the factor loading matrix  $A\in
\mathbb{R}^{3\times 5}$ and the factor values $H \in \mathbb{R}^{5}$ are chosen i.i.d.\ standard normal. The  5-dimensional variable $H$
act as hidden confounders between the observed variables. The noise contribution
$\sigma_j$ is chosen as 1 in environment $e = 1$ and as $1+\kappa$ in
environment $e =2$. We call $\kappa=\sigma_1-\sigma_2$ the
intervention strength as it measures the variance of the additional
noise input in environment $e=1$ over environment $e=2$. In our simulations it is chosen as $8$.
For SEM (B) we generate the data analogously with the dimension of
the hidden variable $H$ being five.

We draw $n \in \{50,100,500,1000\}$ samples in total (across both environments) and compute the confidence intervals for the causal
coefficients $\beta^{0}$ of $Y$ with the unregularized \emph{causal Dantzig}.
For SEM (A), the true causal coefficients for $Y$ are given by
$\beta^{0}=(0,1)$ and the actual coverage and average length of the constructed intervals at confidence level 0.05 with the unregularized \emph{causal Dantzig} is  shown in the two upper rows of Table~\ref{simulations1length} for
variable $X_1$. The
bottom row show the coverage of the confidence
intervals for invariant causal prediction (ICP). For large $n$, ICP often (rightfully) rejects
all models and outputs neither coefficient estimates nor confidence
intervals.  These cases were ignored in the table. ICP is not consistent
and hence has incorrect coverage for growing sample size, as clearly
visible in the table.
\begin{table}[ht]
\centering
\begin{tabular}{rllll}
  \hline
 & $n = 50 $ & $ 100 $ & $ 500 $ & $ 1000 $ \\ 
  \hline
Coverage & 0.93$\pm$0.01 & 0.95$\pm$0.01 & 0.96$\pm$0.01 & 0.96$\pm$0.01 \\ 
  Average length & 65.79$\pm$2918.53 & 4.11$\pm$602.53 & 0.27$\pm$0.62 & 0.18$\pm$0.01 \\ 
  Coverage ICP & 0.92$\pm$0.01 & 0.84$\pm$0.01 & 0.42$\pm$0.02 & 0.3$\pm$0.03 \\ 
   \hline
\end{tabular}
\caption{The first two rows contain actual coverage and average length of confidence intervals of \emph{causal Dantzig} for the first variable in SEM (A) of equation~\eqref{eq:simulatins}. The last row contains the actual coverage of ICP in these settings. The nominal coverage is $0.95$ for causal Dantzig and at least $0.95$ for ICP. For small sample sizes, the variance is relatively large. As discussed in Section~\ref{sec:high-dim}, regularization can be helpful in these settings.} 
\label{simulations1length}
\end{table}
 The \emph{causal Dantzig} has approximately correct coverage for all sample sizes in this example. For small
sample sizes, the variance of the \emph{causal Dantzig} is large and consequently the average length of the confidence intervals of
the \emph{causal Dantzig} is large, too. In such regimes, regularization
is recommended, as discussed in Section~\ref{sec:high-dim}. For larger
sample sizes, the
confidence intervals are shrinking considerably with the
$\sqrt{n}$-rate. For SEM (A), this effect is depicted  in
Table~\ref{simulations1length}.
Table~\ref{simulations2} shows these effects for SEM
(B). Note that also in this case the actual coverage of the \emph{causal Dantzig} is approximately correct.
 \begin{table}[ht]
\centering
\begin{tabular}{rllll}
  \hline
 & $n = 50 $ & $ 100 $ & $ 500 $ & $ 1000 $ \\ 
  \hline
Coverage & 0.95$\pm$0.01 & 0.95$\pm$0.01 & 0.96$\pm$0.01 & 0.96$\pm$0.01 \\ 
  Average length & 11354.75$\pm$2776.95 & 57.27$\pm$28842.69 & 0.69$\pm$7.28 & 0.39$\pm$3.73 \\ 
   \hline
\end{tabular}
\caption{Actual coverage and average length of confidence intervals for first variable in SEM (B) of equation~\eqref{eq:simulatins} with causal Dantzig. The nominal coverage is 0.95.} 
\label{simulations2}
\end{table}
 \subsection{Causal Dantzig and the instrumental variable approach}\label{sec:hidd-instr-vari}
To compare the \emph{causal Dantzig} to instrumental variables, consider a binary
instrument $e \in \{0,1\}$. To be more precise, we consider the model
\begin{align}\label{eq:6}
\begin{split}
  H,\epsilon_{1},\epsilon_{2} &\sim \mathcal{N}(0,1), e \in \{0,1\} \\
  X &= H + 2 e + \epsilon_{1} \\
  Y &= 2X + H + 2 \epsilon_{2}.
\end{split}
\end{align}
The corresponding DAG is depicted in Figure~\ref{fig:ivdag}. In words,
$X$ is a direct cause of $Y$, there is a hidden confounder $H$ that causes both $X$ and
$Y$, and $e$ is an instrument for $X$, meaning that $e$ is a root node
and a direct cause of $X$, but not of $H$ or $Y$. Note that the conditional mean differs between settings, i.e. $\mathbb{E}[X | e=1] \neq\mathbb{E}[X | e=0]$. Hence the IV approach is consistent for
the true causal effect from $X$ to $Y$, as discussed in Section~\ref{sec:comparison-with-instrumental-variables}.
 \begin{center}
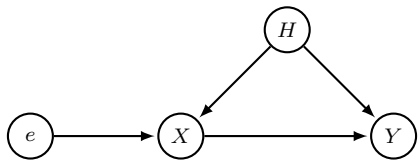
\begin{figure}
\begin{center}
\begin{tikzpicture}[->,>=latex,shorten >=1pt,auto,node distance=2cm,
                    thick]
  \tikzstyle{every state}=[draw=black,text=black, inner sep=0.4pt, minimum size=17pt]

  \node[state] (H) {$H$};
  \node[state] (Y) [below right of=H] {$Y$};
  \node[state] (X) [below left of=H] {$X$};
  \node[state] (e) [left of=X] {$e$};

 \draw (H) --  (Y); \draw (H) -- (X);  \draw (e) --  (X);  \draw (X) -- (Y);

\end{tikzpicture}
\caption{\label{fig:ivdag} The DAG corresponding to the model of equation~\eqref{eq:6}}
\end{center}
\end{figure}
\end{center}
For each environment $e \in \{0,1\}$ we generate $n$ samples and estimate
the direct causal effect via \emph{causal Dantzig} and instrumental variables regression using the function \texttt{ivreg} in the R-package \texttt{AER}. Table~\ref{crvsiv1} shows the mean square error for
both methods. For few observations, the \emph{causal Dantzig} is relatively unstable.
\begin{table}[ht]
\centering
\begin{tabular}{rlll}
  \hline
 & $n = 20 $ & $ 50 $ & $ 100 $ \\ 
  \hline
causal Dantzig & 0.46$\pm$0.41 & 0.03$\pm$0.01 & 0.01$\pm$0 \\ 
  ivreg & 0.07$\pm$0.01 & 0.02$\pm$0 & 0.01$\pm$0 \\ 
   \hline
\end{tabular}
\caption{Mean square error for varying $n$. Instrument is not weak.} 
\label{crvsiv1}
\end{table}
 For larger values of $n$, this is not
the case and  the mean square error shrinks at the $\sqrt{n}$-rate for both
estimators. The instrumental variables (IV) approach outperforms the \emph{causal Dantzig} in this example. This
is due to the fact that IV is a fraction of conditional means, whereas the
\emph{causal Dantzig} is a fraction of conditional covariances. Estimating
conditional means is statistically easier, but it comes at a certain price as we
will see below.\\
For the second model, we change the edge function between $e$ and $X$. Notably,
\begin{align}\label{eq:10}
\begin{split}
  H,\epsilon_{1},\epsilon_{2}, \epsilon_{3} &\sim \mathcal{N}(0,1), e
 \in \{0,1\} \\
  X &= H + 2 e \cdot (0.25 +  \epsilon_{3}) + \epsilon_{1} \\
  Y &= 2X + H + 2 \epsilon_{2}.
\end{split}
\end{align}
Both the conditional variance $\text{Var}(X|e=\bullet), \bullet \in \{0,1\}$ and the conditional mean $\mathbb{E}[X | e=\bullet], \bullet \in \{0,1\}$ change between the environments. However, the conditional mean changes only slightly, imposing difficulties for the IV approach. Again, for each environment $e \in \{0,1\}$ we
generate $n$ samples  and estimate the
direct causal effect via \emph{causal Dantzig} and  \texttt{ivreg}. As seen in Table~\ref{crvsiv2}, for very few observations, both \texttt{ivreg} and \emph{causal Dantzig}
 are comparatively far from the target quantity. For larger values of $n$, the \emph{causal Dantzig} converges with
 the $\sqrt{n}$-rate. The instrumental variables approach is consistent but unstable for these small sample sizes as the instrument is weak. It exhibits large MSE as it does not use the changing variance for inference.

\begin{table}[ht]
\centering
\begin{tabular}{rlll}
  \hline
 & $n = 20 $ & $ 50 $ & $ 100 $ \\ 
  \hline
causal Dantzig & 24.05$\pm$90.75 & 0.03$\pm$0 & 0.01$\pm$0 \\ 
  ivreg & 36634.21$\pm$161096.94 & 4244.29$\pm$15557.82 & 1862.7$\pm$8171.14 \\ 
   \hline
\end{tabular}
\caption{Mean square error for varying $n$. The instrument is weak, but \emph{causal Dantzig} can leverage changes in variance.} 
\label{crvsiv2}
\end{table}

\subsection{Causal Dantzig in high dimensions}\label{sec:hidd-high-dimens}

We consider a structural equation model, where the variables
$X_{1},\ldots, X_{p},Y$ form a chain and the distribution of the
unobserved confounder $\eta$ changes between the
environments. The corresponding directed acylic graph is depicted in Figure~\ref{fig:high-dim}.
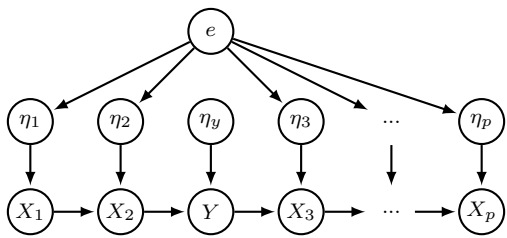
\begin{figure}
\begin{center}
\begin{tikzpicture}[->,>=latex,shorten >=1pt,auto,node distance=1.2cm,
                    thick]
  \tikzstyle{every state}=[draw=black,text=black, inner sep=0.4pt, minimum size=17pt]

  \node[state] (Y) {$Y$};
  \node[state] (etay) [above of=Y] {$\eta_{y}$};
\node[state] (e) [above of=etay] {$e$};
  \node[state] (X2) [left of=Y] {$X_{2}$};
  \node[state] (eta2) [above of=X2] {$\eta_{2}$};
  \node[state] (X1) [left of=X2] {$X_{1}$};
  \node[state] (eta1) [above of=X1] {$\eta_{1}$};
  \node[state] (X3) [right of=Y] {$X_{3}$};
  \node[state] (eta3) [above of=X3] {$\eta_{3}$};
  \node[state,draw=none] (X4) [right of=X3] {$...$};
  \node[state,draw=none] (eta4) [above of=X4] {$...$};

  \node[state] (Xp) [right of=X4] {$X_{p}$};
  \node[state] (etap) [above of=Xp] {$\eta_{p}$};

\draw (e) edge (eta1);
\draw (e) edge (eta2);
\draw (e) edge (eta3);
\draw (e) edge (etap);
\draw (e) edge (eta4);

\draw  (X1)  edge  (X2);
\draw  (X2)   edge  (Y);
\draw  (eta4) edge (X4);
\draw  (eta2)  edge (X2);
\draw  (eta3)  edge (X3);
\draw  (etap)  edge (Xp);
\draw  (eta1)  edge (X1);
\draw (etay) edge (Y);
\draw  (X3)  edge (X4);
\draw  (X4)  edge (Xp);
\draw  (Y)  edge (X3);

\end{tikzpicture}
\caption{\label{fig:high-dim} The directed acylic graph corresponding
  to SEM (C). }
\end{center}
\end{figure}
To be more precise, the distribution of the observed variables $e,X$ and $Y$
is generated according to the following structural equation model:

\begin{align}
 (C):\left\{\begin{array}{rcrrr}
X_1 & \leftarrow  &&& \eta_1 \\
X_2 & \leftarrow   &X_1&+& \eta_2 \\
X_{p+1} = Y & \leftarrow &  X_2&+& \eta_y \\
 X_3 & \leftarrow  & Y &+ &\eta_3 \\
 X_4 & \leftarrow & X_{3}&+ &\eta_4 \\
\vdots && \vdots  &&   \vdots \\
X_{p} & \leftarrow & X_{p-1} &+ &  \eta_{p}
\end{array} \right. \text{, with} \qquad
\begin{split}
  \eta_{k} &= \eta_{k}^{0} +  \delta_{k}^{e} \\
   \delta_{k}^{e} &= \begin{cases}
   0 & e=0 \text{ or } \\ & k=p+1,  \\
   z_{k}  & e=1,
 \end{cases} \\
z_{k} &\sim \mathcal{N}(0,\sigma^{2}) \text{ i.i.d.}, \\
\eta_{k}^{0} &\sim \mathcal{N}(0,1) \text{ i.i.d.}, \\
e &\in \{0,1\}, \\
\end{split}
\end{align}
We assume that $z_{k}$ and $\eta_{k}$ are jointly independent.
The
regularization parameter $\lambda$ is chosen by $10$-fold
cross-validation. Figure~\ref{fig:cool2} shows the regularization path
for two different values of $p$. Figure~\ref{fig:cool3} shows the
regularization path for varying intervention strength $\sigma$. Finally, in
Figure~\ref{fig:cool4} the number of samples collected from each environment $n :=n_{0}=n_{1}$ is varied. In a nutshell, cross-validation seems to select a
reasonable regularization parameter in most cases, estimation
performance deteriorates with increasing $p$, but improves with
increasing $n$ and drastically so with increasing intervention strength $\sigma$.
 \begin{figure}
\begin{center}
\includegraphics[width=0.4\textwidth]{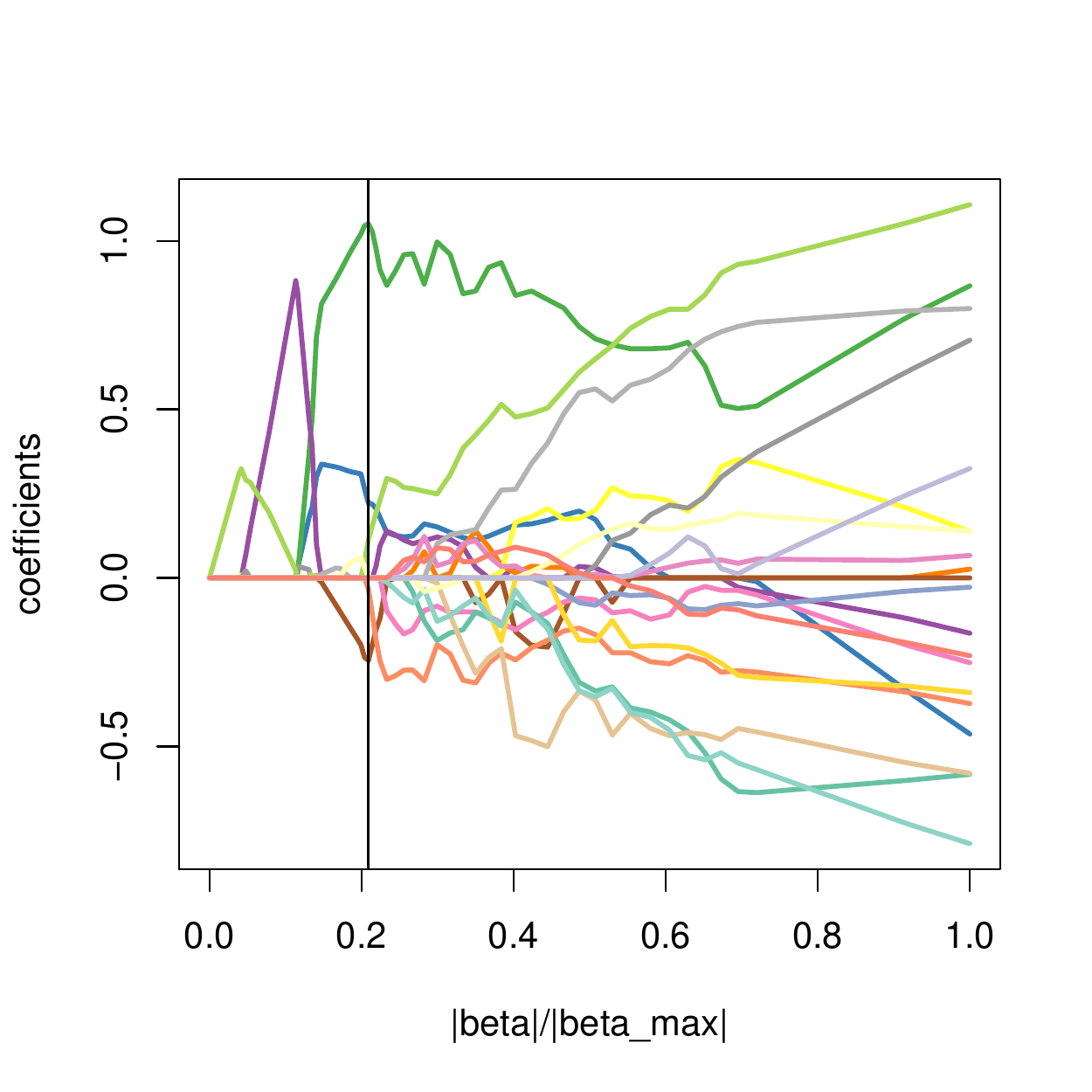}
\includegraphics[width=0.4\textwidth]{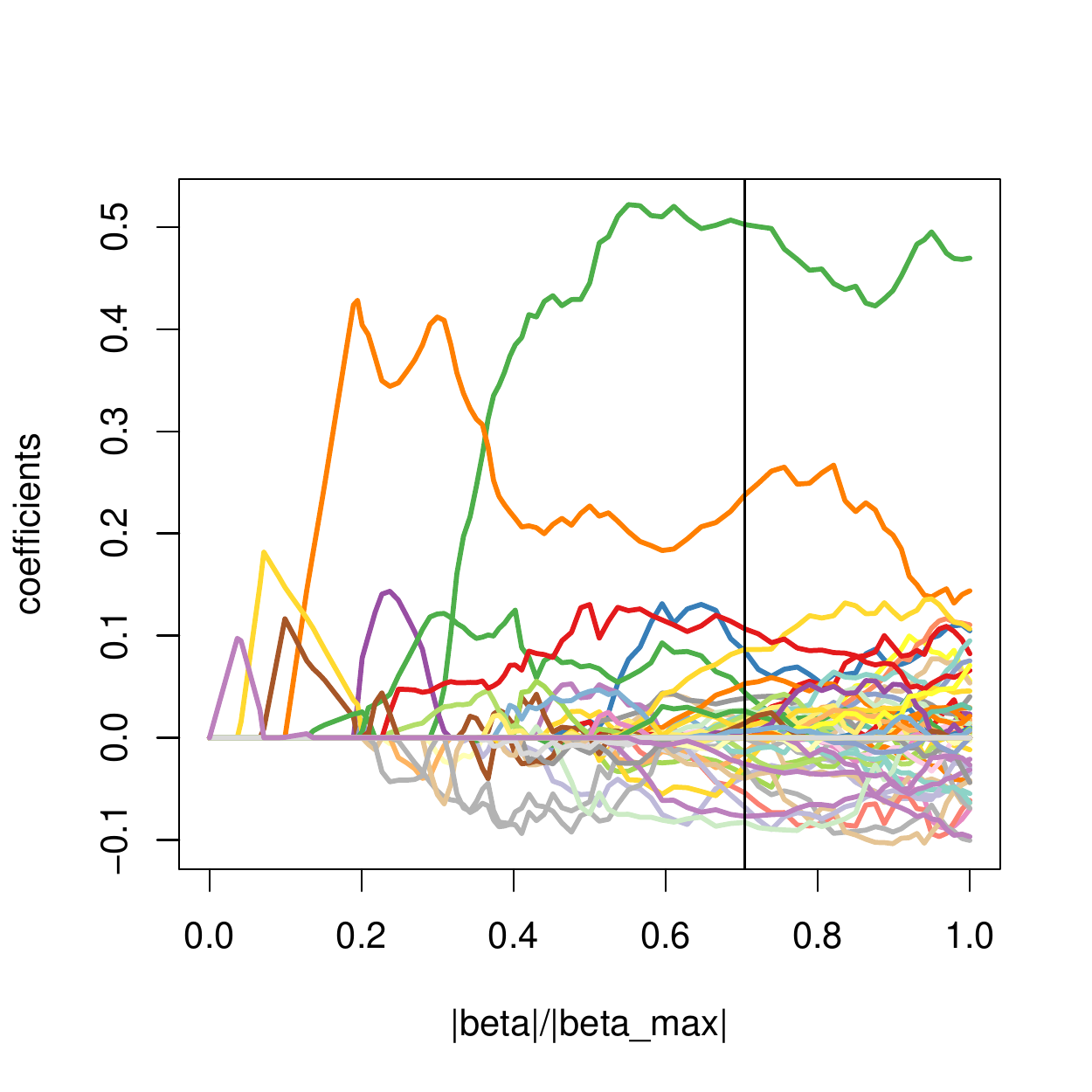}
\caption{ \label{fig:cool2} {\it  Two regularization paths for
    the \emph{causal Dantzig} with $n=30$ and $\sigma=2.5$.  The black vertical line specifies the
    solution chosen by $10$-fold cross-validation. On the left
    $p=20$, on the right $p=200$. The true underlying coefficient is equal to one for one variable (upper green line), and equal to zero for all other variables.
    Though not flawless, cross-validation
    chooses a reasonable regularization parameter in both
    cases.
  }}
\end{center}
\end{figure}

\begin{figure}
\begin{center}
\includegraphics[width=0.4\textwidth]{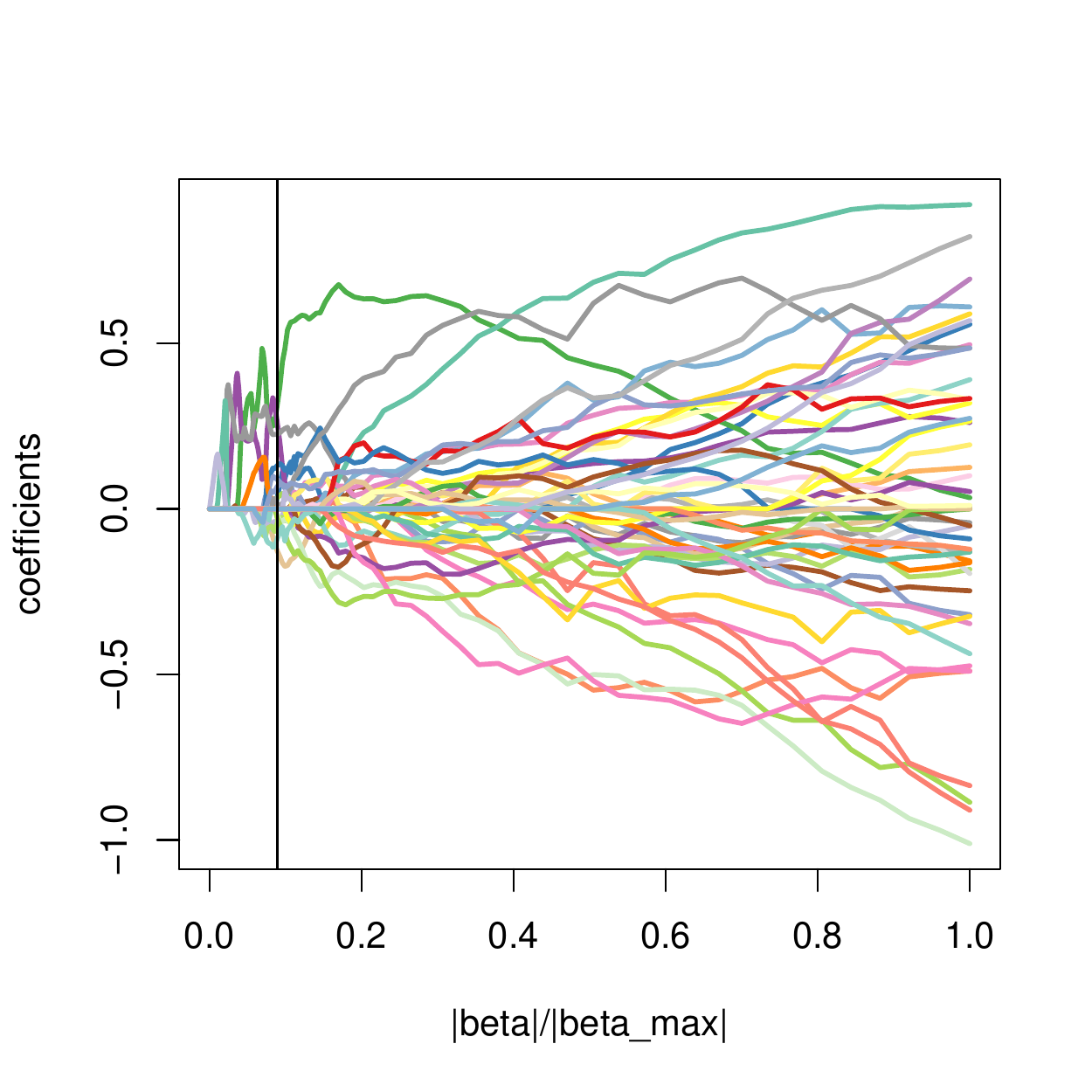}
\includegraphics[width=0.4\textwidth]{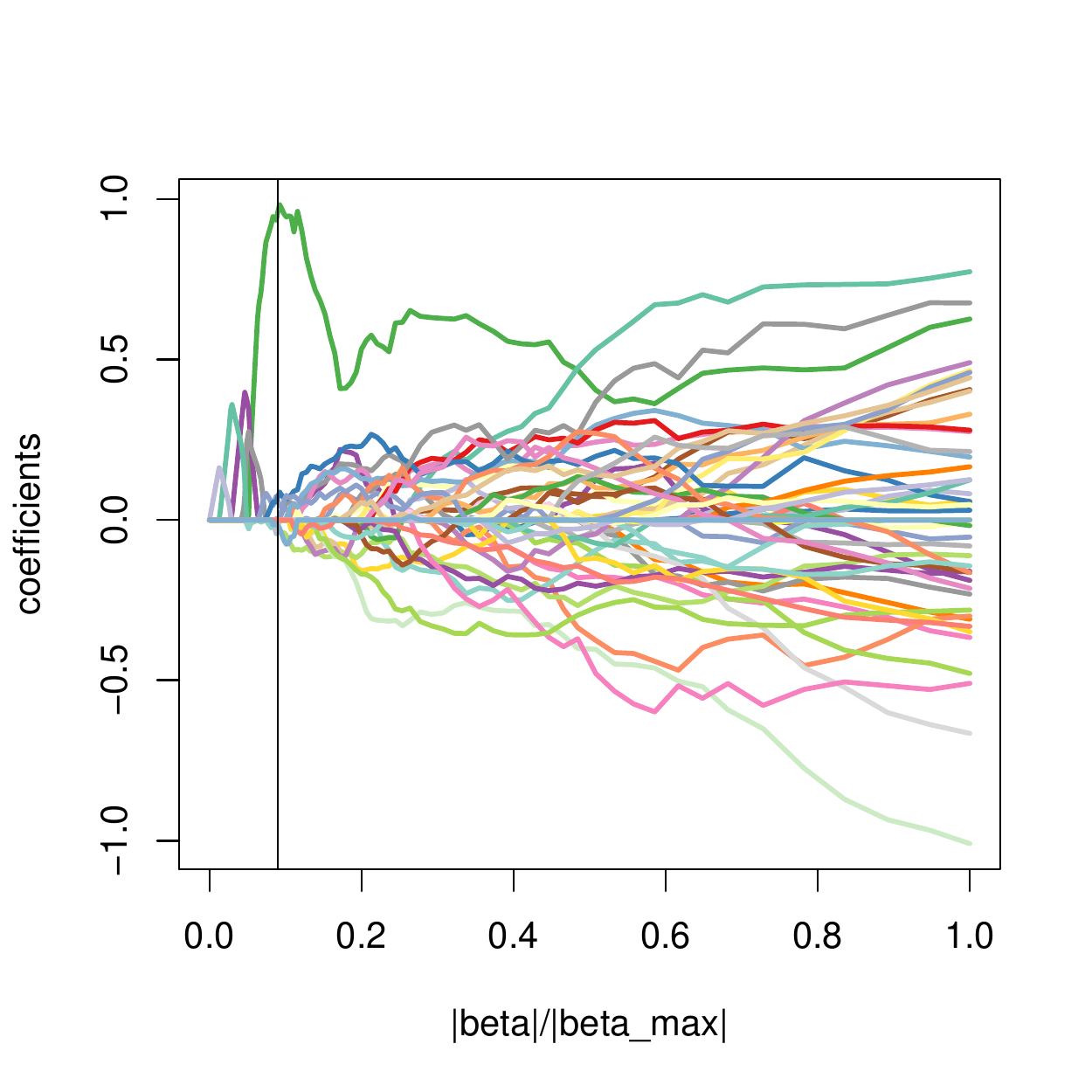}
\caption{ \label{fig:cool3} {\it  Two regularization paths for the
     \emph{causal Dantzig} with $p=30$ and $n=30$.  The black vertical line specifies the
    solution chosen by $10$-fold cross-validation. On the left the
    intervention strength is   $\sigma=2.5$, on the right it is
    $\sigma=3.5$. The true underlying coefficient is equal to one for one variable (upper green line), and equal to zero for all other variables.
Clearly, strong interventions improve estimation performance.}}
\end{center}
\end{figure}

\begin{figure}
\begin{center}
\includegraphics[width=0.4\textwidth]{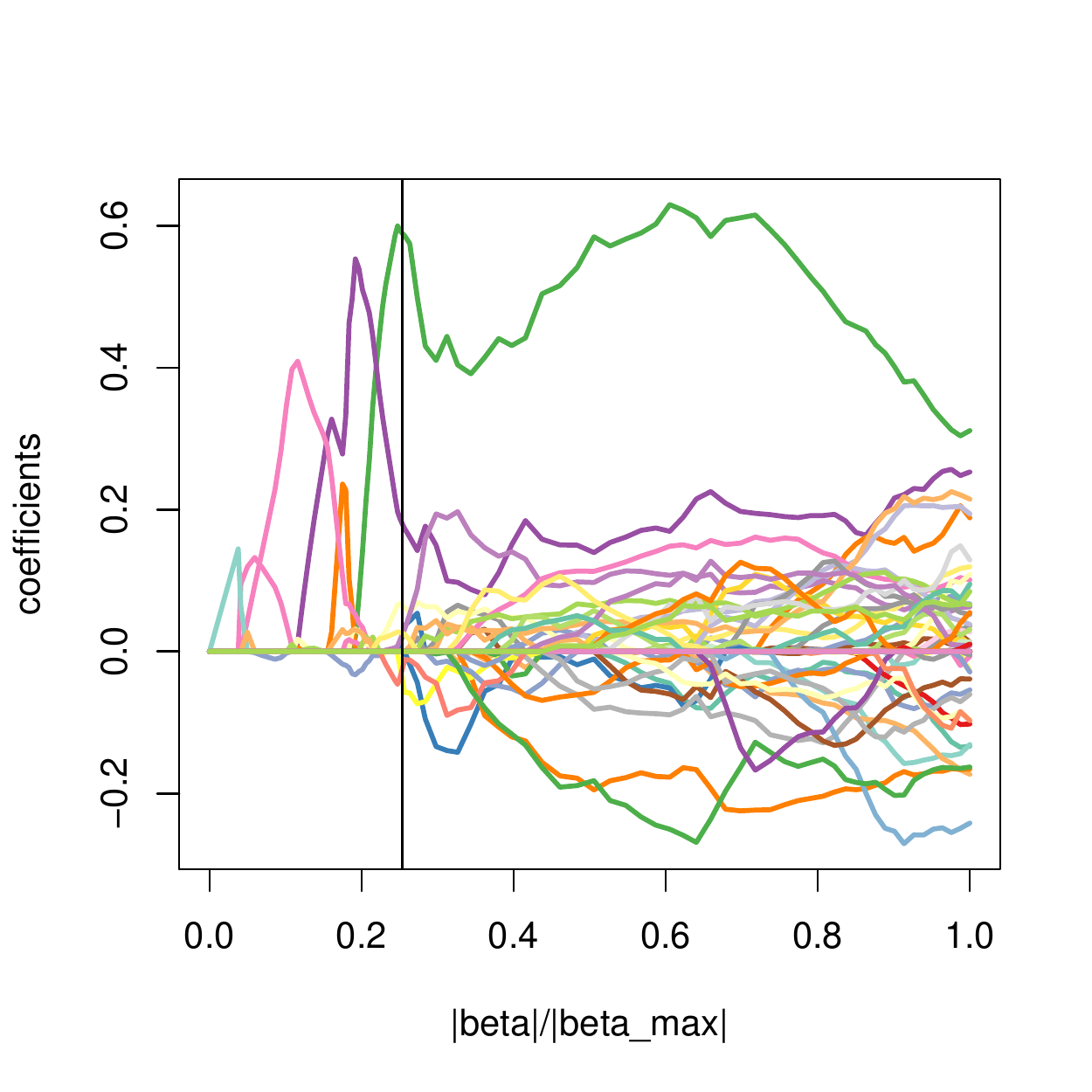}
\includegraphics[width=0.4\textwidth]{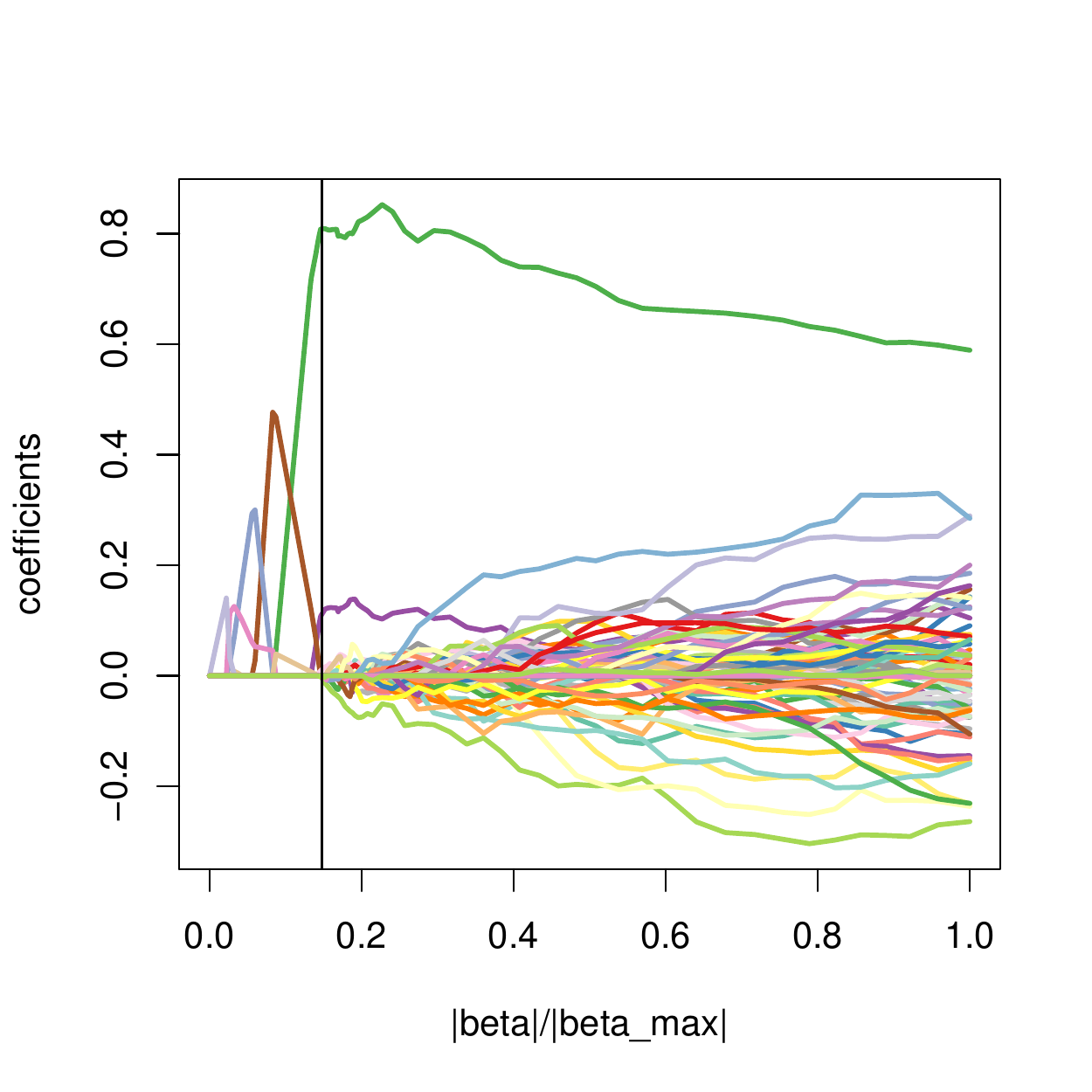}
\caption{ \label{fig:cool4} {\it  Two regularization paths for
    the \emph{causal Dantzig} with $p=100$ and $\sigma=3.5$.  The black vertical line specifies the
    solution chosen by $10$-fold cross-validation. On the left the
    sample size is $n=30$, on the right it is
    $n=60$. The true underlying coefficient is equal to one for one variable (upper
    green line), and equal to zero for all other variables. Estimation performance is clearly better on the right.}}
\end{center}
\end{figure}

\subsection{Gene knockout experiments}\label{sec:gene-knock-exper}

We outline here an application which has  appeared in \citet{meinshausen2016methods}.
The authors consider gene expression  in yeast (Saccharomyces cerevisiae) under deletion of single genes \citep{Kemmeren2014}:  $160$ samples are wild-type (observational); and $1,479$ samples are measured under the deletion of a single gene (intervention). For each of those observations, genome-wide mRNA expression levels  were measured. We denote these measurements by $X_{1},\ldots,X_{p+1}$, where $p+1=6170$. The goal is to predict whether mRNA expression level $Y=X_{p+1}$ changes significantly under a new and unobserved gene-deletion $X_{j}$, $j \neq p+1$. Knocking out a gene is not always successful, and the measured activity of a gene is not constant (or zero) after knocking it out, i.e. the intervention is ``noisy''. Overall, knockouts decrease the activity, which can be interpreted as a negative shift in the measured log-activity of a gene.
\\
The data is split into training and validation data. To this end, the $1,479$ interventional samples are divided into five sets $B_{1},\ldots,B_{5}$. For some $v \in \{1,\ldots,5\}$, the training data consists of the four sets $\{B_{i}\}_{i \in \{1,2,3,4,5\} \setminus \{v\}}$ and the $160$ observational samples.
The samples in  $B_{v}$ are held out for validation.   The interventional effects on the validation set $B_{v}$ were predicted using only training data. This procedure is carried out for all sets $B_{v}, v=1,\ldots,5$, i.e. each gene perturbation is excluded from the training set once. \\
Preselection with the LASSO was used on the pooled data to screen for a superset of the parental set of variable $X_{p+1}$. For some justification of this approach, see Section~\ref{sec:mark-blank-estim-1}. Then, the \emph{causal Dantzig} without regularization was used, with setting $e=1$ for observational data and $e=2$ for interventional training data. Using \emph{causal Dantzig} without screening step is computationally prohibitive due to the large number of variables and as the procedure is repeated for each possible target variable $X_{1},\ldots, X_{p+1}$. The $s$ most often selected intervention predictions were compared to so-called ``strong intervention effects''  (SIEs) as defined in \citet{meinshausen2016methods}. SIEs are computed on the held-out data $B_{v}$ and are a measure for the total causal effect. The results are depicted in Figure~\ref{fig:kemmeren}. As an example, for \emph{causal Dantzig} the four most often selected intervention predictions correspond to SIEs. \\
Screening for  causal effects is a very challenging problem in this setting, mainly due to the high-dimensionality of the dataset and the presence of hidden confounders. The ground truth is not perfectly known but good proxies (strong intervention effects) can be computed on hold-out interventional data.  The strongest discoveries of \texttt{InvariantCausalPrediction} (\texttt{ICP}) and \texttt{causalDantzig} correspond very well to the benchmark. Assuming hidden confounding and shift interventions (\texttt{causalDantzig}) leads to a different ranking of genes compared to assuming the absence of confounding and allowing for arbitrary interventions (\texttt{ICP}). Interestingly, while both methods miss some important variables, making ``wrong'' assumptions such as linearity or absence of latent confounding do not seem to lead to false positives for the first few variables in the ranking. This form of validation and the comparison to other methods are further discussed in
\citet{meinshausen2016methods}. 
\begin{figure}
\begin{center}
\includegraphics[width=0.5\textwidth]{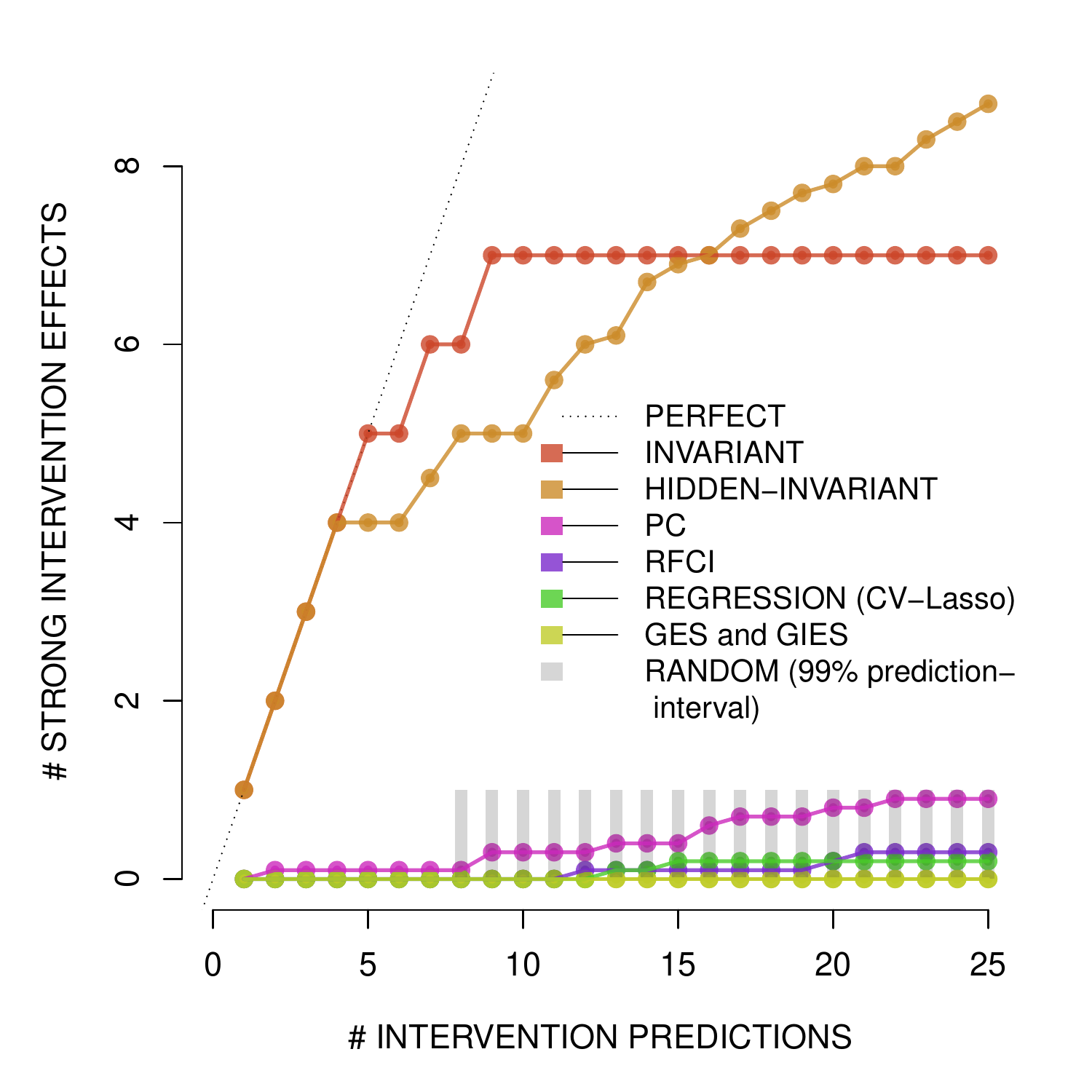}
\caption{ \label{fig:kemmeren} {\it  The results of the gene knockout
    experiments on the Kemmeren dataset \citep{Kemmeren2014}.  This figure has been published in \citet{meinshausen2016methods}. The method \texttt{HIDDEN-INVARIANT} is an unpublished early version of  \texttt{causalDantzig}.}}
\end{center}
\end{figure}

\section{Discussion}

Causal discovery is challenging, particularly in the presence of hidden confounders and feedback loops. However, hidden confounders can rarely be excluded  and feedback loops are to be expected in many real-world applications (e.g., in biological systems). We introduced the notion of \emph{inner-product invariance} and showed that inference in linear structural equation models under \emph{inner-product invariance} is  possible, both for low- and high-dimensional data.

The proposed methods have interesting parallels to widely-used statistical methods. For example, the functional form of the \emph{causal Dantzig} estimator is similar to linear regression. The regularized \emph{causal Dantzig} is similar to the Dantzig selector. For two environments ($|\E |=2$) the \emph{causal Dantzig} estimator can be compared with instrumental variable regression and is consistent in certain settings in which instrumental variable regression fails. Hence, we believe that the \emph{causal Dantzig} will push the boundaries in the analysis of certain types of datasets, in  particular in the analysis of datasets where potentially unknown interventions (or ``perturbations'')  change both the mean and the variance of the observed error distribution. Empirical results show state-of-the-art performance of our proposed estimator on a real-world dataset.

We investigated the identifiability of direct causal effects under the proposed model class. Furthermore, we showed that the regularized \emph{causal Dantzig} can be consistent in the high-dimensional case even if not all covariates have been intervened on.  The estimator can be obtained by solving a linear program and as such is feasible for large-scale causal inference. We derived asymptotic confidence intervals for the unregularized \emph{causal Dantzig}, as well as guarantees for statistical accuracy for the regularized \emph{causal Dantzig}.

The notion of \emph{inner-product invariance} pushes the boundaries for the types of datasets we can leverage for causal discovery. We expect it to be useful for practitioners, in particular as a simple and fast tool for screening for potential direct causal effects. From a theoretical perspective,  the regularized and unregularized \emph{causal Dantzig} provide new perspectives on invariant causal prediction, on the instrumental variable approach and on classical theory for high-dimensional estimation. \\


\newpage

\section{Appendix}

\begin{remark}[Reminder of Assumption~1 and some of the notation]
We assume that the distributions of $(X_{1}^{e},...,X_{p+1}^{e})$, $e \in \E$, are generated by the linear SEM
\begin{align*}
  X_{k}^{e} &\leftarrow \sum_{k' \neq k} A_{k,k'} X_{k'}^{e} + \eta_{k}^{e}, \qquad \text{ for } k=1,\ldots,p+1 \text{ and } e \in \E.
\end{align*}
Assume that there exist random variables $\eta^0,
\delta^e\in \mathbb{R}^p$ with $\mathrm{Cov}(\eta^{0}, \delta^{e}) = 0$ for all $\e\in \E$ such that $\eta^\e$ can be written as
\begin{equation*}
\eta^\e  \stackrel{d}{=} \eta^0  + \delta^\e \qquad\mbox{for all  } \e\in \E.
\end{equation*}
We assume that $\delta^\e_{p+1}\equiv 0$  for all $\e\in\E$ and $\mathbb{E}[\eta^{0}] = 0$.\\

We aim to infer the structural equation for variable $X_{p+1}$, hence we denote it by $Y$. Furthermore, for simplicity we write $\beta^{0} = (A_{p+1,k})_{k=1,\ldots,p}$. The values  $ A_{k,k'}$ form a $ (p+1) \times (p+1)$-dimensional matrix that we denote by $A$.
\end{remark}

\subsection{Proofs for Section~\ref{sec:cond-inner-prod}}

\subsubsection{Proof of Proposition~\ref{proposition:const}}

\begin{proof}
Recall that $Y=X_{p+1}$. We can write equation~\eqref{SEM} under Assumption~\ref{assum:additive} more compactly as $X_{1:(p+1)}^{e} = A X_{1:(p+1)}^{e} + \eta^{e}$, where $A$ is the matrix that contains the structural parameters $A_{k,k'}$. In other words,
\begin{equation*}
  X_{1:(p+1)}^{e} = \left(\text{Id} - A \right)^{-1} \eta^{e}.
\end{equation*}
In the following, we denote the $k$-th unit vector in $\mathbb{R}^{p}$ by $e^{(k)}$, i.e.
\begin{equation*}
  e^{(k)}_{k'} = \begin{cases}
    1 & \text{if } k=k', \\
   0 & \text{else.}
  \end{cases}
\end{equation*}
Recall that $\beta^{0} = (A_{p+1,k})_{k=1,\ldots,p}$ and $Y^{e}=X_{p+1}^{e} $. By Assumption~\ref{assum:additive}, $Y^{e} -  X^{e} \beta^0 = \eta_{p+1}^{e}$. Hence,
\begin{align*}
  X_{k}^{e} (Y^{e} -  X^{e} \beta^0) &=  X_{k}^{e} \left( X_{p+1}^{e} -  \sum_{k' \neq  (p+1)}A_{p+1,k'}X_{k'}^{e} \right) \\
  &= (e^{(k)})^{t}  \left(\text{Id} - A \right)^{-1} \eta^{e} \eta_{p+1}^{e} \qquad \text{ for } k=1,\ldots,p.
\end{align*}
Now we can again use Assumption~\ref{assum:additive}. Recall that $\eta^{e} = \eta^{0} + \delta^{e}$, $\mathbb{E} \eta_{p+1}^{e} = 0$ and that $\eta_{p+1}^{0}$ and $\delta^{e}$ are uncorrelated. Hence for $k=1,\ldots,p$,
\begin{align*}
  \mathbb{E} \left[ X_{k}^{e} \left( X_{p+1}^{e} -  \sum_{k'\neq (p+1)} A_{p+1,k'}X_{k'}^{e} \right)  \right] &= (e^{(k)})^{t} \left(\text{Id} - A \right)^{-1} \mathbb{E} \left[  \eta^{e} \eta_{p+1}^{e}  \right] \\
&= (e^{(k)})^{t}   \left(\text{Id} - A \right)^{-1} \mathbb{E} \left[  ( \eta^{0} + \delta^{e} ) \eta_{p+1}^{0}  \right] \\
&=  (e^{(k)})^{t} \left(\text{Id} - A \right)^{-1} \mathbb{E} \left[   \eta^{0}  \eta_{p+1}^{0}  \right].
\end{align*}
Note that this quantity is the same for all environments $e \in \E$, which concludes the proof.
\end{proof}

\subsubsection{Proof of Proposition~\ref{proposition:errors-variables}}

\begin{proof}
 For all $e,f \in \E$ and $k \in \{1,\ldots,p\}$,
\begin{align*}
\begin{split}
 \mathbb{E} \left[ \tilde X^\e_k ( \tilde Y^\e- \tilde X^\e\beta^{0}) \right] &=  \mathbb{E} \left[ ( X^\e_k  + \zeta_{k}^{e} ) (  Y^\e + \zeta_{y}^{e}- \sum_{k=1}^{p} (X_{k}^\e  + \zeta_{k}^{e}) \beta_{k}^{0} \right] \\
&= \mathbb{E} \left[  X^\e_k  ( Y^\e - \sum_{k=1}^{p} X_{k}^\e \beta^{0}_{k}) -  \sum_{k=1}^{p}  (\zeta_{k}^{e})^{2} \beta_{k}^{0} \right] \\
&= \mathbb{E} \left[  X^f_k  ( Y^f - \sum_{k=1}^{p} X_{k}^f \beta^{0}_{k}) -  \sum_{k=1}^{p}  (\zeta_{k}^{f})^{2} \beta_{k}^{0}\right] \\
&=  \mathbb{E} \left[ ( X^f_k  + \zeta_{k}^{f} )(  Y^f + \zeta_{y}^{f}- \sum_{k=1}^{p} (X_{k}^f  + \zeta_{k}^{f}) \beta_{k}^{0} ) \right] \\
&= \mathbb{E}\left[ \tilde X^f_k ( \tilde Y^f- \tilde X^f\beta^{0})\right].
\end{split}
\end{align*}
In the first line and third line we used that $\zeta_{1}^{e},\ldots,\zeta_{k}^{e},\zeta_{y}^{e}$ are centered and jointly independent for all $e \in \E$. In the second line we used that we have inner product invariance for $(X^{e},Y^{e})$, $e \in \E$ under $\beta^{0}$ and that $\mathbb{E}[(\zeta_{k}^{e})^{2}]=\mathbb{E}[(\zeta_{k}^{f})^{2}]$ for all $e,f \in \E$ and $k=1,\ldots,p$. This proves that we also have inner-product invariance for $(\tilde X^{e}, \tilde Y^{e}), e \in \E$ under $\beta^{0}$.

\end{proof}

\subsection{Proofs for Section~\ref{sec:causal-dantzig}}

\subsubsection{Proof of Theorem~\ref{theorem:confidence-intervals}}

\begin{proof}
First note  that $V^{1}$ and $V^{2}$ are invertible as $\G$ and the covariance matrix of $(X^{e})^{t} \eta_{p+1}^{e}$, $e \in \{1,2\}$ are assumed to be invertible. Now note that by inner-product invariance of $(X^{e}, Y^{e})$ under $\beta^{0}$ we have $\G^{-1} \Z = \beta^{0}$ and hence
\begin{equation*}
   (X^{e})^{t} \eta_{p+1}^{e}=  - (X^{e})^{t} X^{e} \G^{-1} \Z + (X^{e})^{t} Y^{e} \text{ for } e \in \{1,2\}.
\end{equation*}
In particular,
\begin{equation}\label{eq:43}
  \G^{-1} (X^{e})^{t} \eta_{p+1}^{e}=  - \G^{-1} (X^{e})^{t} X^{e} \G^{-1} \Z + \G^{-1} (X^{e})^{t} Y^{e} \text{ for } e \in \{1,2\}.
\end{equation}
We denote $\mbox{GL}_{p}$ the set of real-valued invertible $p \times p$ matrices. Define the function $ f : \mbox{GL}_p \times \mathbb{R}^p  \rightarrow
\mathbb{R}^p $  by
\begin{equation*}
  f(\tilde \G,\tilde \Z) := \tilde \G^{-1} \tilde \Z.
\end{equation*}
By elementary matrix algebra, this function is continuously differentiable with
derivative in direction  $(D,d) \in \mathbb{R}^{p \times p} \times \mathbb{R}^p$
\begin{equation*}
  D_{(D,d)}f(\tilde \G, \tilde \Z) =-  \tilde \G^{-1} D \tilde \G^{-1}\tilde \Z+ \tilde \G^{-1}d.
\end{equation*}
As $( \hat \G, \hat \Z ) -\left( \G, \Z\right) = \mathcal{O}_{P}\left( \max \left( \frac{1}{\sqrt{n_{1}}} , \frac{1}{\sqrt{n_{2}}}  \right) \right)$ and $\beta^{0} = f(\G, \Z)$, the delta method yields
\begin{align*}
   &\left( \frac{V^{1}}{n_{1}} + \frac{V^{2}}{n_{2}} \right)^{-\frac{1}{2}} \left( \hat \beta - \beta^{0} \right) \\
=&\left( \frac{V^{1}}{n_{1}} + \frac{V^{2}}{n_{2}} \right)^{-\frac{1}{2}} \left( f(\hat \G, \hat \Z) - f(\G, \Z) \right) \\
=& \left( \frac{V^{1}}{n_{1}} + \frac{V^{2}}{n_{2}} \right)^{-\frac{1}{2}}
    \left(D_{(\hat \G -\G,\hat \Z - \Z)}f(\G,\Z)+\scriptstyle\mathcal{O}
    \textstyle_P \left(\max \left( \frac{1}{\sqrt{n_{1}}} , \frac{1}{\sqrt{n_{2}}}  \right) \right) \right) \\
  =& \left( \frac{V^{1}}{n_{1}} + \frac{V^{2}}{n_{2}} \right)^{-\frac{1}{2}} \left( - \G^{-1}( \hat \G -\G) \G^{-1}\Z+
    \G^{-1} (\hat \Z - \Z) \right)+\scriptstyle\mathcal{O}
    \textstyle_P (1) \\
\rightharpoonup &\mathcal{N}(0,\text{Id}).
\end{align*}
In the last line we used independence of the samples of environment $e=1$ and $e=2$, the CLT and the definition of $V^{1}$ and $V^{2}$ together with equation~\eqref{eq:43}.

\end{proof}

\subsubsection{Proof of Theorem~\ref{theorem:popul-g-invert}}

\begin{proof}
\textbf{Part A:} In this part we prove claim 2 and $"\Leftarrow"$ of claim 1.  By Proposition~\ref{proposition:const} we have inner-product invariance for $\beta^{0}$ and hence the solution set of the population \emph{causal Dantzig} contains $\beta^{0}$.  Assume that for each $k$ there exists $e \in \E$ such that $\delta_{k}^{e} \not \equiv 0$. We want to show that under this assumption the \emph{causal Dantzig} is unique in the population case. By Proposition~\ref{proposition:const} we have inner-product invariance for $\beta^{0}$ and hence each solution $\beta^{*}$ to the population \emph{causal Dantzig} satisfies
\begin{equation*}
  \max_{e \in \E} \| \mathbb{E}[\hat \Z^{e}] - \mathbb{E}[\hat \G^{e}] \beta^{*}] \|_{\infty} =0.
\end{equation*}
Denote $\tilde e$ the ``observational'' environment, i.e. the environment with $\delta_{k}^{\tilde e} \equiv 0$ for $k=1,\ldots,p$. By inner-product invariance under $\beta^{0}$,
\begin{equation*}
  \mathbb{E}[\hat \Z^{\tilde e}] - \mathbb{E}[\hat \G^{\tilde e}] \beta^{*}  = 0 =   \mathbb{E}[\hat \Z^{\tilde e}] - \mathbb{E}[\hat \G^{\tilde e}] \beta^{0}.
\end{equation*}
By rearranging,
\begin{equation*}
  \mathbb{E}[\hat \G^{\tilde e}] \left( \beta^{*} - \beta^{0}\right) = 0.
\end{equation*}
As we want to show $\beta^{*} = \beta^{0}$ it suffices to show that $ \mathbb{E} [ \hat \G^{\tilde e}]$ is invertible.
In the following, for notational brevity we write $(\mbox{Id}-A)_{1:p,1:p}^{-t}$ instead of  $\left((\mbox{Id}-A)^{-t} \right)_{1:p,1:p}$ and $(\mbox{Id}-A)_{1:p,1:p}^{-1}$  instead of  $\left((\mbox{Id}-A)^{-1} \right)_{1:p,1:p}$.
By definition, we have
\begin{align*}
  \mathbb{E}[\G^{\tilde e}] &= \mathbb{E} \left[  \left( X^{\tilde e} \right)^{t} X^{\tilde e} - \frac{1}{|\E|-1} \sum_{e \neq \tilde e} \left(  X^{e}\right)^{t} X^{e} \right] \\
&= \left( (\mbox{Id}-A)^{-1} \right)_{1:p,\bullet}  \mathbb{E} \left[    \left(  \delta^{\tilde e}  \right)^{t}  \delta^{\tilde e} - \frac{1}{|\E|-1}  \sum_{e \neq \tilde e}\left(  \delta^{e}   \right)^{t} \delta^{e}                     \right] \left((\mbox{Id}-A)^{-t} \right)_{\bullet,1:p}  \\
&= (\mbox{Id}-A)_{1:p,1:p}^{-1}  \mathbb{E} \left[    \left(  \delta_{1:p}^{\tilde e}  \right)^{t}  \delta_{1:p}^{\tilde e} - \frac{1}{|\E|-1}  \sum_{e \neq \tilde e}\left(  \delta_{1:p}^{e}   \right)^{t} \delta_{1:p}^{e}                     \right] (\mbox{Id}-A)_{1:p,1:p}^{-t}.
\end{align*}
In the last line we used $\delta_{p+1}^{e} \equiv 0$ for all $e \in \E$. As setting $\tilde e$ is ``observational'', i.e. $\delta^{\tilde e} \equiv 0$,
\begin{align}\label{eq:27}
  \mathbb{E}[\G^{\tilde e}] = \; &  (\mbox{Id}-A)_{1:p,1:p}^{-1}  \mathbb{E} \left[    \left(  \delta_{1:p}^{\tilde e}  \right)^{t}  \delta_{1:p}^{\tilde e} - \frac{1}{|\E|-1}  \sum_{e \neq \tilde e}\left(  \delta_{1:p}^{e}   \right)^{t} \delta_{1:p}^{e}                     \right] (\mbox{Id}-A)_{1:p,1:p}^{-t} \nonumber \\
= \; &  (\mbox{Id}-A)_{1:p,1:p}^{-1}  \mathbb{E} \left[      - \frac{1}{|\E|-1}  \sum_{e \neq \tilde e}\left(  \delta_{1:p}^{e}   \right)^{t} \delta_{1:p}^{e}  \right] (\mbox{Id}-A)_{1:p,1:p}^{-t}
\end{align}
Now we want to show that $ (\mbox{Id}-A)_{1:p,1:p}^{-1}$ is invertible. If this is not the case then there exists $\gamma \in \mathbb{R}^{p} \setminus \{0\} $ such that $\gamma^{t}  (\mbox{Id}-A)_{1:p,1:p}^{-1} = 0$. As $ (\mbox{Id}-A)^{-1}$ is invertible,
\begin{align}\label{eq:39}
\begin{split}
0 \neq \; &  (\gamma,0)^{t}  (\mbox{Id}-A)^{-1} \\
= \; &\left(0,\ldots,0, (\gamma,0)^{t}  (\mbox{Id}-A)_{1:(p+1),p+1}^{-1} \right)  \\
= \; & (0,\ldots,0, \gamma^{t} (\mbox{Id}-A)_{1:p,p+1}^{-1})
\end{split}
\end{align}
In particular, $ \gamma^{t} (\mbox{Id}-A)_{1:p,p+1}^{-1} \neq 0$. As $(X^{e},Y^{e})^{t} =  (\mbox{Id}-A)^{-1} ( \eta^{e})^{t}$, by equation~\eqref{eq:39} we have $\gamma^{t} (X^{e})^{t} =   \gamma^{t} (\mbox{Id}-A)_{1:p,p+1}^{-1}   \eta_{p+1}^{e}$. As $Y^{e}= X_{p+1}^{e}= X^{e} \beta^{0} + \eta_{p+1}^{e}$,
\begin{align*}
  Y^{e} &= X^{e} \beta^{0} + \eta_{p+1}^{e} \\
&= X^{e} \beta^{0} + X^{e} \frac{\gamma}{ \gamma^{t} (\mbox{Id}-A)_{1:p,p+1}^{-1} } \\
&= X^{e} \left( \beta^{0}  + \frac{\gamma}{ \gamma^{t} (\mbox{Id}-A)_{1:p,p+1}^{-1} } \right).
\end{align*}
As we assumed that the Gram matrix of $(X^{e},Y^{e})$ is positive definite for all $e \in \E$, this is a contradiction.  Hence  $ (\mbox{Id}-A)_{1:p,1:p}^{-1}$ is invertible.
Thus the matrix in equation~\eqref{eq:27} is invertible if and only if
\begin{equation*}
  \mathbb{E} \left[\frac{1}{|\E|-1}  \sum_{e \neq \tilde e}\left(  \delta_{1:p}^{e}   \right)^{t} \delta_{1:p}^{e} \right]
\end{equation*}
is invertible. Let $\xi \in \mathbb{R}^{p} \setminus \{0\}$ such that  $  \xi^{t}  \mathbb{E} \left[\sum_{e \neq \tilde e} \left(  \delta_{1:p}^{e}   \right)^{t} \delta_{1:p}^{e}  \right] \xi = 0$. We will lead this to a contradiction. As all matrices $\mathbb{E} \left[\left(\delta_{1:p}^{e} \right)^{t} \delta_{1:p}^{e} \right]$,  $e \neq \tilde e$ are positive semi-definite we have
\begin{equation*}
  \xi^{t} \mathbb{E} \left[\left(\delta_{1:p}^{e} \right)^{t} \delta_{1:p}^{e} \right] \xi =0 \qquad \text{ for all } e \neq \tilde e.
\end{equation*}
As $\xi \neq 0$ there exists $k$ such that $\xi_{k} \neq 0$. Fix such a $k$. By assumption there exists $e \neq \tilde e$ such that $\delta_{k}^{e} \not \equiv 0$. Fix such an environment $e$. Define $S = \{ s : 1 \le s \le p \text{ such that } \delta_{s}^{e} \not \equiv 0\}$, the support of $\delta_{1:p}^{e}$. By definition,
\begin{equation*}
 \xi^{t} \mathbb{E} \left[\left(\delta_{1:p}^{e} \right)^{t} \delta_{1:p}^{e} \right] \xi=  \xi_{S}^{t}\mathbb{E} \left[\left(\delta_{S}^{e} \right)^{t} \delta_{S}^{e} \right] \xi_{S}.
\end{equation*}
But by assumption, the matrix  $\mathbb{E} \left[\left(\delta_{S}^{e} \right)^{t} \delta_{S}^{e} \right]$ is positive definite. Hence
\begin{equation*} \xi_{S}^{t} \mathbb{E} \left[\left(\delta_{S}^{e} \right)^{t} \delta_{S}^{e} \right] \xi_{S} >0. \end{equation*}
Contradiction! Thus $\mathbb{E} \left[\hat \G^{e}\right]$ is invertible and $\beta^{*}=\beta^{0}$. This concludes the proof of part A.\\

\textbf{Part B:}
In this part we prove $"\Rightarrow"$ of claim 1. Proof by contradiction. Assume there exists a $k$ such that $\delta_{k}^{e} \equiv 0$ for all $e \in \E$. We want to show that there exists a second SEM with $\tilde \beta^{0} \neq \beta^{0}$ that satisfies Assumption~\ref{assum:additive} and generates the distributions  of $(X^{e}, Y^{e}), e \in \E$. Fix a $k$ such that $\delta_{k}^{e} \equiv 0$ for all $e \in \E$. As above, it is possible to show that  for all $ \tilde e \in \E$,
\begin{align}\label{eq:32}
 \mathbb{E}[\G^{\tilde e}] =    (\mbox{Id}-A)_{1:p,1:p}^{-1}  \mathbb{E} \left[    \left(  \delta_{1:p}^{\tilde e}  \right)^{t}  \delta_{1:p}^{\tilde e} - \frac{1}{|\E|-1}  \sum_{e \neq \tilde e}\left(  \delta_{1:p}^{e}   \right)^{t} \delta_{1:p}^{e}                     \right] (\mbox{Id}-A)_{1:p,1:p}^{-t} .
\end{align}
As $\delta_{k}^{e} \equiv 0$ for all $e \in \E$ there exists $\Delta \in \mathbb{R}^{p} \setminus \{0\}$ such that $\mathbb{E}[\G^{\tilde e}] \Delta = 0$ for all $\tilde e \in \E$.
Now we want to show that there exists a SEM with $\tilde \beta^{0} \neq \beta^{0} $ that generates the distributions of $(X^{e},Y^{e}), e \in \E$ and satisfies Assumption~\ref{assum:additive}. For $X_{1:p}$ we keep the structural equations $\tilde A_{1:p,\bullet} = A_{1:p,\bullet}$. For the variable $Y$ we define the new structural equation
\begin{equation*}
\tilde A_{p+1,\bullet} := \begin{pmatrix} \beta^{0} + \gamma \Delta \\ 0 \end{pmatrix}^{t},
\end{equation*}
where we choose $\gamma$ small enough to make $\text{Id} - \tilde A $ invertible. \\ Furthermore define
$ (\tilde \eta^{0})^{t} = \left( \text{Id} - \tilde A \right) \left( \text{Id} - A \right)^{-1} (\eta^{0})^{t}$ and $\tilde  \delta^{e} =  \delta^{e}$. Note that this SEM still satisfies that one environment $e$ is ``observational'', i.e. $\tilde \delta^{e} \equiv 0$ and that all interventions $\tilde \delta^{e}$ are full-rank on its support as the same holds true for $\delta^{e}$. Now we want to show that this SEM satisfies inner-product invariance under $\tilde \beta^{0} = \beta^{0} + \gamma \Delta$. By inner-product invariance under $\beta^{0}$, and as $\mathbb{E}[\G^{e}] \Delta =0$ for all $e \in \E$,
\begin{align*}
 \mathbb{E}[  \Z^{e}] - \mathbb{E}[ \G^{e} ] \tilde \beta^{0}  &= \mathbb{E}[  \Z^{e}] - \mathbb{E}[ \G^{e} ] (\beta^{0} + \gamma \Delta) \\
&=  \mathbb{E}[  \Z^{e}] - \mathbb{E}[ \G^{e} ] \beta^{0} \\
&=  \mathbb{E}[  \Z^{e'}] - \mathbb{E}[ \G^{e'} ] \beta^{0} \\
&=  \mathbb{E}[  \Z^{e'}] - \mathbb{E}[ \G^{e'} ] (\beta^{0} + \gamma \Delta) \\
&=  \mathbb{E}[  \Z^{e'}] - \mathbb{E}[ \G^{e'} ] \tilde \beta^{0}
\text{ for all $e,e' \in \E$.}
\end{align*}
Hence we also have inner-product invariance of $(X^{e}, Y^{e}), e \in \E$ under $\tilde \beta^{0}$. Now we want to show that the new SEM generates the distributions of $(X^{e},Y^{e}), e \in \E$, i.e. we want to show that
\begin{equation}\label{eq:28}
  (X^{e},Y^{e})^{t} \overset{!}{=} \left( \text{Id} - \tilde A\right)^{-1} \left( \tilde \eta^{0} + \tilde \delta^{e} \right)^{t}.
\end{equation}
By definition,
\begin{align*}
(X^{e},Y^{e})^{t} = \left( \text{Id} -  A\right)^{-1} \left(  \eta^{0} + \delta^{e} \right)^{t},
\end{align*}
and again by definition we know $\tilde \eta^{0} = \left( \text{Id} - \tilde A \right) \left( \text{Id} - A \right)^{-1} \left( \eta^{0} \right)^{t}$. Hence to prove equation~\eqref{eq:28} it suffices to show
\begin{equation*}
  \left( \text{Id} -  A\right)^{-1} \left( \delta^{e} \right)^{t} \overset{!}{=} \left( \text{Id} - \tilde A\right)^{-1} \left( \tilde \delta^{e} \right)^{t}.
\end{equation*}
As we defined $\tilde \delta^{e} := \delta^{e}$ it suffices to show
\begin{equation*}
  \left( \text{Id} -  A\right)^{-1} \left( \delta^{e} \right)^{t} \overset{!}{=} \left( \text{Id} - \tilde A\right)^{-1} \left(  \delta^{e} \right)^{t}.
\end{equation*}
Rearranging yields
\begin{equation}\label{eq:29}
 \left( \text{Id} - \tilde A\right) \left( \text{Id} -  A\right)^{-1} \left( \delta^{e} \right)^{t} \overset{!}{=}  \left(  \delta^{e} \right)^{t}.
\end{equation}
We know that there exists $\tilde e$ such that $\delta^{\tilde e} \equiv 0$. Using equation~\eqref{eq:27} we obtain
\begin{equation}\label{eq:30}
 \mathbb{E}[\G^{\tilde e}] =    (\mbox{Id}-A)_{1:p,\bullet}^{-1}  \mathbb{E} \left[  - \frac{1}{|\E|-1}  \sum_{e \neq \tilde e}\left(  \delta^{e}   \right)^{t} \delta^{e}                     \right] (\mbox{Id}-A)_{\bullet,1:p}^{-t}.
\end{equation}
By construction $\mathbb{E}[ \G^{\tilde e}] \Delta = 0$, which implies that $\Delta^{t} \mathbb{E}[\G^{\tilde e}] \Delta = 0$. Combining this fact with equation~\eqref{eq:30} yields
\begin{equation*}
  \delta^{e}   (\mbox{Id}-A)^{-t} \begin{pmatrix}\Delta \\ 0 \end{pmatrix} = 0 \qquad \text{ for all } e \in \E.
\end{equation*}
Equivalently,
\begin{equation*}
 (\Delta^{t},0)    (\mbox{Id}-A)^{-1} ( \delta^{e})^{t} = 0.
\end{equation*}
Now we can prove equation~\eqref{eq:29}:
\begin{align*}
 \left( \text{Id} - \tilde A\right) \left( \text{Id} -  A\right)^{-1} \left( \delta^{e} \right)^{t} &= \left(\text{Id} - A - \begin{pmatrix}
  0_{p,p} & 0_{p,1} \\
  \gamma \Delta^{t} & 0 \\
\end{pmatrix} \right) \left( \text{Id} -  A\right)^{-1} \left( \delta^{e} \right)^{t} \\
&= \left( \text{Id} - A\right) \left(\text{Id} - A \right)^{-1} ( \delta^{e})^{t} \\
&=\left( \delta^{e} \right)^{t}.
\end{align*}
This proves equation~\eqref{eq:29} and hence the new SEM generates $(X^{e},Y^{e}), e \in \E$. Hence $\beta^{0}$ is not identifiable. This concludes the proof of part B.

\end{proof}

\subsubsection{Proof of Theorem~\ref{thm:partiali}}

\begin{proof}
It suffices to show that
\begin{equation*}
  Y^{e} - X^{e} \beta =   \eta_{1:p}^{0} (\mathrm{Id} - A)_{1:p,p+1}^{-t}  -  \eta_{1:p}^{0} (\mathrm{Id} - A)_{1:p,1:p}^{-t} \beta
\end{equation*}
for all $e \in \E \cup \{ \tilde e\}$ as the distribution on the right hand side of the data is the same across all environments $ e \in \E \cup \{\tilde e\}$. By Assumption~\ref{assum:additive},
\begin{equation}\label{eq:40}
  X^{e} = (\eta_{1:p}^{0} + \delta^{e})    (\mathrm{Id} - A)_{1:p,1:p}^{-t} \text{ for all } e \in \E \cup \{ \tilde e\}
\end{equation}
and hence
\begin{equation}\label{eq:41}
    Y^{e} - X^{e} \beta = (\eta_{1:p}^{0} + \delta^{e}) \left(   (\mathrm{Id} - A)_{1:p,p+1}^{-t}  -  (\mathrm{Id} - A)_{1:p,1:p}^{-t} \beta \right) \text{ for all } e \in \E \cup  \{ \tilde e \}.
\end{equation}
Hence it suffices to show that
\begin{equation}\label{eq:38}
     (\mathrm{Id} - A)_{k,p+1}^{-t}  -  (\mathrm{Id} - A)_{k,1:p}^{-t} \beta = 0
\end{equation}
for all $k \in \cup_{e \in \E} \{k': \delta_{k'}^{e} \not \equiv 0 \} $. To this end, let $e'$ denote the observational environment, i.e. the environment $e' \in \E$ with $\delta^{e'} \equiv 0$. By Proposition~\ref{proposition:const},
\begin{equation*}
\mathbb{E} [  \Z^{e}]  - \mathbb{E}[ \G^{e} ] \beta^{0} = 0 \text{ for all } e \in \E.
\end{equation*}
Hence also
\begin{equation}\label{eq:42}
 \mathbb{E} [  \Z^{e'}]  - \mathbb{E} [ \G^{e'} ] \beta = 0.
\end{equation}
Using equation~\eqref{eq:40} and equation~\eqref{eq:41}, with $\tilde \beta := (\mathrm{Id} - A)_{1:p,p+1}^{-t}-  (\mathrm{Id} - A)_{1:p,1:p}^{-t} \beta $, equation~\eqref{eq:42} is equivalent to
\begin{equation*}
  ( \mathrm{Id} - A)_{1:p,1:p}^{-1}   \sum_{e \in \E , e \neq e'}  \mathbb{E} [ ( \delta_{1:p}^{e})^{t} \delta_{1:p}^{e} ] \tilde \beta = 0.
\end{equation*}
As shown in the proof of Theorem~\ref{theorem:popul-g-invert}, $  ( \mathrm{Id} - A)_{1:p,1:p}^{-1} $ is invertible. Hence the preceding equation is equivalent to
\begin{equation*}
  \sum_{e \in \E , e \neq e'}  \mathbb{E} [ ( \delta_{1:p}^{e}  )^{t} \delta_{1:p}^{e} ] \tilde \beta = 0
\end{equation*}
Analogously as in the proof of Theorem~\ref{theorem:popul-g-invert} we can use positive definiteness of $\mathbb{E} [ (\delta_{S^{e}}^{e} )^{t} \delta_{S^{e}}^{e} ]$ to conclude that $\tilde \beta_{k} \equiv 0 $ for all $k \in S^{e}, e \in \E, e \neq e'$. As
\begin{equation*}
  \beta_{k} =      (\mathrm{Id} - A)_{k,p+1}^{-t}  -  (\mathrm{Id} - A)_{k,1:p}^{-t} \beta = 0,
\end{equation*}
we proved equation~\eqref{eq:38}, which concludes the proof.
\end{proof}

\subsubsection{Proof of Proposition~\ref{prop:inner-prod-invarpot}}

\begin{proof} First, recall that assumption (A4) says that
\begin{equation}\label{eq:2}
   \mathbb{E}[Y(X) - Y(0) | X=x,E=e] = \mathbb{E}[Y(x) - Y(0) | X=x,E=e] = x \beta^{0}.
\end{equation}
We have
\begin{align*}
&  \mathbb{E}[ X^{t} (Y - X \beta^{0}) | E=e] \\
= \; &   \mathbb{E}[ X^{t} \mathbb{E}[(Y - X \beta^{0}) | X, E = e ] | E=e] \\
= \; &   \mathbb{E}[ X^{t} (\mathbb{E}[(Y(X) - Y(0) | X,E=e] + Y(0)- X \beta^{0})  | E=e] \\
= \; &  \mathbb{E}[ X^{t} (  X \beta^{0} + Y(0)- X \beta^{0})  | E=e] \\
= \; &  \mathbb{E}[ X^{t}  Y(0)  | E=e] \\
= \;&  \mathbb{E}[ X(e)^{t}  Y(0)  | E=e] \\
= \;& \mathbb{E}[ X(e)^{t}  Y(0)  ].
\end{align*}
In the second line we used (A1). In the fourth line we used equation~\eqref{eq:2}. In the last line we used (A2).
By assumption (A5) we have $\mathbb{E}[Y(0)] = \mathbb{E}[Y(0)-Y(X)] = \mathbb{E}[ \mathbb{E}[ Y(0) - Y(X)  |X,E]] 
= \mathbb{E}[X] \beta^{0} = 0$ and hence
\begin{equation*}
  \mathbb{E}[ X^{t} (Y - X \beta^{0}) | E=e] = \mathbb{E}[ X(e)^{t}  Y(0)  ]= \text{Cov}(X(e) , Y(0)).
\end{equation*}
Using assumption (A3) concludes the proof.
\end{proof}

\subsection{Proofs for Section~\ref{sec:high-dim}}

\subsubsection{Proof of Lemma~\ref{lemma:finite-sample-bound}}

\begin{proof}
The proof follows the technique used in  \cite{ye2010rate}. As $z^{*} \le \lambda$, $\beta^{0} \in \{ \beta:  \| \hat \Z - \hat \G
\beta\|_{\infty} \le \lambda  \}$. By definition of $\hat \beta^{\lambda}$, we have $\|\hat \beta^{\lambda} \|_{1} \le \| \beta^{0} \|_{1} $. As the active set of $\beta^{0}$ is $S$ we have $ \| (\hat \beta^{\lambda} - \beta^{0})_{S^{c}} \|_{1} = \| \hat \beta^{\lambda} \|_{1} - \| \hat \beta^{\lambda}_{S} \|_{1} \le \| \beta^{0} \|_{1} - \| \hat \beta^{\lambda}_{S} \|_{1} \le \| (\hat \beta^{\lambda} - \beta^{0})_{S} \|_{1}$. Hence we can invoke the definition of $\mathrm{CCIF}_{q}(S,\hat \G)$ to obtain
\begin{equation}\label{eq:8}
  \|\hat \beta^{\lambda} - \beta^{0} \|_{q} \le \frac{| S|^{1/q} \| \hat \G (\hat \beta^{\lambda} - \beta^{0}) \|_{\infty} }{\mathrm{CCIF}_{q}(S, \hat \G)}.
\end{equation}
To bound the right hand side of equation~(\ref{eq:8}),
\begin{align}\label{eq:9}
\begin{split}
  \|  \hat \G (\hat \beta^{\lambda} - \beta^{0}) \|_{\infty} &\le  \| \hat \Z - \hat \G  \beta^{0} \|_{\infty} +  \|  \hat \Z - \hat \G \hat \beta^{\lambda}  \|_{\infty} \\
&\le z^{*} + \lambda.
\end{split}
\end{align}
Combining equation~(\ref{eq:8}) and equation~(\ref{eq:9}) concludes the proof.
\end{proof}

\subsubsection{Proof of Lemma~\ref{lemma:finite-sample-bound-1}}

\begin{proof}
Using inner-product invariance of $(X^{e},Y^{e})$ under $\beta^{0}$,
\begin{align}\label{eq:44}
\begin{split}
    \| \hat \Z - \hat \G \beta^{0} \|_{\infty} =& \max_{k} \left| \frac{1}{n_{1}}( \Xone_{\bullet k})^{t}(\Yone- \Xone \beta^{0})  - \frac{1}{n_{2}} (\Xtwo_{\bullet k})^{t}(\Ytwo- \Xtwo \beta^{0}) \right| \\
\le& \max_{k} \left| \frac{1}{n_{1}}( \Xone_{\bullet k})^{t}(\Yone- \Xone \beta^{0}) - \mathbb{E} [ (X^{1})^{t} (Y^{1} - X^{1} \beta^{0}) ]\right| + \\
& \max_{k} \left| \frac{1}{n_{2}} (\Xtwo_{\bullet k})^{t}(\Ytwo- \Xtwo \beta^{0}) - \mathbb{E} [ (X^{2})^{t} (Y^{2} - X^{2} \beta^{0}) ] \right|
\end{split}
\end{align}
Now we can use that $\XX^{e}_{i k} (\YY^{e}-\XX^{e}_{i \bullet} \beta^{0})$, $i=1,\ldots,n_{e}$ are i.i.d. with distribution $X_{k}^{e} \eta_{p+1}^{e}$, $e \in \{1,2\}$. By \cite{van2009conditions}, for all $t \ge 0$, with probability exceeding $ 1- 2 \exp(-t)$,
\begin{align*}
&  \frac{1}{n_{e}} \left| (\XX^{e}_{\bullet k})^{t} (\YY^{e}-\XX^{e} \beta) - \mathbb{E} \left[(\XX^{e}_{\bullet k})^{t} (\YY^{e}-\XX^{e}_{i \bullet} \beta) \right] \right| \\
&\le  \sigma_{\varepsilon}  \sqrt{\Var(X_{k}^{e})}  \left( \sqrt{\frac{4t}{n_{e}}} + \frac{4t}{n_{e}} \right).
\end{align*}
Taking a union bound over $k=1,\ldots,p$, for all $t \ge 0$, with probability exceeding $ 1-2 \exp(-t)$,
\begin{align*}
\begin{split}
 &    \max_{k}  \frac{1}{n_{e}} \left| (\XX^{e}_{\bullet k})^{t} (\YY^{e}-\XX^{e} \beta) - \mathbb{E} \left[(\XX^{e}_{\bullet k})^{t} (\YY^{e}-\XX^{e}_{i \bullet} \beta) \right] \right|  \\
&\le \sigma_{\varepsilon}  \sigma_{\text{max}}^{e} \left( \sqrt{\frac{4t+4 \log(p)}{n_{e}}} + \frac{4t + 4 \log(p)}{n_{e}} \right).
\end{split}
\end{align*}
Using the bound for $e=1$ and $e=2$ and equation~\eqref{eq:44} yields the desired result.

\end{proof}

\subsubsection{Proof of Theorem~\ref{theorem:finite-sample-bound-2}}

\begin{proof}
As $  \sigma_{\varepsilon} \sigma_{\text{max}}^{e} \le C$ and as  $  \sqrt{\log(p)/ \min_{e \in  \{1,2\}}n_{e}} \rightarrow 0$ for $n_{1},n_{2},p \rightarrow \infty$, for $t = 0.2 \log p$ we have eventually
\begin{align*}
\begin{split}
&  \sigma_{\varepsilon} \sum_{e \in \{1,2\}}  \sigma_{\text{max}}^{e}\left( \sqrt{\frac{4t +  4\log(p)}{n_{e}}} +   \frac{4t + 4 \log(p)}{n_{e}} \right) \\
 \le \;  & C \sum_{e \in \{1,2\}} \left( \sqrt{\frac{4t +  4\log(p)}{n_{e}}} +   \frac{4t + 4 \log(p)}{n_{e}} \right) \\
\le \; &   2.1 C   \sqrt{\frac{4t +  4\log(p)}{\min_{e \in \{1,2\}}n_{e}}} \\
\le \; & 2.1C \cdot 2.2 \sqrt{\frac{\log(p)}{\min_{e \in \{1,2\}}n_{e}}} \\
\le \; & 4.7 C \sqrt{\frac{\log(p)}{\min_{e \in \{1,2\}}n_{e}}}.
\end{split}
\end{align*}
As $ \lambda \asymp  5 C \sqrt{\log(p)/ \min_{e \in  \{1,2\}}n_{e}}$, eventually
\begin{equation*}
 \sigma_{\varepsilon} \sum_{e \in \{1,2\}}  \sigma_{\text{max}}^{e}\left( \sqrt{\frac{4t +  4\log(p)}{n_{e}}} +   \frac{4t + 4 \log(p)}{n_{e}} \right)  \le  \lambda.
\end{equation*}
Using Lemma 3 for $t=0.2 \log(p)$, the probability of the event $z^{*} \le \lambda$ eventually exceeds $1-4 \exp(-0.2 \log(p))$, which converges to $1$ for $p \rightarrow \infty$.
By Lemma 2, on the event $z^{*} \le \lambda$,
\begin{equation*}
    \| \hat \beta^{\lambda} - \beta^{0} \|_{q} \le  \frac{|S|^{1/q} (\lambda + z^{*})}{\mathrm{CCIF}_{q}(S,\hat \G)} \le  \frac{2|S|^{1/q} }{\mathrm{CCIF}_{q}(S,\hat \G)}  5C \sqrt{\frac{  \log(p)}{ \min_{e \in  \{1,2\}}n_{e} }}.
\end{equation*}
This concludes the proof.
\end{proof}

\subsubsection{Proof of Proposition~\ref{prop:screening}}

\begin{proof}
Using Theorem~\ref{theorem:finite-sample-bound-2}  for $q=\infty$,
\begin{equation*}
 \| \hat \beta^{\lambda} - \beta^{0} \|_{\infty}   \le \frac{10 C}{\mathrm{CCIF}_{\infty}(S,\hat \G)}  \sqrt{\frac{  \log(p)}{ \min_{e \in  \{1,2\}}n_{e} }}  \text{ with $\mathbb{P} \rightarrow 1$ for }n_{1},n_{2},p \rightarrow \infty.
\end{equation*}
Using the betamin-condition,
\begin{align*}
  0& <  \min_{k \in S} | \beta_{k}^{0} |  -  \frac{10 C}{\mathrm{CCIF}_{\infty}(S,\hat \G)}  \sqrt{\frac{  \log(p)}{ \min_{e \in  \{1,2\}}n_{e} }}    \\
&\le   \min_{k \in S} | \hat \beta_{k}^{\lambda}|
\end{align*}
with $\mathbb{P} \rightarrow 1$ for $n_{1},n_{2},p \rightarrow \infty$.  Hence $ \min_{k \in S} | \hat \beta_{k}^{\lambda}| >0$  with $\mathbb{P} \rightarrow 1$ for $n_{1},n_{2},p \rightarrow \infty$. This concludes the proof.
\end{proof}

\subsubsection{Proof of Lemma~\ref{lemma:CCIF-bound}}

\begin{proof}
Consider an $u$ with $\| u_{S^{c}} \|_{1} \le \| u_{S} \|_{1}$. Hence,  $\| u \|_{1} = \| u_{S^{c}} \|_{1} + \| u_{S} \|_{1} \le 2 \| u_{S} \|_{1}$. Using this,
  \begin{align*}
\left| \frac{|S|^{1/q}\| \hat \G u \|_{\infty}}{\| u \|_{q}}  - \frac{|S|^{1/q}\| \G u \|_{\infty}}{\| u \|_{q}} \right|  &\le \frac{|S|^{1/q}\| ( \hat \G - \G) u \|_{\infty}}{\| u \|_{q}} \\
&\le  \frac{|S|^{1/q}\|\hat  \G - \G \|_{\infty}  \| u \|_{1}}{\| u \|_{q}} \\
&\le  \frac{|S|^{1/q}\| \hat \G - \G \|_{\infty}  2\| u_{S} \|_{1}}{\| u \|_{q}} \\
&\le  \frac{|S|^{1/q}\| \hat \G - \G \|_{\infty}  2\| u_{S} \|_{1}}{\| u_{S} \|_{q}} \\
&\le 2 |S | \| \hat \G - \G \|_{\infty}.
  \end{align*}
In the last line we used that $q \ge 1$. This concludes the proof.
\end{proof}

\subsubsection{Causal Dantzig as a LP}\label{sec:causal-dantzig-as}
For fixed $\lambda$, the regularized \emph{causal Dantzig} can be cast as a linear program. For notational simplicity, will show this for the case  $| \E | =2$. Define
\begin{equation*}
  A := \begin{pmatrix}
            - \hat \G \quad \hat \G \\
            \hat \G \quad - \hat \G
  \end{pmatrix}, \quad b := \begin{pmatrix}
            -\hat \Z \\
            \hat \Z
  \end{pmatrix} + \begin{pmatrix}
    \lambda \\
\vdots \\
\lambda
  \end{pmatrix}, \quad \text{ and } c := \begin{pmatrix}
    1 \\
\vdots \\
1
  \end{pmatrix}.
\end{equation*}
Let $\Gamma^{\lambda}$ be the  solution set  of the linear program
\begin{align*}
 \text{minimize } &c^{t} \gamma \\
 \text{ subject to } &A\gamma \le b \text{ and } \gamma \ge 0.
 \end{align*}
Let $B^{\lambda}$ be the solution set of \eqref{eq:7}. The following Lemma shows that $B^{\lambda}$ can easily be obtained from $\Gamma^{\lambda}$.
\begin{lemma}
$ B^{\lambda} =  \{  \gamma_{1:p}- \gamma_{(p+1):2p} : \gamma \in \Gamma^{\lambda} \}$
\end{lemma}
\begin{proof}
Let $\gamma \in \Gamma^{\lambda}$. By constraint, all entries of $\gamma$ are non-negative. Furthermore, $\gamma_{k}$ and $\gamma_{p+k}$ cannot be nonzero at the same time: In that case, $\tilde \gamma$ defined as
\begin{equation*}
  \tilde \gamma_{k'} = \begin{cases}
    \gamma_{k'} & k' \neq k \text{ or } k' \neq p+k, \\
    \gamma_{k}-\min(\gamma_{k},\gamma_{k+p}) & k'=k, \\
    \gamma_{k+p} - \min(\gamma_{k},\gamma_{k+p}) & k'=k+p.
  \end{cases}
\end{equation*}
would suffices $A \tilde \gamma = A  \gamma \le b$, $\tilde \gamma \ge 0$, $c^{t} \tilde \gamma < c^{t} \gamma$, which is a contradiction to the definition of $\gamma$. As either  $\gamma_{k}$ or $\gamma_{p+k}$ are equal to zero, $c^{t} \gamma = \| \gamma_{1:p} - \gamma_{(p+1):2p} \|_{1}$. Analogously, one can show that any solution $\gamma$ to
\begin{align*}
   \text{minimize } &  \| \gamma_{1:p} - \gamma_{(p+1):2p}  \|_{1} \\
 \text{ subject to } & \begin{pmatrix}- \hat \G \\ \hat \G  \end{pmatrix}
( \gamma_{1:p} - \gamma_{(p+1):2p}) \le b \text{ and } \gamma \ge 0
\end{align*}
satisfies that either $\gamma_{i} = 0$ or $\gamma_{i+p}=0$. Hence $\Gamma^{\lambda}$ is also the solution set of
\begin{align*}
   \text{minimize } &  \| \gamma_{1:p} - \gamma_{(p+1):2p}  \|_{1} \\
 \text{ subject to } & \begin{pmatrix}- \hat \G \\  \hat \G  \end{pmatrix}
( \gamma_{1:p} - \gamma_{(p+1):2p}) \le b \text{ and } \gamma \ge 0.
\end{align*}
By rewriting the constraint, this problem is equivalent to solving
\begin{align}\label{eq:34}
\begin{split}
   \text{minimize } &  \| \gamma_{1:p} - \gamma_{(p+1):2p}  \|_{1} \\
 \text{ subject to } & \| \hat \Z - \hat \G ( \gamma_{1:p} - \gamma_{(p+1):2p})  \|_{\infty} \le \lambda \text{ and } \gamma \ge 0.
\end{split}
\end{align}
Now for each solution $\gamma$ of this problem we can define $\beta(\gamma) := \gamma_{1:p}- \gamma_{(p+1):2p}$ and $\beta(\gamma)$ satisfies the constraint $ \| \hat  \Z -\hat \G \beta  \|_{\infty} \le \lambda$. Furthermore the objective functionals match, i.e.  $\| \gamma_{1:p} - \gamma_{(p+1):2p}  \|_{1} = \| \beta(\gamma) \|_{1}$.  On the other hand, for each solution $\beta$ of
\begin{align}\label{eq:33}
\begin{split}
   \text{minimize } &  \| \beta  \|_{1} \\
 \text{ subject to } & \| \hat \Z - \hat \G \beta  \|_{\infty} \le \lambda.
\end{split}
\end{align}
we can define $\gamma(\beta) \in \mathbb{R}^{2p}$ via $\gamma(\beta)_{1:p} =  \max(\beta, 0_{p})$ and $\gamma(\beta)_{(p+1):2p} = - \min(\beta,0_{p})$ . Note that by definition $\gamma(\beta)$ satisfies the constraints $\gamma(\beta) \ge 0$, $\| \hat \Z - \hat \G ( \gamma_{1:p} - \gamma_{(p+1):2p})  \|_{\infty} \le \lambda$ and again the objective functionals match, i.e. $\| \beta \|_{1} = \| \gamma(\beta)_{1:p} - \gamma(\beta)_{(p+1):2p} \|_{1}$.  Hence $ B^{\lambda} =  \{  \gamma_{1:p}- \gamma_{(p+1):2p} : \gamma \in \Gamma^{\lambda} \}$. This concludes the proof.
\end{proof}

\subsection{Proof for Section~\ref{sec:pract-cons}}

\subsubsection{Proof of Lemma~\ref{lemma:hidden-markov-blanket}}\label{sec:mark-blank-estim}

\begin{proof}
Proof by contradiction. Let $X_{k}$ be a parent or child of $Y$ in $D_{total}$ with $k \not \in S$. Without loss of generality let us assume that $Y \rightarrow X_{k}$. As the regression coefficient of $X_{k}$ is zero and as $X_{1},\ldots,X_{p},Y$ are multivariate Gaussian, $ Y $ is conditionally independent of $ X_{k} $ given $ X_{S}$. As the distribution of  $X_{1},...,X_{p},Y,H_{1},...,H_{q}$ is faithful to $D_{total}$, $Y$ and $X_{k}$ are $d$-separated by $X_{S}$ in $D_{total}$, see e.g. \cite{Pearl2009} for a reference. Hence the path $Y \rightarrow X_{k}$ is blocked by $X_{S}$. But the path $Y \rightarrow X_{k}$ can only be blocked if $k \in S$. Contradiction. This concludes the proof.
\end{proof}

\subsection{Asymptotic efficiency}\label{sec:asympt-effic}

Assume that  for $e \in \{1,2\}$ the variables $(X^{e},Y^{e})$ are centered (non-degenerate) Gaussian  random variables that are generated from a structural equation model under Assumption~\ref{assum:additive}. Intuitively, as  the Gram matrices are asymptotically efficient estimators of $\mathbb{E} (X^{e})^{t}  X^{e} $ and $  \mathbb{E} (X^{e})^{t} Y^{e}$ one would expect the plug-in estimator $\hat \beta= \hat \G^{-1} \hat \Z$ to be efficient, too. That is still true in some sense, but we have to be a bit careful with the notion of efficiency. There are two issues that we have to take care of. First, we have the additional constraint that the data is generated by a specific SEM that satisfies inner-product invariance under the true causal coefficient $\beta^{0}$. Can this constraint be exploited to lower asymptotic variance? Additionally, we have to deal with the fact that $n_{1}$ and $n_{2} $ may have different asymptotic growth rates.
The following Lemma gives an answer to the first question if we allow for errors-in-variables as defined in equation~\eqref{eq:21}.
\begin{lemma}\label{lemma:construct}
Consider distributions $(\mathring X^{1}, \mathring Y^{1}) \sim \mathcal{N}(0, \mathring \SI^{1})$ and $(\mathring X^{2},\mathring Y^{2}) \sim \mathcal{N}(0,\mathring\SI^{2})$ with inner-product invariance under $ \beta^{0}$ that satisfy Assumption~\ref{assum:additive} and have errors-in-variables as defined in equation~\eqref{eq:21}. For any distribution $(\tilde X^{1},\tilde Y^{1}) \sim \mathcal{N}(0, \tilde \SI^{1})$ with $\tilde \SI^{1}$ sufficiently close to $\mathring \SI^{1}$ and $(\tilde X^{2},\tilde Y^{2}) \sim \mathcal{N}(0,\tilde \SI^{2})$ with $\tilde \SI^{2}$ sufficiently close to  $ \mathring \SI^{2}$, there exists an linear structural equation model with error-in-variables that satisfies Assumption~\ref{assum:additive} and equation~\eqref{eq:21}.
\end{lemma}
This Lemma shows that the fact that our model is generated by a Gaussian linear SEM with additive interventions and errors-in-variables that satisfies inner-product invariance does not restrict the distributions in a neighborhood of other models that satisfy these properties.
Now let us turn to the question what statements can be made about the limit $n_{1} \rightarrow \infty$, $n_{2} \rightarrow \infty$. It is straightforward to model this the following way: for each sample $i=1,\ldots,n$, first a coin is tossed. With probability $0< \pi<1$ we observe a sample from setting $e_{i}=1$ and with probability $1- \pi$ we observe a sample of setting $e_{i}=2$. To be more precise, the corresponding log density can be written as
\begin{align*}
  &\sum_{i=1}^{n} 1_{e_{i} = 1} \log f_{\SI^{1}}(\XX_{i \bullet},\YY_{i}) + 1_{e_{i} = 2} \log f_{\SI^{2}}(\XX_{i \bullet},\YY_{i}) \\
&+  1_{e_{i} = 1} \log(\pi) + 1_{e_{i} = 2} \log(1-\pi),
\end{align*}
where $f_{\SI}$ denotes the density of a centered Gaussian distribution with covariance $\SI \in \mathbb{R}^{(p+1) \times (p+1)}$. Hence, $(\XX^{1},\YY^{1})$ is a sufficient statistics for $\SI^{1}$ and $(\XX^{2}, \YY^{2})$ is a sufficient statistics for $\SI^{2}$. By \citet{anderson1973asymptotically}, the Gram matrix of $(\XX^{1},\YY^{1})$ is asymptotically efficient for estimating $\SI^{1}$ and the Gram matrix of $(\XX^{2},\YY^{2})$ is asymptotically efficient for estimating $\SI^{2}$. The Fisher information matrix is block diagonal with blocks for $\SI^{1}$, $\SI^{2}$ and $\pi$. Thus, the Gram matrices of $(\XX^{1},\YY^{1})$ and $(\XX^{2},\YY^{2})$ are asymptotically efficient for jointly estimating $\SI^{1}$ and $\SI^{2}$. By the delta method, the plug-in estimator $\beta = \hat \G^{-1} \hat \Z$ is asymptotically efficient for estimating $\beta^{0}$.
Note that in the discussion above we have $n_{1} \sim \pi \cdot n$ and $n_{2} \sim (1-\pi) \cdot n$. Hence this is a ``balanced'' scenario and this type of analysis does not work for, say,  $n_{1} = \scriptstyle \mathcal{O} \textstyle (n_{2})$. In the latter case, the asymptotic variance of estimating the Gram matrix in setting $e=1$ is dominating the asymptotic variance of estimating the Gram matrix in setting $e=2$. Hence it can be shown that $\hat \beta$ has the same asymptotic variance as an efficient estimator for $\beta^{0}$ assuming the Gram matrix in setting $e=2$ is known.

\subsubsection{Proof of Lemma~\ref{lemma:construct}}

\begin{proof}
 Choose $\tilde \beta^{0}$ such that $ (\tilde \SI^{1} - \tilde \SI^{2})_{1:p,1:p} \tilde \beta^{0} =  (\tilde \SI^{1} - \tilde \SI^{2})_{1:p,p+1} $. By construction of $\tilde \beta^{0}$ the random variables satisfy
\begin{align*}
  \mathbb{E} \left[ \tilde X_{k}^{1} ( \tilde Y^{1} - \tilde X^{1} \beta^{0} )    \right] &= \tilde \SI^{1}_{k,p+1} -\tilde  \SI_{k ,1:p}^{1} \beta^{0} \\
&= \tilde \SI^{2}_{k,p+1} - \tilde \SI_{k ,1:p}^{2} \beta^{0} \\
&=  \mathbb{E} \left[ \tilde X_{k}^{2} (\tilde  Y^{2} - \tilde X^{2} \beta^{0} )    \right] \text{ for all } k=1,\ldots,p.
\end{align*}
In other words, we have \emph{inner product invariance} under $\tilde \beta^{0}$.  Now we want to show that the distribution of $(\tilde X^{e}, \tilde Y^{e}), e \in \E$ can be generated by a structural equation model of the following form. We want to show that there exist independent random variables $\eta^{0},\delta^{e} \in \mathbb{R}^{p+1} \zeta_{1}^{e},\ldots, \zeta_{p}^{e}, \zeta_{y}^{e} , e \in \{1,2\}$ with $\eta^{e} = \eta^{0} + \delta^{e}$ such that $\eta^{0},\delta^{e},\eta^{e}, e \in \{1,2\}$ satisfy  Assumption~\ref{assum:additive} and such that $\zeta_{1}^{e},\ldots, \zeta_{p}^{e}, \zeta_{y}^{e}, e \in \{1,2\}$ satisfy the assumption mentioned after equation~\eqref{eq:21}. Furthermore, with slight abuse of notation we want that the following structural equation model with error-in-variables
\begin{align}\label{eq:36}
\begin{split}
   X_{k}^{e} &\leftarrow \eta_{k}^{e} + \zeta_{k}^{e}, \text{ for }k =1 ,\ldots,p,\\
   Y^{e} &\leftarrow \sum_{k=1}^{p}  X_{k} \tilde \beta_{k}^{0} + \eta_{p+1}^{e} + \zeta_{y}^{e},
\end{split}
\end{align}
generates the distribution of $(\tilde X^{e},\tilde Y^{e})$, i.e.  satisfies $(X^{e},Y^{e}) \sim \mathcal{N}(0, \tilde \SI^{e}), e \in \{1,2\}$. As $\tilde X_{1}^{e},\ldots,\tilde X_{p}^{e}, \tilde Y^{e}, e \in \{1,2\}$ are centered multivariate Gaussian it suffices to show that the covariance matrix of $(\tilde X_{1}^{e},\ldots,\tilde X_{p}^{e},\tilde Y^{e} - \tilde X^{e} \tilde \beta^{0}), e \in \{1,2\}$  can be decomposed into
\begin{align}\label{eq:37}
  \Sigma_{\eta} + \Sigma_{\delta}^{e}+\Sigma_{\zeta}^{e},
\end{align}
with positive semi-definite matrices $\Sigma_{\eta}, \Sigma_{\delta}^{e}, \Sigma_{\zeta}^{e}$ satisfying
\begin{enumerate}
  \item  $(\Sigma_{\delta}^{e})_{p+1,\bullet} \equiv 0$,
  \item  $\Sigma_{\zeta}^{e}, e \in \{1,2\}$ are diagonal matrices with $(\Sigma_{\zeta}^{1})_{k,k} = (\Sigma_{\zeta}^{2})_{k,k}$ for $k=1,\ldots,p$.
\end{enumerate}
To this end, define $r^{1} = \tilde Y^{1} - \sum_{k=1}^{p} \tilde X^{1}_{k} \beta_{k}^{0}$ and $r^{2} = \tilde Y^{2} - \sum_{k=1}^{p} \tilde X_{k}^{2}\tilde  \beta_{k}^{0}$.  Define the matrices
\begin{align*}
\begin{split}
  S^{1}&= \begin{pmatrix}
  \text{Cov}(\tilde X_{1:p}^{1}) & \text{Cov}(\tilde X_{1:p}^{1}, r^{1})\\
  \text{Cov}(\tilde X_{1:p}^{1},r^{1})^{t} & \text{Var}(\mathring \eta_{p+1}^{1}) - \epsilon \end{pmatrix} \text{ and } \\
  S^{2} &= \begin{pmatrix}
  \text{Cov}(\tilde X_{1:p}^{2}) & \text{Cov}(\tilde X_{1:p}^{2}, r^{2})\\
  \text{Cov}(\tilde X_{1:p}^{2},r^{2})^{t} & \text{Var}(\mathring \eta_{p+1}^{2}) - \epsilon      \end{pmatrix}.
\end{split}
 \end{align*}
Here, $\mathring \eta_{p+1}^{e}$ denotes the noise contribution of $\mathring X_{p+1}^{e}= \mathring Y^{e}$ in the corresponding structural equation model. For $\tilde \SI^{e} \rightarrow \mathring{ \SI}^{e}, e \in \{1,2\}$, $\text{Cov}(\tilde X_{1:p}^{e})$ converges to $\text{Cov}(\mathring X_{1:p}^{e})$ and $\text{Cov}(\tilde X_{1:p}^{e},r^{e}) $ converges to $  \text{Cov}(\mathring X_{1:p}^{e}, \mathring Y^{e}- \mathring X^{e} \beta^{0}) =  \text{Cov}(\mathring X_{1:p}^{e}, \mathring \eta_{p+1}^{e})$.
Recall that the covariance matrices of $(\mathring X_{1}^{e},\ldots,\mathring X_{p}^{e},  \mathring \eta_{p+1}^{e}), e \in \{1,2\}$ are positive definite. Hence $S^{1}$ and $S^{2}$ are positive definite   for $\tilde \SI^{1}$  close to $\mathring \SI^{1}$, $\tilde \SI^{2}$ close to $\mathring \SI^{2}$ and  $\epsilon > 0$ small enough. Now we can define
\begin{align*}
\begin{split}
  (\Sigma_{\zeta}^{1})_{k,k'} &:= \begin{cases}
    \text{Var}(r^{1}) - \text{Var}(\mathring \eta_{p+1}^{1})) + \epsilon  & k=k'=p+1, \\
    0 & \text{else}.
  \end{cases} \\
  (\Sigma_{\zeta}^{2})_{k,k'} &:= \begin{cases}
    \text{Var}(r^{2}) - \text{Var}(\mathring \eta_{p+1}^{2})) + \epsilon & k=k'=p+1, \\
    0 & \text{else}.
  \end{cases}
\end{split}
\end{align*}
With this definition the covariance matrix of $(\tilde X_{1}^{e},\ldots,\tilde X_{p}^{e},\tilde Y^{e} - \tilde X^{e} \tilde \beta^{0}), e \in \{1,2\}$ can be decomposed as $ S^{e}+\Sigma_{\zeta}^{e}, e \in \{1,2\}$. \\
For $\tilde \SI^{1}$  close to $\mathring \SI^{1}$ and $\tilde \SI^{2}$ close to $\mathring \SI^{2}$, the matrices $ \Sigma_{\zeta}^{e}, e \in \{1,2\} $,  are positive semi-definite as $r^{e}$ has asymptotic variance $\text{Var}(\mathring \eta_{p+1}^{e} +  \mathring\zeta_{y}^{e})$, where $\mathring \zeta_{y}^{e}$ denotes the measurement error of $\mathring Y$ in environment $e$ in the corresponding structural equation model.  Thus by equation~\eqref{eq:37} it suffices to show that $S^{e}, e \in \{1,2\}$ can be decomposed into positive semi-definite matrices $\Sigma_{\eta}+ \Sigma_{\delta}^{e}$ such that $(\Sigma_{\delta}^{e})_{p+1,\bullet} \equiv 0$. \\
To this end let us define
\begin{equation*}
x= \begin{pmatrix}
  \frac{\text{Cov}(\tilde X_{1:p}^{1}, r^{1})}{\sqrt{ \text{Var}(\mathring \eta_{p+1}^{1}) - \epsilon}} \\  \vspace{0.05cm} \\\sqrt{ \text{Var}(\mathring \eta_{p+1}^{1}) - \epsilon}
\end{pmatrix}\in \mathbb{R}^{p+1}.
\end{equation*}
Now we want to show that
\begin{equation*}
\Sigma_{\delta}^{e} := \begin{pmatrix} S_{1:p,1:p}^{e} - x_{1:p} x_{1:p}^{t} & 0_{p} \\ 0_{p} & 0 \end{pmatrix}
\end{equation*}
 are positive semi-definite for $e \in \{1,2\}$. To this end take $v \in \mathbb{R}^{p}$. Then for $e \in \{1,2\}$,
\begin{align*}
  v^{t} (\Sigma_{\delta}^{e})_{1:p,1:p} v &= v^{t} S_{1:p,1:p}^{e}  v - v^{t} x_{1:p} x_{1:p}^{t} v \\
&= \left(v^{t}, - \frac{v^{t} x_{1:p}}{\sqrt{ \text{Var}(\mathring \eta_{p+1}^{1}) - \epsilon}}\right)  S^{e} \begin{pmatrix}
     v \\
    - \frac{v^{t} x_{1:p}}{\sqrt{ \text{Var}(\mathring \eta_{p+1}^{1}) - \epsilon}}
\end{pmatrix} \ge 0
\end{align*}
Note that we used that $S^{e}, e \in \{1,2\}$ are positive definite, that by Assumption~\ref{assum:additive}  $ \text{Var}(\mathring \eta_{p+1}^{1}) =  \text{Var}(\mathring \eta_{p+1}^{2})$, and that by inner-product invariance,  $\text{Cov}(\tilde X_{1:p}^{1},r^{1}) = \text{Cov}(\tilde X_{1:p}^{2},r^{2})$. Now by defining $\Sigma_{\eta} := x x^{t}$ we obtain the  decomposition $S^{e} = \Sigma_{\eta} + \Sigma_{\delta}^{e}$. Note that here we used again that by inner-product invariance under $\beta^{0}$, $\text{Cov}(\tilde X_{1:p}^{1},r^{1}) = \text{Cov}(\tilde X_{1:p}^{2},r^{2})$. This completes the proof.

\end{proof}

\bibliographystyle{plainnat}
\bibliography{bibliography-updated}

\end{document}